\PassOptionsToPackage{prologue,dvipsnames}{xcolor}
\documentclass[11pt]{article}
\usepackage{amsmath,amssymb}

\usepackage{times}
\usepackage{bm}
\usepackage[authoryear]{natbib}
\usepackage{algorithm}
\usepackage{amssymb}
\usepackage{mathrsfs}
\usepackage{graphicx}
\usepackage[flushleft]{threeparttable}
\usepackage{longtable}
\usepackage{booktabs}
\usepackage{enumitem} %%%enumitem
\usepackage{star}

\usepackage[colorlinks,
linkcolor=blue,
anchorcolor=blue,
citecolor=blue
]{hyperref}

\newcommand {\train}{\textnormal {(tr)}}

\newcommand {\valid}{\textnormal {(v)}}

\newcommand {\RA}{\textnormal {R}}

\def\T{{ \mathrm{\scriptscriptstyle T} }}

\def\##1\#{\begin{align}#1\end{align}}
\def\$#1\${\begin{align*}#1\end{align*}}

\renewcommand{\tr}{\textrm{trace}}

\def\T{{ \mathrm{\scriptscriptstyle T} }} %%%transpose operator

\newcommand{\Rom}[1]{\text{\uppercase\expandafter{\romannumeral #1\relax}}}

\usepackage{txfonts}

\usepackage{geometry}

\usepackage{tocloft}% http://ctan.org/pkg/tocloft
\makeatletter
\renewcommand{\numberline}[1]{%
  \@cftbsnum #1\@cftasnum~\@cftasnumb%
}
\makeatother

\begin{document}

\title{ \LARGE On self-training of summary data with genetic applications}     % Option 1

%Developing Genetic Prediction Models through Self-Training on Summary Data Autotraining

\author{Buxin Su\thanks{Department of Mathematics, University of Pennsylvania; Email: \texttt{subuxin@sas.upenn.edu}.} 
\and  Jiaoyang Huang\thanks{Department of Statistics and Data Science, University of Pennsylvania; Email: \texttt{huangjy@wharton.upenn.edu}.}
\and  Jin Jin\thanks{Department of Biostatistics, Epidemiology and Informatics,  Perelman School of Medicine, University of Pennsylvania; Email: \texttt{jin.jin@pennmedicine.upenn.edu}.}
\and  Bingxin Zhao\thanks{Department of Statistics and Data Science, University of Pennsylvania; Email: \texttt{bxzhao@wharton.upenn.edu}.}\\~\\
}

%\date{November 1, 2024}

\maketitle

\vspace{-0.25in}

% typeset the title of the contribution
\begin{abstract}
Prediction model training is often hindered by limited access to individual-level data due to privacy concerns and logistical challenges, particularly in biomedical research. 
Resampling-based self-training presents a promising approach for building prediction models using only summary-level data. 
These methods leverage summary statistics to sample pseudo datasets for model training and parameter optimization, allowing for model development without individual-level data. 
Although increasingly used in precision medicine, the general behaviors of self-training remain unexplored.
In this paper, we leverage a random matrix theory framework to establish the statistical properties of self-training algorithms for high-dimensional sparsity-free summary data. 
Notably, we demonstrate that, within a class of linear estimators, resampling-based
self-training achieves the same asymptotic predictive accuracy as conventional training methods that require individual-level datasets.
These results suggest that self-training with only summary data incurs no additional cost in prediction accuracy, while offering significant practical convenience. 
Our analysis provides several valuable insights and counterintuitive findings. For example, while pseudo-training and validation datasets are inherently dependent, their interdependence unexpectedly cancels out when calculating prediction accuracy measures, effectively preventing overfitting in self-training algorithms.
Furthermore, we extend our analysis to show that the self-training framework maintains this no-cost advantage when combining multiple methods (e.g., in ensemble learning) or when jointly training on data from different distributions (e.g., in multi-ancestry genetic data training). We numerically validate our findings through extensive simulations and real data analyses using the UK Biobank. 
Our study highlights the potential of resampling-based self-training to advance genetic risk prediction and other fields that make summary data publicly available.

\end{abstract}
\noindent
{\bf Keywords}:  Genetic prediction; Prediction models; Precision medicine; Random matrix theory; Resampling; Summary statistics.

% \tableofcontents

\section{Introduction}\label{sec:1}

Developing prediction models—one of the major tasks in statistical learning and various scientific fields—typically relies on access to individual-level data for model training and validation. For example, in a traditional prediction model training process, individual-level datasets are required and are often split into independent training and validation subsets (or used in a cross-validation design) for parameter optimization, to avoid overfitting and ensure generalizability.
However, in many applications, accessing individual-level data is challenging. A notable example is genetic risk prediction in precision medicine \citep{lennon2024selection}, which uses genetic variants from genome-wide association studies (GWAS) \citep{uffelmann2021genome} as predictors. 
Due to privacy restrictions and logistical challenges, summary-level GWAS data (such as marginal genetic effect estimates and standard errors) rather than original individual-level genetic profiles, have become the standard for data sharing in genetic research \citep{pasaniuc2017dissecting}. Using these summary statistics, polygenic risk scores (PRS) have been widely developed to assess genetic risk for various complex traits and diseases \citep{purcell2009common}. Millions of genetic variants in GWAS are used as predictors, each contributing only a small amount of information \citep{boyle2017expanded}.
Over the last two decades, a wide variety of statistical methods have been developed to improve PRS performance by better aggregating predictive power across high-dimensional genetic predictors \citep{vilhjalmsson2015modeling,mak2017polygenic,ma2021genetic,ge2019polygenic,pattee2020penalized,hu2017leveraging,yang2020accurate}. 
Although most PRS prediction models only require summary-level data from the training sample as model input, they typically still need access to 
%a small set of 
individual-level validation data during the training process for model parameter tuning \citep{choi2020tutorial,mak2017polygenic,vilhjalmsson2015modeling,marquez2021incorporating}. However, this individual-level validation dataset may not always be available to researchers who need to develop the prediction model, as sharing and accessing genetic datasets—even small validation data—can pose risks and raise concerns regarding data privacy and policy. Such privacy and logistical barriers have significantly limited the accessibility and scalability of PRS applications \citep{bonomi2020privacy,wan2022sociotechnical,jin2025pennprs}.

Recent advances in genetic fields have facilitated pseudo-training for PRS models directly using high-dimensional summary-level data \citep{zhao2024optimizing,zhao2021pumas,jiang2024tuning,chen2024fast,song2019summaryauc}. The principle of this approach is that, given the summary-level data, it is possible to sample pseudo-training and validation summary statistics from their underlying probability distribution. 
These sampled summary statistics closely mimic what would be obtained if there were access to two independent subsets of the individual-level data. 
Therefore, using these pseudo-training and validation datasets, it is possible to tune parameters and ultimately derive a prediction model, allowing for the self-training of summary statistics. 
Despite their strong preliminary numerical performance in several application examples \citep{jiang2024tuning,jin2025pennprs} and emerging extensions to other biomedical fields  \citep{wu2023large,wang2024integrating}, little statistical research has been conducted to understand the general properties of resampling-based self-training for summary-level data. 
It remains unclear whether these methods lead to potential reductions in model performance, especially regarding the trade-offs between practical convenience and decreased prediction accuracy, as well as the key factors influencing their performance in practical applications. Therefore, it is crucial to establish a framework to quantify the performance of self-training and compare it to conventional model training with individual-level data.

In this paper, we propose a general random matrix theory framework \citep{bai2010spectral,Yao_Zheng_Bai_2015,dobriban2021distributed} to model and understand self-training with high-dimensional summary-level data. Our study provides several key contributions. 
First, we establish the statistical properties of self-training algorithms for high-dimensional predictors with a general covariance structure, without imposing sparsity constraints on regression coefficients.
We demonstrate that, under flexible conditions, a class of linear estimators achieves the same asymptotic predictive accuracy as their counterparts trained using individual-level data. 
We provide detailed analytical evaluations of ridge-type estimators and marginal thresholding estimators, both of which, along with their variants, are widely used in PRS applications \citep{ma2021genetic}.
These analyses provide deep insights for practical applications and reveal notable counterintuitive findings. 
For example, unlike individual-level training and validation datasets, pseudo-training and validation datasets are inherently dependent, as they are derived from the same summary-level data. However, surprisingly, we find that this interdependence cancels out during the calculation of prediction accuracy measures, such as $R$-squared ($R^2$) or mean squared error. This phenomenon effectively prevents overfitting in self-training algorithms.
%\bxs{Do we need to be consistent on self-training or self-training?} algorithms.
Furthermore, we show that self-training can facilitate the combination of different estimators using summary statistics within an ensemble learning framework. Self-training also enables joint training across multiple datasets from different distributions, which is particularly relevant to the emerging field of multi-ancestry genetic risk prediction \citep{kachuri2024principles,zhang2023new,jin2024mussel}. 
We numerically validate our theoretical findings with extensive simulations and real data analyses from the UK Biobank \citep{bycroft2018uk}. 

%\paragraph{Paper Overview}
The rest of the paper proceeds as follows. 
In Section~\ref{sec:model}, we introduce the model setup and the framework for modeling self-training of summary data. 
Section~\ref{sec:general} presents the random matrix theory results for the class of linear estimators.
Section~\ref{sec:emsemble} extends the algorithm and analysis to ensemble learning. 
We model multi-ancestry genetic data resources in Section~\ref{sec:multi}. 
Numerical experiments are presented in Section~\ref{sec:numer}. 
In Section~\ref{sec:disc}, we discuss potential future research directions. Supplemental material collects proof of the main results and technical lemmas.
We define $\RR$ and $\RR_{+}$ as the sets of all real numbers and positive real numbers, respectively. For any positive natural number $p \in \NN_{>0}$, we define $[p]$ to be the set of $\{1, 2, \cdots ,p\}$. 
For two sequences of random variables $\{X_n\}_{n \in \NN}$ and $\{Y_n\}_{n \in \NN}$, we write $X_n = Y_n + o_p(1)$ or $X_n - Y_n \overset{p}{\to} 0$ when their difference converges in probability to $0$. 
Moreover, we use $X_n - Y_n \overset{d}{\to} 0$ to represent the convergence in distribution to $0$.
%As $n,p \to \infty$ proportionally, we say $X_n = Y_n + o_p(1)$, if $X_n - Y_n \overset{p}{\to} 0$. 
For any sequence of functions $f(\theta)$ and $g(\theta)$ with variable $\theta$, we say $f(\theta)$ is proportional to $g(\theta)$, $f(\theta) \propto g(\theta)$, if $f(\theta)/g(\theta) = c$ for any $\theta$ and some constant $c$ independent with $\theta$.
% \bxs{As $n,p \to \infty$ proportionally, we say $X_n = Y_n + o_p(1)$, if $X_n - Y_n \overset{p}{\to} 0$. For any sequence of functions $f(\theta)$ and $g(\theta)$ with variable $\theta$, we say $f(\theta)$ is proportional to $g(\theta)$, $f(\theta) \propto g(\theta)$, if $f(\theta)/g(\theta) = c$ for any $\theta$ and some constant $c$ independent with $\theta$. 
% }
%\bxz{Please double-check the notations here and conditions below, and ensure there are no obvious typos.}

%%%%%%%%%%%%%%%%%%%%%%%%%%%%%%%%%%%%%%%%%%%%%%%%
%%%%%%%%%%%%%%%%%%%%%%%%%%%%%%%%%%%%%%%%%%%%%%%%
%%%%%%%%%%%%%%%%%%%%%%%%%%%%%%%%%%%%%%%%%%%%%%%%
\section{Self-training of summary data}\label{sec:model}
\subsection{The model and data}\label{sec:model-1}
We specify the data generation model for the dataset $(\Xb, \yb)$, where $\Xb \in \RR^{n \times p}$ and $\yb \in \RR^{n}$.
The linear model relating the observation $\yb$ and design matrix $\Xb$ can be expressed as follows
\begin{align} \label{eqn:linear_model}
    \yb = \Xb \bbeta + \bepsilon,
\end{align}
where $\yb$ represents a phenotype and $\Xb \in \RR^{n \times p}$ denotes $p$ genetic variants. The genetic effects $\bbeta \in \RR^{p}$ and noise $\bepsilon \in \RR^{n}$ are random variables. 
The matrix $\Xb$ typically contains millions of single nucleotide polymorphisms (SNPs) from a large number of samples, for example, half a million participants in the UK Biobank \citep{bycroft2018uk}.
The following conditions on $\Xb$ are frequently used in the application of random matrix theory for such high-dimensional data \citep{dobriban2018high, ledoit2011eigenvectors,bai2010spectral}.
Condition \ref{cond-np-ratio} indicates a high-dimensional regime where the sample size and feature dimensionality are proportional.

\begin{condition} \label{cond-np-ratio} 
The sample size $n\to \infty$ while the dimensionality $p \to \infty$, such that the aspect ratio $p/n \to \gamma >0$. 
\end{condition}

Condition~\ref{cond-X} outlines the regularity conditions on $\Xb$ required by technical lemmas in random matrix theory, such as bounded moments and eigenvalues. 
Notably, these conditions do not assume a Gaussian distribution for the elements of $\Xb$ and have been shown to be robust to minor potential violations, such as the presence of small eigenvalues, which may arise in practical PRS applications due to the existence of highly correlated genetic variants \citep{zhao2024blockwise}. 
\begin{condition} \label{cond-X}%[Condition on $\Xb$]
    We assume $\Xb = \Xb_0 \bSigma^{1/2}$. Entries of $\Xb_0$ is real-value i.i.d. random variables with mean zero, variance one, and a finite $4$-th order moment. The $\bSigma$ are $p \times p$ population level deterministic positive definite matrices with uniformly bounded eigenvalues. Specifically, we have $0 < c \leq \lambda_{\min}(\bSigma) \leq \lambda_{\max}(\bSigma) \leq C$ for all $p$ and some constants $c, C$, where $\lambda_{\min}(\cdot)$ and $\lambda_{\max}(\cdot)$ are the smallest and largest eigenvalues of a matrix, respectively.
\end{condition}

To establish the linear model framework, we use random-effect conditions on $\bbeta$, commonly used to model a large number of small genetic effects without imposing sparsity constraints \citep{jiang2016high,su2024exact}. 
Let $F(0, V)$ denote a generic distribution with
mean zero, variance $V$, and a finite $4$-th order moment. We introduce the following conditions on genetic effects.
\begin{condition}  \label{cond-beta}%[Distribution assumption for $\bbeta$]
There exist sparsity level $0 \leq \kappa \leq 1$ and signal strength $\sigma_{\bbeta}^2>0$ such that
the distribution of each coordinate of $\bbeta$, denoted as $\bbeta_{i}$, is an i.i.d. random variable that follows
\$
\bbeta_{i} \sim (1 - \kappa) \delta(0) + \kappa F(0, \sigma_{\bbeta}^2/p).
\$
Furthermore, as $n, p \to \infty$ with $p/n \to \gamma >0$, we assume that $\sum_{j=1}^p (\bbeta_{j})^2 \to \kappa \sigma_{\bbeta}^2$. 
\end{condition}

The following condition imposes similar conditions on the noise vector $\bepsilon$.
\begin{condition} \label{cond-eps} %[Noise assumption]
Random errors in $\bepsilon$ are independent random variables and each coordinate has the following distribution
\begin{align*}
    \bepsilon_i \overset{i.i.d.}{\sim} F(0, \sigma_{\bepsilon}^2), \quad \mbox{for} \quad 1 \leq i \leq n.
\end{align*}
\end{condition}

Two types of summary-level data are typically used for training prediction models. The first relates to the estimation of genetic effects, typically provided as marginal GWAS summary statistics and can be represented as $\Xb^{\T} \yb$ \citep{pasaniuc2017dissecting}. 
In practice, additional associated summary statistics, such as the variance of genetic effect estimates, $P$-values, and minor allele frequencies, are also commonly shared along with $\Xb^{\T} \yb$.
The second is information about the linkage disequilibrium (LD) pattern, $\bSigma$, which can be quantified by $\Xb^{\T} \Xb$. However, in practice, $\Xb^{\T} \Xb$ is not often shared and is thus not publicly available. To address this, researchers typically use an external reference panel $\Wb \in \RR^{n_w \times p}$ as a substitute for $\Xb$, with $\Wb^{\T} \Wb$ serving as an approximation for $\Xb^{\T} \Xb$ in LD estimation. One of the most popular reference panels is the 1000 Genomes \citep{10002015global}.
%, UK10K \citep{uk10k2015uk10k}, and TOPMed \citep{taliun2021sequencing}. 
Therefore, the summary-level data considered in this paper are associated with $\Xb^{\T} \yb$ and $\Wb^{\T} \Wb$. 
Condition~\ref{cond-W} assumes that the reference panel $\Wb$ satisfies similar conditions to those imposed on $\Xb$.
% In the context of genetic data predictions, individual-level genotypes in the training data are often not publicly accessible due to privacy and transmission concerns. 
% Even after anonymization, individual-level data can still be vulnerable to privacy attacks (e.g., see \citep{}).  
% Consequently, neither $\Xb$ nor the sample covariance $\Xb^{\T} \Xb$ are available.
% In such scenarios, an external reference panel $\Wb \in \RR^{n_w \times p}$ is commonly used as a substitute for $\Xb$, with $\Wb^{\T} \Wb$ serving the role of $\Xb^{\T} \Xb$. As a result, researchers often must rely on summary statistics , $\Xb^{\T} \yb$ and $\Wb^{\T} \Wb$, instead of having direct access to individual data.

\begin{condition} \label{cond-W}%[Condition on $\Wb$]
    We assume $\Wb = \Wb_0 \bSigma^{1/2} \in \RR^{n_w \times p}$. Each entry of $\Wb_0$ is real-value i.i.d. as that of $\Xb_0$. In addition, the sample size $n_w \to \infty$ while the dimensionality $p \to \infty$, such that the aspect ratio $p/n_w \to\gamma_w > 0$. 
\end{condition}

We define the heritability in genetics as follows, representing the proportion of phenotypic variance attributable to genetic predictors.
Intuitively, a larger heritability implies a higher signal-to-noise ratio \citep{dobriban2018high}.

\begin{definition}%[Heritability and Sparsity]
Conditional on $\bbeta$, the heritability $h^2$ of the training data $(\Xb, \yb)$ is defined as 
$h^2 = \lim_{n, p \to \infty} \var(\Xb \bbeta)/\var(\yb)= \lim_{n, p \to \infty} \bbeta^{\T} \Xb^{\T} \Xb \bbeta/(\bbeta^{\T} \Xb^{\T} \Xb \bbeta +\bepsilon^{\T} \bepsilon)$.
% \begin{align*}
%     \begin{split}
%     h^2 = \lim_{n, p \to \infty} \frac{\var(\Xb \bbeta)}{\var(\yb)} = \lim_{n, p \to \infty} \frac{\bbeta^{\T} \Xb^{\T} \Xb \bbeta}{\bbeta^{\T} \Xb^{\T} \Xb \bbeta + \bepsilon^{\T} \bepsilon}.
%     \end{split}
% \end{align*}
Thus, we have $h^2 \in [0, 1]$.
As $n, p \to \infty$ with $p/n \to \gamma >0$, $h^2$ can be asymptotically represented as 
\begin{align} \label{eqn:h^2}
\begin{split}
    h^2 
    = \lim_{n, p \to \infty} \frac{ \| \bbeta \|_{\bSigma}^2 }{\| \bbeta \|_{\bSigma}^2 + \sigma_{\bepsilon}^2}
    = \lim_{n, p \to \infty} \frac{\kappa \sigma_{\bbeta}^2 \cdot \tr(\bSigma)/p}{\kappa \sigma_{\bbeta}^2 \cdot \tr(\bSigma)/p + \sigma_{\bepsilon}^2}.
\end{split}
\end{align} 
\end{definition}

\subsection{Model training and performance measures}\label{sec:risk_measure}
In this section, we model the processes of individual-level and summary data-based model training and outline the objectives of this paper. 
%\bxz{We may add a Figure 1 to show the two training approaches, as you did in the slides.}
\subsubsection{Individual data-based model training}
\label{sec:ind_split}
We first introduce the prediction estimators and their accuracy measures. 
We begin with the conventional case where individual-level data is accessible and assume that model development involves splitting the whole dataset ($\Xb$, $\yb$) into two independent subsets: a training dataset $(\Xb^{\train}, \yb^{\train})$ and a validation dataset $(\Xb^{\valid}, \yb^{\valid})$. 
% \bxz{We can consider rename ts to v}
We model the data-splitting procedure as follows: given the design matrix $\Xb$, we sample the training dataset by considering a diagonal matrix $\Qb = \diag \{q_1, q_2, \cdots q_n\}$, where $q_{i} \overset{i.i.d.}{\sim} \textnormal{Bernoulli}(n^{\train}/n)$ for some $n^{\train}$ that is proportional to and smaller than $n$. 
Let $\Xb^{\train} = \Qb \Xb$ and $\Xb^{\valid} = (\Ib_n - \Qb) \Xb$, with $\yb^{\train}$ and $\yb^{\valid}$ being defined accordingly. 
By Condition \ref{cond-X}, $(\Xb^{\train}, \yb^{\train})$  is conditionally independent from $(\Xb^{\valid}, \yb^{\valid})$. 
One can train a general estimator $\hat{\bbeta}_{\rm G}(\theta)$ on $(\Xb^{\train}, \yb^{\train})$ and evaluate its prediction performance on $(\Xb^{\valid}, \yb^{\valid})$ using the out-of-sample $R^2$, given by 
$R_{\rm ind, G}^2(\theta)=\langle {\Xb^{\valid}{}^{\T}\yb^{\valid}}, \hat{\bbeta}_{\rm G}(\theta)\rangle^2/\{\| \Xb^{\valid} \hat{\bbeta}_{\rm G}(\theta)\|_2^2 \cdot \|\yb^{\valid}\|_{2}^2\}$. 
%\bxs{Here, we use $\hat{\bbeta}_{\rm G}(\theta)$ to denote the general estimator, which aligns with the definition in Section \ref{sec:general_lin}.}
The out-of-sample $R^2$ is a widely used measure of prediction accuracy in genetic data prediction \citep{ma2021genetic} and is closely related to the mean squared error \citep{su2024exact}.

In this paper, 
%since our goal is to evaluate self-training with summary statistics $\Xb^{\T} \yb$ and $\Wb^{\T} \Wb$, 
we consider a class of general estimators $\hat{\bbeta}_{\rm G}(\theta)$ that are linear with respect to the summary statistic $\Xb^{\train}{}^{\T} \yb^{\train}$ and parameterized by the scale or vector $\theta$, possibly incorporating the reference panel data $\Wb^{\T} \Wb$. This is because we aim to model the real-world scenario where summary statistics from the training dataset, $\Xb^{\train}{}^{\T} \yb^{\train}$, are often publicly available, and the general estimator $\hat{\bbeta}_{\rm G}(\theta)$ is trained using independent individual-level validation data $(\Xb^{\valid}, \yb^{\valid})$, possibly with external LD reference panel. 
Specifically, we define $\hat{\bbeta}_{\rm G}(\theta)$ as 
\begin{equation} \label{eqn:estimator}
    \hat{\bbeta}_{\rm G}(\theta)= \Ab(\Wb^{\T} \Wb, \theta) \Xb^{\train}{}^{\T} \yb^{\train}.
\end{equation}
Here $\Ab(\Wb^{\T} \Wb, \theta) \in \RR^{p \times p}$ is a matrix depending on $\theta$ and potentially on $\Wb^{\T} \Wb$.
We denote the estimator as $\hat{\bbeta}_{\rm G}(\theta)$ for a general scale or vector $\theta$.
In subsequent sections, we may use $\Theta$ to explicitly denote a vector parameter. For each specific estimator analyzed below, we will clearly specify the domain of $\theta$ or $\Theta$.
%represent the estimators with the vector of parameters $\Theta$.
The out-of-sample $R^2$ of $\hat{\bbeta}_{\rm G}(\theta)$ is defined as
\begin{align}\label{eqn:R2_ind_def}
    \begin{split}
        R_{\rm ind, G}^2(\theta) 
        = \frac{\left\langle {\Xb^{\valid}{}^{\T}\yb^{\valid}}, \hat{\bbeta}_{\rm G}(\theta) \right\rangle^2}{ \| \Xb^{\valid} \hat{\bbeta}_{\rm G}(\theta)\|_2^2 \cdot \|\yb^{\valid}\|_{2}^2 } = \frac{n^{\valid}}{\|\yb^{\valid}\|_{2}^2} \cdot \frac{\left\langle {\Xb^{\valid}{}^{\T}\yb^{\valid}}, \Ab(\Wb^{\T} \Wb, \theta) \Xb^{\train}{}^{\T} \yb^{\train} \right\rangle^2}{n^{\valid} \cdot \| \Ab(\Wb^{\T} \Wb, \theta) \Xb^{\train}{}^{\T} \yb^{\train} \|_{\bSigma}^2} + o_{p}(1).
    \end{split}
\end{align}
In the model training process, one aims to find the best hyperparameter by choosing the 
$\theta$ that maximizes $R_{\rm ind, G}^2(\theta)$ in \eqref{eqn:R2_ind_def}, as specified in Algorithm \ref{alg:ind} of the
supplementary material. %\bxz{Can we change it to Algorithm S.1.?} $\theta^{*}$

Estimators in the form of Equation~\eqref{eqn:estimator} and their variants are widely used in PRS applications \citep{power2015polygenic,prive2019making,ge2019polygenic,ma2021genetic}. 
An example is the ridge-type estimator, which can be formulated as 
\begin{equation} \label{eqn:ref-panel-ridge}
\hat{\bbeta}_{\rm R}(\theta) = (\Wb^{\T} \Wb + \theta n_w \Ib_{p})^{-1} \Xb^{\train}{}^{\T} \yb^{\train}
\end{equation}
for any $\theta \in \RR_{+}$.  
In the case of the ridge-type estimator, the best-performing hyperparameter in $\hat{\bbeta}_{\rm R}(\theta)$ is selected by optimizing the following expression
\begin{align*}
    \theta^{*}_{\rm ind, R} = \arg\max_{\theta \in \RR_{+}} R^2_{\rm ind, R}(\theta), \quad \text{with} \quad R^2_{\rm ind, R}(\theta) = \frac{n^{\valid}}{\|\yb^{\valid}\|_{2}^2} \cdot \frac{\left\langle {\Xb^{\valid}{}^{\T}\yb^{\valid}}, (\Wb^{\T} \Wb + \theta n_w \Ib_{p})^{-1} \Xb^{\train}{}^{\T} \yb^{\train} \right\rangle^2}{n^{\valid} \cdot \| (\Wb^{\T} \Wb + \theta n_w \Ib_{p})^{-1} \Xb^{\train}{}^{\T} \yb^{\train} \|_{\bSigma}^2}. 
\end{align*}
It is worth noting that since $n^{\valid}$  and $\|\yb^{\valid}\|_{2}^2$ are constant across all $\theta$, the key contribution of the validation data $(\Xb^{\valid}, \yb^{\valid})$ to the training process is through the item $\Xb^{\valid}{}^{\T}\yb^{\valid}$. As detailed in later sections, this key insight motivates the summary data-based model training approach.

\subsubsection{Summary data-based model training}
\label{sec:sum_split}
Now we consider the practical scenario where only summary statistics, $\Xb^{\T} \yb$ and $\Wb^{\T} \Wb$, are available, rather than having access to $\Xb^{\train}{}^{\T} \yb^{\train}$, $\Wb^{\T} \Wb$, and individual-level data $(\Xb^{\valid}, \yb^{\valid})$. 
%$(\Xb^{\train}, \yb^{\train})$,$(\Xb^{\valid}, \yb^{\valid})$, as well as the split-dependent summary statistics $\Xb^{\train}{}^{\T} \yb^{\train}$ and 
It follows that we are not able to obtain the $R_{\rm ind, G}^2(\theta)$ defined in Equation~\eqref{eqn:R2_ind_def}. 

We model the resampling-based self-training procedure with $\Xb^{\T} \yb$ and $\Wb^{\T} \Wb$ as follows. 
Given the summary statistic $\Xb^{\T} \yb$, we sample pseudo-training and validation summary statistics, $\sbb^{\train}$ and $\sbb^{\valid}$, from the approximate distributions of $\Xb^{\train}{}^{\T} \yb^{\train}$ and $\Xb^{\valid}{}^{\T} \yb^{\valid}$, as detailed by pseudo-code in Algorithm \ref{alg:sum}. These sampled statistics are expected to closely mimic what would be obtained if individual-level data were available. For example, when only summary statistics are available, the ridge-type estimator is now trained on $\sbb^{\train}$ and given by 
\begin{align} \label{eqn:ref-panel-ridge-sum}
    \hat{\bbeta}_{\rm R}(\theta)^{*} = (\Wb^{\T} \Wb + \theta n_w \Ib_{p})^{-1} \sbb^{\train}
\end{align}
for any $\theta \in \RR_{+}$. 

\begin{algorithm}
\caption{Summary data-based model training}\label{alg:sum}
\begin{algorithmic}
\Require Summary data $\Xb^{\T} \yb$, $\Wb^{\T} \Wb$,  and hyperparameter $\theta$. 
\vspace{1mm}

\State $\hb \gets \cN(0,\Ib_p),$  \hfill 		\texttt{//} Sample $p$-dimension standard Gaussian random variable.
\vspace{1mm}

\State $\sbb^{\train} \gets \frac{n^{\train}}{n} \Xb^{\T} \yb + \sqrt{\frac{n^{\train} (n - n^{\train})}{n^2}} \Cov(\Xb^{\T} \yb)^{1/2} \hb,$  \hfill 		\texttt{//} Construct $\sbb^{\train}$ for training.
\vspace{1mm}

\State $\sbb^{\valid} \gets \Xb^{\T} \yb - \sbb^{\train},$  \hfill 		\texttt{//} Construct $\sbb^{\valid}$ for validation.
\vspace{1mm}

\State $\hat{\bbeta}_{\rm G}(\theta)^* \gets \Ab(\Wb^{\T} \Wb, \theta) \sbb^{\train},$  \hfill 		\texttt{//} Obtain the estimator.
\vspace{1mm}

\State $R_{\rm sum, G}^2(\theta) \gets ({n^{\valid}}\left/{\|\yb^{\valid}\|_{2}^2}\right.) \cdot {\left \langle {\sbb^{\valid}}, \hat{\bbeta}_{\rm G}(\theta)^* \right \rangle^{2} } \left/({n^{\valid} \cdot \| \hat{\bbeta}_{\rm G}(\theta)^* \|_{\bSigma}^{2}}), \right.$\vspace{1mm}

\Statex \hfill \texttt{//} Compute the $R^2_{\rm sum, G}(\theta)$ in \eqref{eqn:R2_sum_def}.
\vspace{1mm}

\State $\theta^{*}_{\rm sum, G} \gets \max_{\theta} R^2_{\rm sum, G}(\theta),$  \hfill 		\texttt{//} Choose the best-performing hyperparameter.
\vspace{1.5mm}

\Return $R_{\rm sum, G}^2(\theta)$ and $\theta^{*}_{\rm sum, G}$.

\end{algorithmic}
\end{algorithm}
Notably, in Algorithm \ref{alg:sum}, the hyperparameter $\theta$ is tuned on $\sbb^{\valid}$, with the measure being defined as follows
%its counterpart $\hat{\bbeta}_{\rm G}(\theta)^* = \Ab(\Wb^{\T} \Wb, \theta) \sbb^{\valid}$, and 
\begin{align} \label{eqn:R2_sum_def}
    R^2_{\rm sum, G}(\theta) = \frac{\left\langle \sbb^{\valid}, \hat{\bbeta}_{\rm G}(\theta)^* \right\rangle^2}{ \| \Xb^{\valid} \hat{\bbeta}_{\rm G}(\theta)^*\|_2^2 \cdot \|\yb^{\valid}\|_{2}^2 } = \frac{n^{\valid}}{\|\yb^{\valid}\|_{2}^2} \cdot \frac{\left \langle {\sbb^{\valid}}, \Ab(\Wb^{\T} \Wb, \theta) \sbb^{\train} \right \rangle^{2} }{n^{\valid} \cdot \| \Ab(\Wb^{\T} \Wb, \theta) \sbb^{\train} \|_{\bSigma}^{2}} + o_p(1).
\end{align}
The key difference between $R^2_{\rm sum, G}(\theta)$ from Algorithm \ref{alg:sum} and $R^2_{\rm ind, G}(\theta)$ from Algorithm \ref{alg:ind} lies in using $\sbb^{\train}$ and $\sbb^{\valid}$ in place of $\Xb^{\train}{}^{\T} \yb^{\train}$ and $\Xb^{\valid}{}^{\T} \yb^{\valid}$, respectively.
To compare the performance of Algorithm \ref{alg:sum} and Algorithm \ref{alg:ind}, we consider 
the ratio 
% We aim to compare the performance of Algorithms \ref{alg:ind} and \ref{alg:sum} in terms of prediction accuracy. Specifically, we aim to compare the testing prediction accuracy, $R^2_{\rm ind}(\theta)$ from Algorithm \ref{alg:ind} and $R^2_{\rm sum}(\theta)$ from Algorithm \ref{alg:sum}, by considering  following ratio 
\begin{align*}
    R^2_{\rm sum, G}(\theta)/R^2_{\rm ind, G}(\theta).
\end{align*}
When $R^2_{\rm sum, G}(\theta)/R^2_{\rm ind, G}(\theta)$ converges to $1$ for any $\theta$ as $n$ and $p \to \infty$, we say Algorithm \ref{alg:sum} has no additional prediction accuracy cost compared to Algorithm \ref{alg:ind} and expect the two algorithms return the same best-performing hyperparameter, that is, we have $\theta^{*}_{\rm sum, G} = \theta^{*}_{\rm ind, G}$.

\subsection{Fixed-dimension intuition of Algorithm \ref{alg:sum}}
\label{sec:intuition}
%and provide key intuitions to help readers understand the insights of the self-training algorithm.

Before presenting the formal results in the next section, we first outline the key intuitions behind resampling-based self-training when it is applied in practical contexts in precision medicine and genetic research \citep{jin2025pennprs}. 
%\citep{zhao2024optimizing}. 
Notably, while the PRS applications of the self-training algorithm lie in high-dimensional genetic predictors, its underlying intuition originates from classical fixed-dimension arguments. However, as with many statistical phenomena, we find that the simplicity of fixed-dimension intuition of self-training does not directly extend to high-dimensional settings, where increased complexity invalidates the original reasoning. The continued effectiveness of the self-training algorithm in high-dimensional data suggests the presence of deeper underlying factors, which require investigation using entirely different approaches and technical tools, such as the random matrix theory used in our formal analysis. 

We begin by closely examining the measure $R^2_{\rm ind, G}(\theta)$. Due to the linear structure in marginal summary statistics, we always have the data split ${\Xb^{\valid}{}^{\T}\yb^{\valid}} = \Xb^{\T} \yb - \Xb^{\train}{}^{\T}\yb^{\train}$. 
It follows that 
\begin{align*}
    R^2_{\rm ind, G}(\theta) 
    =\ & \frac{n^{\valid}}{\|\yb^{\valid}\|_{2}^2} \cdot \frac{\left\langle {\Xb^{\T} \yb - \Xb^{\train}{}^{\T}\yb^{\train}}, \Ab(\Wb^{\T} \Wb, \theta) \Xb^{\train}{}^{\T} \yb^{\train} \right\rangle^2}{n^{\valid} \cdot \| \Ab(\Wb^{\T} \Wb, \theta) \Xb^{\train}{}^{\T} \yb^{\train} \|_{\bSigma}^2} + o_p(1). 
\end{align*}
Therefore, conditional on the summary data $\Xb^{\T} \yb $ and $\Wb^{\T} \Wb$, we only need to resample once to obtain $\sbb^{\train}$ as an approximation of $\Xb^{\train}{}^{\T} \yb^{\train}$, and can then set $\sbb^{\valid}=\Xb^{\T} \yb -\sbb^{\train}$.
In practice, Algorithm \ref{alg:sum} samples $\sbb^{\train}$ from a $p$-dimensional Gaussian random variable with mean $(n^{\train}/n) \Xb^{\T} \yb$ and covariance $[n^{\train} (n - n^{\train})]/n^2 \cdot \Cov(\Xb^{\T} \yb)$, where $\Cov(\Xb^{\T} \yb)$ is defined to be
\begin{align*}
    \Cov(\Xb^{\T} \yb) = \left(\Xb^{\T} \yb - n \bSigma \bbeta \right) \left(\Xb^{\T} \yb - n \bSigma \bbeta \right)^{\T}.
\end{align*}
This sampling distribution is chosen based on the fixed-dimension scenario where $p$ is fixed, while the sample size $n \to \infty$. 
In such a setting, this sampling distribution of $\sbb^{\train}$ approaches the limiting distribution of $\Xb^{\train}{}^{\T} \yb^{\train}$ according to the multivariate central limit theorem (CLT).    
To show this, note that $\Xb^{\train}{}^{\T} \yb^{\train} = \Xb^{\T} \Qb^{\T} \Qb \yb = \Xb^{\T} \Qb \yb$. 
% \begin{align*}
%     \Xb^{\train}{}^{\T} \yb^{\train} = \Xb^{\T} \Qb^{\T} \Qb \yb = \Xb^{\T} \Qb \yb. 
% \end{align*}, applied to the randomness of $\Qb$, 
The multivariate CLT implies that
\begin{align*}
    \frac{1}{\sqrt{n}} \left( \Xb^{\train}{}^{\T} \yb^{\train} - \frac{n^{\train}}{n} \Xb^{\T} \yb \right) = \frac{1}{\sqrt{n}} \sum_{j=1}^{n} \left( q_{j} \Xb_j \yb_j - \frac{n^{\train}}{n} \Xb_j \yb_j \right) \overset{d}{\to} N(\bm 0, \bSigma),
\end{align*}
where $\bSigma$ is the covariance matrix given by
\begin{align*}
    \bSigma = \lim_{n \to \infty} \EE_{\Qb} \Cov\left( \frac{1}{\sqrt{n}} \sum_{j=1}^{n} q_{j} \Xb_j \yb_j \right) = \lim_{n \to \infty} \frac{n^{\train} (n - n^{\train})}{n^2} \cdot \frac{1}{n} \Cov\left( \Xb^{\T} \yb \right). 
\end{align*}
It further implies that 
$$n^{-1/2} \Xb^{\train}{}^{\T} \yb^{\train}  - n^{-1/2} \sbb^{\train} = n^{-1/2} \Xb^{\T} \Qb \yb - n^{-1/2} \sbb^{\train} \overset{d}{\to} 0.$$
By the continuous mapping theorem, we have
\begin{align*}
    R^2_{\rm ind, G}(\theta)/R^2_{\rm sum, G}(\theta) =\ &  \frac{\left\langle \Xb^{\T} \yb - \Xb^{\T} \Qb \yb, \Ab(\Wb^{\T} \Wb, \theta) \Xb^{\train}{}^{\T} \yb^{\train} \right\rangle^2}{n^{\valid} \cdot \| \Ab(\Wb^{\T} \Wb, \theta) \Xb^{\train}{}^{\T} \yb^{\train} \|_{\bSigma}^2} \Bigg/  \frac{\left\langle \Xb^{\T} \yb - \sbb^{\train}, \Ab(\Wb^{\T} \Wb, \theta) \sbb^{\train} \right\rangle^2}{n^{\valid} \cdot \| \Ab(\Wb^{\T} \Wb, \theta) \sbb^{\train} \|_{\bSigma}^2} 
     \overset{d}{\to} 1. 
\end{align*}
Because converging in distribution to a constant implies convergence in probability to the constant, we have $R^2_{\rm ind, G}(\theta)/R^2_{\rm sum, G}(\theta) \overset{p}{\to} 1$. 
In summary, when the dimensionality $p$ is fixed, Algorithm \ref{alg:sum} has no additional prediction accuracy cost compared to Algorithm \ref{alg:ind} by resampling $\sbb^{\train}$ directly from the asymptotic distribution of $\Xb^{\train}{}^{\T} \yb^{\train}$. 

While the above fixed $p$ discussions provide insight into the rationale behind the construction of $\sbb^{\train}$ in Algorithm \ref{alg:sum}, it is important to note that these CLT-based derivations in low dimensions cannot be easily extended to middle or high-dimensional applications, such as our PRS genetic data prediction problem, where both $n$ and $p \to \infty$, as described in Condition \ref{cond-np-ratio}. Briefly, this is because high-dimensional CLT is much more complex, and general results across the parameter space typically do not exist \citep{chernozhukov2017central, fang2021high}. 
In other words, the $p$-dimensional Gaussian distribution, from which $\sbb^{\train}$ is sampled, is no longer the asymptotic distribution of $\Xb^{\train}{}^{\T} \yb^{\train}$. 
Therefore, new tools and proof framework are required to evaluate the relative performance of Algorithm \ref{alg:sum} and Algorithm \ref{alg:ind} in high dimensions. 

Despite this mismatch in distributions, a closer examination of $R^2_{\rm ind, G}(\theta)$ and $R^2_{\rm sum, G}(\theta)$ suggests that we may only need the following approximations to hold
   \begin{align*}
        \EE_{\Xb, \yb} \sbb^{\train} \approx \EE_{\Xb, \yb, \Qb} \Xb^{\train}{}^{\T} \yb^{\train} \quad
        \mbox{and} \quad \EE_{\Xb, \yb} \left\langle \sbb^{\train}, 
        \Ab(\Wb^{\T} \Wb, \theta) \sbb^{\train} \right\rangle \approx \EE_{\Xb, \yb, \Qb} \left\langle \Xb^{\train}{}^{\T} \yb^{\train}, \Ab(\Wb^{\T} \Wb, \theta) \Xb^{\train}{}^{\T} \yb^{\train} \right\rangle.
    \end{align*}
This observation suggests that achieving the same prediction accuracy may only require the first and second moments of $\sbb^{\train}$ to match those of $\Xb^{\train}{}^{\T} \yb^{\train}$, rather than matching the entire distribution. 
Rigorous analyses of these approximations are provided in the next section, with the proof relying heavily on random matrix theory \citep{bai2010spectral,Yao_Zheng_Bai_2015} and deterministic equivalents \citep{dobriban2021distributed}. We derive the closed-form asymptotic results for $R^2_{\rm ind, G}(\theta)$ and $R^2_{\rm sum, G}(\theta)$, demonstrating that the no-cost property of Algorithm \ref{alg:sum} remains valid in high dimensions without requiring the sampling distribution of  $\sbb^{\train}$ to be the asymptotic distribution of $\Xb^{\train}{}^{\T} \yb^{\train}$.

\section{Asymptotic results of self-training}
\label{sec:general}
In this section, we present random matrix theory results to demonstrate that self-training with $\sbb^{\train}$ can achieve the same asymptotic predictive accuracy as using $(\Xb^{\train}, \yb^{\train})$ in high dimensions. In Section \ref{sec:ref_ridge}, we examine the reference panel-based ridge estimator defined in Equation~\eqref{eqn:ref-panel-ridge}. The marginal thresholding estimator is evaluated in Section \ref{sec:marginal}. These two concrete examples are motivated by the widely used ridge-type estimators and thresholding procedures in genetic risk prediction \citep{ge2019polygenic, choi2020tutorial}. 
In Section \ref{sec:general_lin}, we further provide a general theorem for the class of linear estimators defined in Equation~\eqref{eqn:estimator}.

\subsection{Resampling-based ridge estimator}
\label{sec:ref_ridge}
%\bxs{Define $R^2_{\rm ind, R}(\theta)$ and $R^2_{\rm sum, R}(\theta)$ as the out-of-sample $R^2$ in Equations \eqref{eqn:R2_ind_def} and \eqref{eqn:R2_sum_def}, replacing $\hat{\bbeta}_{\rm G}(\theta), \hat{\bbeta}_{\rm G}(\theta)^*$ with $\hat{\bbeta}_{\rm R}(\theta), \hat{\bbeta}_{\rm R}(\theta)^*$, respectively.}
We quantify the out-of-sample performance of reference panel-based ridge estimators $\hat{\bbeta}_{\rm R}(\theta)$ in Equation \eqref{eqn:ref-panel-ridge} and $\hat{\bbeta}_{\rm R}(\theta)^*$ in Equation \eqref{eqn:ref-panel-ridge-sum} for $\theta \in \RR_{+}$, denoted as $R^2_{\rm ind, R}(\theta)$ and $R^2_{\rm sum, R}(\theta)$, respectively. 
Our analysis mainly relies on deterministic equivalents \citep{dobriban2021distributed}, which is a recent tool established based on standard random matrix theory \citep{bai2010spectral}. 
For deterministic or random matrix sequences $\Db$ and $\Eb \in \RR^{n \times p}$ and $n, p \to \infty$ proportionally, we say $\Db$ and $\Eb$ are deterministic equivalent and denote this as $\Db \asymp \Eb$ if
$ \lim_{n,p \to \infty} \left| \tr \left[ \Cb (\Db - \Eb) \right] \right| = 0$
% \begin{align*}
% \end{align*}
almost surely for any sequence $\Cb$ of matrices with bounded trace norm, that is
\begin{align} \label{eqn:bnd_trace_norm}
    \lim \sup_{n,p \to \infty} \tr \left[ (\Cb^{\T} \Cb)^{1/2} \right] < \infty.
\end{align}
Here $\Cb$, $\Db$, and $\Eb$ are not necessarily symmetric. 
%\bxz{You need $p \to \infty$ in many places of this section?}\bxs{done.}
We present the detailed results on deterministic equivalence in Section \ref{sec:prelim} of the supplementary material. 
The following lemma provides the first and second-order generalization of the classic Marchenko-Pastur Law for the sample covariance matrix. 
%The proof will defer to Appendix \ref{sec:prelim}. 
\begin{lemma} \label{lemma:DE_second_order}
Let $\hat{\bSigma}_{n} = \Xb^{\T} \Xb/n$.
Under Conditions \ref{cond-np-ratio} - \ref{cond-X}, for any $\theta \in \RR_{+}$, with probability one, we have 
%\begin{align*}
    $(\hat{\bSigma}_{n} + \theta  \Ib_{p})^{-1} \asymp (\tau_{n} (\theta) \bSigma + \theta \Ib_{p})^{-1}$
%\end{align*}
and
\begin{equation} \label{eqn:second_order_MP}
\begin{split}
    \left(\hat{\bSigma}_{n} + \theta \Ib_{p}\right)^{-1} \bSigma \left( \hat{\bSigma}_{n}  + \theta \Ib_{p}\right)^{-1}  \asymp (\rho_{n}(\theta) + 1) \cdot \left( \tau_{n}(\theta)\bSigma  + \theta \Ib_{p}\right)^{-1} \bSigma \left(\tau_{n}(\theta)\bSigma  + \theta \Ib_{p}\right)^{-1}.
\end{split}
\end{equation}
Here $\tau_{n}(\theta)$ and $\rho_{n}(\theta)\in \CC_{+}$ are solutions to the fixed point equations 
    \begin{equation} \label{eqn:tau}
        \tau_{n}(\theta)^{-1}=1+\frac{1}{n}\tr\left[\bSigma (\Ab + \tau_{n}(\theta) \bSigma + \theta \Ib_p)^{-1}\right]
    \end{equation}
    and 
    \begin{align} \label{eqn:rho}
        \rho_{n}(\theta) = \left(1-\frac{\tau_{n}^2(\theta)}{n}\tr\left[(\tau_{n}(\theta) \bSigma + \theta \Ib_{p})^{-2} \bSigma^2 \right]\right)^{-1} \frac{\tau_{n}^2(\theta)}{n}\tr\left[(\tau_{n}(\theta) \bSigma + \theta \Ib_{p})^{-2} \bSigma^2 \right] .
    \end{align}
\end{lemma}
The proof of Lemma \ref{lemma:DE_second_order} is based on the generalized Marchenko-Pastur Law (e.g., \cite{rubio2011spectral}) and Lemma \ref{lemma:diff_deter_equi} in Section \ref{sec:prelim} of the
supplementary material, stating that differentiation and deterministic equivalence are interchangeable. 
Lemma~\ref{lemma:second_moment} below is also crucial to our proof, as it establishes the asymptotic results for the trace of a matrix structure involving $\Cb \hat{\bSigma}_{n}^2$. The detailed proof is included in Section \ref{sec:proof_of_second_mom} of the
supplementary material.

\begin{lemma} \label{lemma:second_moment}
    For any deterministic positive semi-definite symmetric matrix $\Cb \in \RR^{p \times p}$ satisfying Equation \eqref{eqn:bnd_trace_norm}, we have
    \begin{align} \label{eqn:second_moment}
        \frac{1}{p} \tr\left(\Cb \hat{\bSigma}_{n}^2 \right) - \left\{ \frac{1}{n} \tr \left(\bSigma \right) \cdot \frac{1}{p} \tr \left(\Cb \bSigma \right) + \frac{1}{p} \tr \left(\Cb \bSigma^2 \right) \right\} \overset{p}{\to} 0. 
    \end{align}
\end{lemma}

Lemma \ref{lemma:DE_second_order} and Lemma~\ref{lemma:second_moment} are used to analyze the limits of functionals involved in ridge-type estimators.
Theorem~\ref{thm:reference_panel} below provides the asymptotic prediction accuracy of $\hat{\bbeta}_{\rm R}(\theta)$ and $\hat{\bbeta}_{\rm R}(\theta)^*$ respectively trained by Algorithm \ref{alg:ind} and Algorithm \ref{alg:sum}, demonstrating that Algorithm \ref{alg:sum} has no additional cost compared to Algorithm \ref{alg:ind} for reference panel-based ridge regression.

%Equivalently, we consider the case $\hat{\bbeta}_{\rm G}(\theta) := \Ab(\Wb^{\T} \Wb, \theta) \Xb^{\train}{}^{\T} \yb^{\train}$ with $\Ab(\Wb^{\T} \Wb, \theta) = ( \Wb^{\T} \Wb + \theta n_w \Ib_{p} )^{-1}$. 
\begin{theorem} \label{thm:reference_panel}
    Consider any random sequence $\{\bbeta, \bepsilon, \Xb, \Wb\}_{(p, n, n_w) \in \NN^3}$ satisfying Conditions \ref{cond-np-ratio}-\ref{cond-W} and any $\theta \in \RR_{+}$. For $\hat{\bbeta}_{\rm R}(\theta)^*$ defined in Equation \eqref{eqn:ref-panel-ridge-sum}, 
    the out-of-sample $R^2$ is 
    $$
    R^2_{\rm sum, R}(\theta) = \frac{n^{\valid}}{\|\yb^{\valid}\|_{2}^2} \cdot \frac{n^{\train}}{p} \cdot \kappa \sigma_{\bbeta}^2 \cdot \frac{\left( \tr \left[ \left( \tau_{n_w}(\theta) \bSigma + \theta \Ib_p \right)^{-1} \bSigma^2 \right]\right)^2}{\tr \left(\bSigma \right) \cdot \tr \left(\Bb(\theta) \bSigma \right)/h^2 + n^{\train} \cdot \tr \left(\Bb(\theta) \bSigma^2 \right)}  + o_p(1),
    $$ 
   where $\Bb(\theta)=(\rho_{n_w}(\theta) + 1)\left( \tau_{n_w}(\theta)\bSigma + \theta \Ib_{p}\right)^{-1} \bSigma \left(\tau_{n_w}(\theta)\bSigma + \theta \Ib_{p}\right)^{-1}$ is the right-hand side of Equation \eqref{eqn:second_order_MP},
    and $\tau_{n_w}(\theta)$ and $\rho_{n_w}(\theta)$ are defined in Equations \eqref{eqn:tau} and \eqref{eqn:rho} by replacing $n$ by $n_w$, respectively.
    Furthermore, for any $\theta \in \RR_{+}$, we have $$R^2_{\rm sum, R} (\theta) / R^2_{\rm ind, R} (\theta) \overset{p}{\to} 1.$$ 
\end{theorem}

Theorem \ref{thm:reference_panel} establishes the asymptotic prediction accuracy in resampling-based self-training and demonstrates its equivalence to conventional training when individual-level data are available.
By the definitions of $R^2_{\rm sum, R}(\theta)$ and $R^2_{\rm ind, R}(\theta)$, we have  
\begin{align*}
    R^2_{\rm sum, R}(\theta) \propto \frac{\left \langle {\sbb^{\valid}}, (\Wb^{\T} \Wb + \theta n_w \Ib_{p})^{-1} \sbb^{\train} \right \rangle^{2} }{n^{\valid}{}^{2} \cdot \| (\Wb^{\T} \Wb + \theta n_w \Ib_{p})^{-1} \sbb^{\train} \|_{\bSigma}^{2}}\quad \mbox{and} \quad R^2_{\rm ind, R}(\theta) \propto  \frac{\left\langle {\Xb^{\valid}{}^{\T}\yb^{\valid}}, (\Wb^{\T} \Wb + \theta n_w \Ib_{p})^{-1} \Xb^{\train}{}^{\T} \yb^{\train} \right\rangle^2}{n^{\valid}{}^{2} \cdot \| (\Wb^{\T} \Wb + \theta n_w \Ib_{p})^{-1} \Xb^{\train}{}^{\T} \yb^{\train} \|_{\bSigma}^2}. 
\end{align*}
For $R^2_{\rm ind, R}(\theta)$, the estimator $\hat{\bbeta}_{\RA}(\theta)$ is trained on $(\Xb^{\train}, \yb^{\train})$, and the out-of-sample $R_{\rm ind, R}^2(\theta)$ is calculated on an independent dataset $(\Xb^{\valid}, \yb^{\valid})$. 
In contrast, when computing $R^2_{\rm sum, R}(\theta)$, the resampling-based pseudo-training and validation summary statistics $\sbb^{\train}$ and $\sbb^{\valid}$ are inherently dependent. Surprisingly, Theorem \ref{thm:reference_panel} demonstrates that $R^2_{\rm sum, R} (\theta) / R^2_{\rm ind, R} (\theta) \overset{p}{\to} 1$, suggesting that tuning parameters using dependent pseudo-training/validation datasets does not lead to overfitting or reduced out-of-sample prediction performance in high dimensions. Below, we provide a proof sketch to provide insights into this counterintuitive result.

\begin{proof}[Proof sketch of Theorem \ref{thm:reference_panel}]
    By the continuous mapping theorem, it suffices to show that the numerator and denominator of $R^2_{\rm sum, R}(\theta)$ are asymptotically equal to that of $R^2_{\rm ind, R}(\theta)$, respectively. 
    We take the numerator as an example and sketch the proof of Equation \eqref{eqn:num_proof_ridge} below. 
    \begin{align}\label{eqn:num_proof_ridge}
        \left \langle {\sbb^{\valid}}, (\Wb^{\T} \Wb + \theta n_w \Ib_{p})^{-1} \sbb^{\train} \right \rangle/n^{\valid} = \left\langle {\Xb^{\valid}{}^{\T}\yb^{\valid}}, (\Wb^{\T} \Wb + \theta n_w \Ib_{p})^{-1} \Xb^{\train}{}^{\T} \yb^{\train} \right\rangle/n^{\valid} + o_p(1)
    \end{align}
    Using Lemma \ref{lemma:DE_second_order} and Lemma B.26 from \cite{bai2010spectral}, we can conclude the right-hand side of Equation \eqref{eqn:num_proof_ridge} that
    \begin{align} \label{eqn:RHS}
    \begin{split}
        \left\langle {\Xb^{\valid}{}^{\T}\yb^{\valid}},\right.&\left. (\Wb^{\T} \Wb + \theta n_w \Ib_{p})^{-1} \Xb^{\train}{}^{\T} \yb^{\train} \right\rangle/n^{\valid}
        \\
        & \qquad \qquad 
        = \left\{ \frac{n^{\train}}{n_w} \cdot \frac{\kappa \sigma_{\bbeta}^2}{p} \cdot \tr\left[ \left( \tau_{n_w}(\theta) \bSigma + \theta \Ib_p \right)^{-1} \bSigma^2 \right] \right\} + o_p(1).
    \end{split}
    \end{align}
    To show that the left-hand side of Equation  \eqref{eqn:num_proof_ridge} converges to the same limit, we first decompose it and omit the zero-limit cross terms, leaving three non-zero terms as follows 
    % = \left(\Xb^{\T} \yb - n \bSigma \bbeta \right) \left(\Xb^{\T} \yb - n \bSigma \bbeta \right)^{\T}
    \begin{align*}
    & \left \langle {\sbb^{\valid}}, (\Wb^{\T} \Wb + \theta n_w \Ib_{p})^{-1} \sbb^{\train} \right \rangle/n^{\valid}
    \\
    \ & =\left( \frac{n^{\train}}{n} \Xb^{\T}\Xb \bbeta  \right)^{\T} (\Wb^{\T} \Wb + \theta n_w \Ib_{p})^{-1} \left(  \frac{n - n^{\train}}{n} \Xb^{\T}\Xb \bbeta \right)/n^{\valid}
    \\
    &+ \left( \frac{n^{\train}}{n} \Xb^{\T} \bepsilon \right)^{\T} (\Wb^{\T} \Wb + \theta n_w \Ib_{p})^{-1} \left( \frac{n - n^{\train}}{n} \Xb^{\T}\bepsilon \right)/n^{\valid}
    \\
    &- \left( \sqrt{\frac{n^{\train} (n - n^{\train})}{n^2}} \Cov(\Xb^{\T} \yb)^{1/2} \hb \right)^{\T} (\Wb^{\T} \Wb + \theta n_w \Ib_{p})^{-1} \left( \sqrt{\frac{n^{\train} (n - n^{\train})}{n^2}} \Cov(\Xb^{\T} \yb)^{1/2} \hb \right)/n^{\valid}
    \\
     & =\Rom{1}^{(1)}_{\rm sum} - \Rom{1}^{(2)}_{\rm sum}.
\end{align*}
% We group first two summation as $\Rom{1}^{(1)}_{\rm sum}$ as it arises from first moment $n^{\train}/n \cdot \Xb^{\train}{}^{\T} \yb^{\train}$ in $\sbb^{\train}$. 
% We use $\Rom{1}^{(2)}_{\rm sum}$ to denote the last second moment term 
We group the first two terms as $\Rom{1}^{(1)}_{\rm sum}$. 
Intuitively, $\Rom{1}^{(1)}_{\rm sum}$ can be interpreted as a ``dependence-induced term" since it arises from the dependence between $\sbb^{\train}$ and $\sbb^{\valid}$, which is absent from the right-hand side of Equation \eqref{eqn:num_proof_ridge}, where the training and validation datasets are independent.
Furthermore, we refer to $\Rom{1}^{(2)}_{\rm sum}$ as a ``compensatory term", as we will demonstrate that its limiting behavior precisely offsets the dependence-induced term.
Lemma B.26 in \cite{bai2010spectral} allows us to reorganize  $\Rom{1}^{(1)}_{\rm sum}$ as  
\begin{align*}
    \Rom{1}^{(1)}_{\rm sum} =\ & \frac{n^{\train}}{n_w} \cdot \frac{\kappa \sigma_{\bbeta}^2}{p} \cdot \tr \left[ (\hat{\bSigma}_{n_w} + \theta \Ib_{p})^{-1} \hat{\bSigma}_{n}^2 \right] + \frac{n^{\train}}{n_w} \cdot \frac{\sigma_{\bepsilon}^2}{n} \cdot  \tr \left[ (\hat{\bSigma}_{n_w} + \theta \Ib_{p})^{-1} \hat{\bSigma}_{n} \right] + o_p(1)
\end{align*}
where $\hat{\bSigma}_{n_w} = \Wb^{\T} \Wb/n_w$. 
Furthermore, we have
\begin{align*}
    \Rom{1}^{(2)}_{\rm sum} =\ & 
    \frac{n^{\train}}{n_w} \cdot \tr \left[ (\hat{\bSigma}_{n_w} + \theta \Ib_{p})^{-1} \Cov(\Xb^{\T} \yb) \right] + o_p(1)
    \\
    =\ & 
    \frac{n^{\train}}{n_w} \cdot \tr \left[ (\hat{\bSigma}_{n_w} + \theta \Ib_{p})^{-1} \left(\Xb^{\T} \yb - n \bSigma \bbeta \right) \left(\Xb^{\T} \yb - n \bSigma \bbeta \right)^{\T} \right] + o_p(1)
    \\
    =\ & \frac{n^{\train}}{n_w} \cdot \frac{\kappa \sigma_{\bbeta}^2}{p} \cdot \tr\left[ (\hat{\bSigma}_{n_w} + \theta \Ib_{p})^{-1} \hat{\bSigma}_{n}^2 \right] + \frac{n^{\train}}{n_w} \cdot \frac{\sigma_{\bepsilon}^2}{n} \cdot \tr\left[ (\hat{\bSigma}_{n_w} + \theta \Ib_{p})^{-1} \hat{\bSigma}_{n} \right]
    \\
    &- \frac{n^{\train}}{n_w} \cdot \frac{\kappa \sigma_{\bbeta}^2}{p} \cdot \tr\left[ (\hat{\bSigma}_{n_w} + \theta \Ib_{p})^{-1} \bSigma^2 \right] + o_p(1)\\
    =\ &\Rom{1}^{(1)}_{\rm sum}- \frac{n^{\train}}{n_w} \cdot \frac{\kappa \sigma_{\bbeta}^2}{p} \cdot \tr\left[ (\hat{\bSigma}_{n_w} + \theta \Ib_{p})^{-1} \bSigma^2 \right] + o_p(1).
\end{align*}
By Lemma \ref{lemma:DE_second_order}, the limit of the left-hand side of Equation \eqref{eqn:num_proof_ridge} is
\begin{align*}
    \left \langle {\sbb^{\valid}}, (\Wb^{\T} \Wb + \theta n_w \Ib_{p})^{-1} \sbb^{\train} \right \rangle/n^{\valid} = \frac{n^{\train}}{n_w} \cdot \frac{\kappa \sigma_{\bbeta}^2}{p} \cdot \tr\left[ \left( \tau_{n_w}(\theta) \bSigma + \theta \Ib_p \right)^{-1} \bSigma^2 \right] + o_p(1),
\end{align*}
which exactly matches with Equation \eqref{eqn:RHS}. 
\end{proof}

The proof sketch of Theorem \ref{thm:reference_panel} provides several key insights. 
First, the limiting behavior of $R^2_{\rm sum, R}(\theta)$ depends solely on the first and second-order moments of $\sbb^{\train}$, rather than its entire distribution.  This observation suggests that resampling-based algorithms may not be sensitive to the selected resampling distribution and 
matching only the first and second moments of $\sbb^{\train}$ to those of $\Xb^{\train}{}^{\T} \yb^{\train}$ is sufficient.
Indeed, the Gaussian distribution may not be the asymptotic distribution of $\Xb^{\train}{}^{\T} \yb^{\train}$ in high dimensions. In Algorithm \ref{alg:sum}, if $\sbb^{\train}$ is sampled by replacing Gaussian $\hb$ to another $p$-dimensional random variables whose coordinates are i.i.d. with mean zero and variance one, Algorithm \ref{alg:sum} may still have robust output and achieve the same performance as Algorithm \ref{alg:ind}. 

Furthermore, our analysis explains why the interdependence between $\sbb^{\train}$ and $\sbb^{\valid}$ does not result in overfitting.
From the moment perspective, $\Rom{1}^{(1)}_{\rm sum}$ can be regarded as a ``first-moment dependence term'' since it arises solely from the first moment, $n^{\train}/n \cdot \Xb^{\train}{}^{\T} \yb^{\train}$, of $\sbb^{\train}$ in Algorithm \ref{alg:sum}. If we match only the first moment of $\sbb^{\train}$ with that of $\Xb^{\train}{}^{\T} \yb^{\train}$, $R^2_{\rm sum, R}(\theta)$ would have inflation compared to $R^2_{\rm ind, R}(\theta)$. However, $\Rom{1}^{(2)}_{\rm sum}$ acts as a ``second-moment correction'' term, originating from the matched second moment $\Cov(\Xb^{\T} \yb)$ in $\sbb^{\train}$. The limiting behavior of $\Rom{1}^{(2)}_{\rm sum}$ precisely cancels out this inflation, ensuring that the resulting $R^2_{\rm sum, R}(\theta)$ aligns exactly with the right-hand side of Equation \eqref{eqn:num_proof_ridge}. Thus, the dependence between $\sbb^{\train}$ and $\sbb^{\valid}$ does not negatively impact prediction accuracy. Our proof sketch focuses on the numerator for illustration. A similar behavior is observed in the denominator, where Lemma~\ref{lemma:second_moment} is additionally required to analyze the involved functionals. The detailed proof of Theorem \ref{thm:reference_panel} is provided in Section \ref{sec:proof_thm_reference_panel} of the supplementary material.

In summary, our analysis demonstrates that even without access to individual-level data, predictive models can be trained and tuned by splitting high-dimensional summary statistics $\Xb^{\T} \yb$ using a resampling-based approach. 
If the first and second moments of the sampling distribution are properly specified, 
the resulting $R^2_{\rm sum, R}(\theta)$, computed from summary statistics, will be asymptotically identical to
$R^2_{\rm ind, R}(\theta)$ based on individual-level data. 
Consequently, the same optimal hyperparameter can be selected, regardless of whether individual data is available. In other words, asymptotically, Algorithm \ref{alg:sum} returns the same prediction model as Algorithm \ref{alg:ind}. 

Figure~\ref{fig:ridge_R2} presents numerical comparisons of the prediction accuracy of $\hat{\bbeta}_{\rm R}(\theta)$ and $\hat{\bbeta}_{\rm R}(\theta)^*$. 
Consistent with our theoretical findings in Theorem~\ref{thm:reference_panel}, the left panel of Figure \ref{fig:ridge_R2} shows that $R^2_{\rm sum, R}(\theta)$ and $R^2_{\rm ind, R}(\theta)$ remain closely matched across all values of the hyperparameter $\theta$.
As indicated by the blue and red dashed vertical lines, the best-performing hyperparameters $\theta^*_{\rm sum, R}$ and $\theta^*_{\rm ind, R}$ exist and are closely aligned.
The right panel of Figure \ref{fig:ridge_R2} further illustrates that the alignment between $\hat{\bbeta}_{\rm R}(\theta_{\rm ind, R}^*)$ and $\hat{\bbeta}_{\rm R}(\theta_{\rm sum, R}^*)^*$ holds across varying  levels of heritability, dimensionality, and signal sparsity. Specifically, resampling-based self-training achieves prediction accuracy comparable to individual-level training in all cases. 
In addition, as expected, prediction accuracy of $\hat{\bbeta}_{\rm R}(\theta_{\rm ind, R}^*)$ and $\hat{\bbeta}_{\rm R}(\theta_{\rm sum, R}^*)^*$ increases with higher heritability, and under the same heritability, both estimators perform better with lower sparsity levels and a lower  $p/n$ ratio.
Overall, Figure \ref{fig:ridge_R2} suggests that resampling-based self-training selects an optimal hyperparameter for the ridge-type estimator that closely matches the one chosen by individual-level data training. 

\begin{figure}[t]
    \centering
    \begin{subfigure}[b]{0.45\textwidth}
        \includegraphics[width=\textwidth]{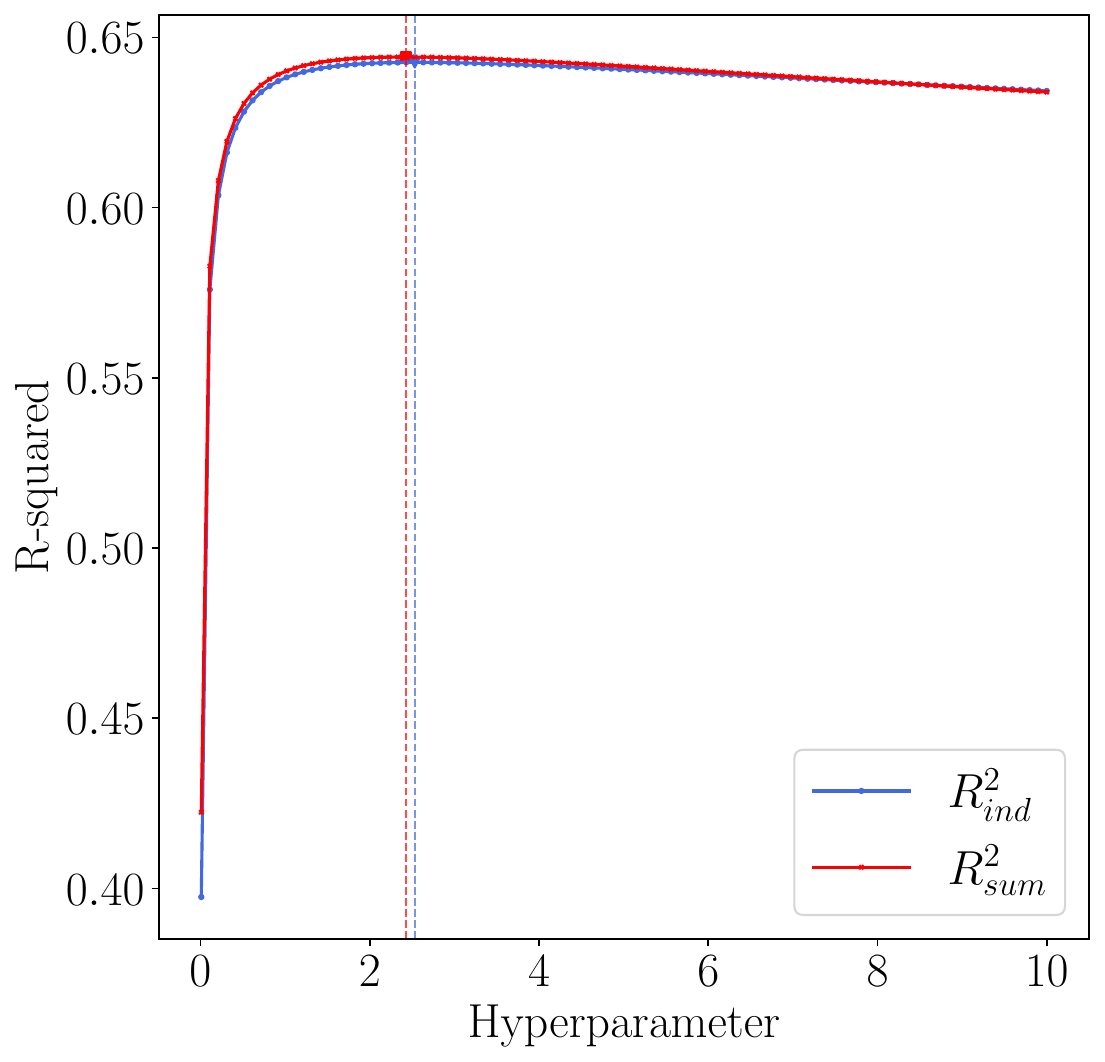}
    \end{subfigure}
    \begin{subfigure}[b]{0.485\textwidth}
        \includegraphics[width=\textwidth]{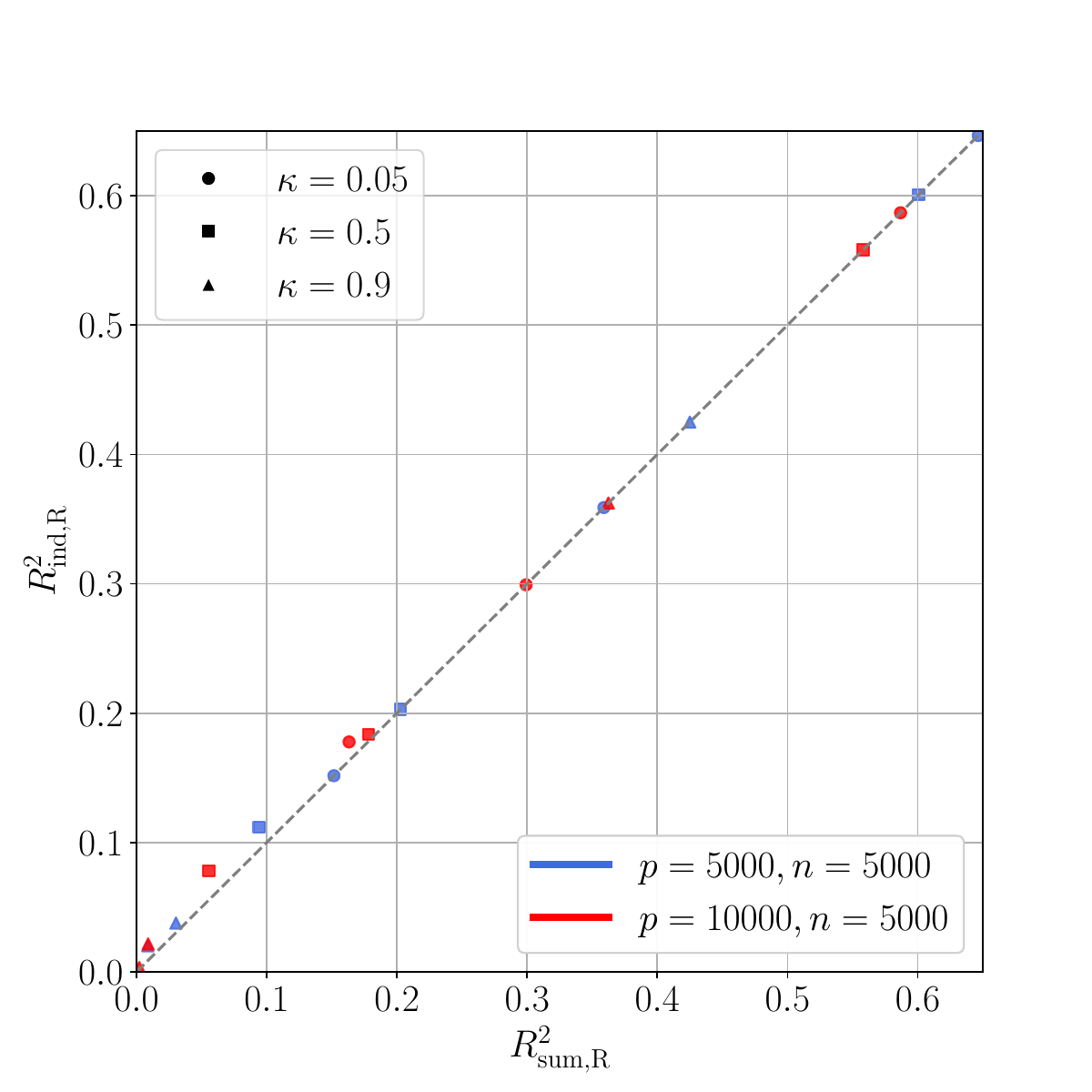}
    \end{subfigure}
    \caption{
    \textbf{Numerical comparison of prediction accuracy between resampling-based and individual-level training for the ridge-type estimator across various hyperparameter values, heritability levels, dimensionalities, and sparsity levels.} 
    In this numerical analysis, we assume that $\bSigma$ is a block-wise diagonal matrix, $\bSigma = \diag\{\bSigma_1, \bSigma_2, \dots, \bSigma_{n_{\rm block}}\}$, where each block $\bSigma_i$ follows an AR(1) process with correlation $\rho = 0.9$, as detailed in Equation~\eqref{eqn:AR1} and Section~\ref{sec:numer}. 
    \textbf{Left:} 
    %Validation $R^2$ for ridge estimators, $\hat{\bbeta}_{\rm R}(\theta)^*$ and $\hat{\bbeta}_{\rm R}(\theta)$. 
    The out-of-sample $R^2$ values, denoted as $R_{\rm sum, R}^2(\theta)$ and $R_{\rm ind, R}^2(\theta)$, are computed using Algorithm~\ref{alg:sum} and Algorithm~\ref{alg:ind}, respectively. 
    We evaluate $R_{\rm sum, R}^2(\theta)$ and $R_{\rm ind, R}^2(\theta)$ across a range of hyperparameter values, highlighting the best-performing hyperparameters, $\theta_{\rm sum, R}^*$ and $\theta_{\rm ind, R}^*$, with dashed vertical lines. 
    The parameters are set as follows: $h^2 = 0.8$, $n = p = 5000$, $\kappa = 0.1$, and $n_w = 1000$. 
    \textbf{Right:} The out-of-sample $R^2$ for $\hat{\bbeta}_{\rm R}(\theta_{\rm sum, R}^*)^*$ and $\hat{\bbeta}_{\rm R}(\theta_{\rm ind, R}^*)$, where $\theta_{\rm sum, R}^*$ and $\theta_{\rm ind, R}^*$ denote the best-performing hyperparameters selected by Algorithm~\ref{alg:sum} and Algorithm~\ref{alg:ind}, respectively. 
    We compare the prediction accuracy across varying levels of heritability, dimensionality, and sparsity.
    %and highlight the reference line $y = x$ using a dashed diagonal line. 
    The parameters are set as follows: $h^2 \in \{2/5, 1/2, 2/3, 4/5\}$, $p \in \{5000, 10000\}$, $n = 5000$, $\kappa \in \{0.05, 0.5, 0.9\}$, $n_{\rm block} = 20$, and $n_w = 1000$. 
    }
  \label{fig:ridge_R2}
\end{figure}

Corollary \ref{cor:reference_panel} presents the results for the special case $\bSigma = \Ib_{p}$, where we obtain closed-form expressions for $R^2_{\rm sum, R}(\theta)$. 

\begin{corollary}\label{cor:reference_panel}
    Under the same conditions as in Theorem \ref{thm:reference_panel} and $\bSigma = \Ib_{p}$,
    we have
    \begin{align*}
        R^2_{\rm sum, R}(\theta) = \frac{\theta + \tau_{n_w}(\theta)}{\rho_{n_w}(\theta) + 1} \cdot \frac{n_w}{p/h^2 + n^{\train}} + o_p(1),
    \end{align*}
    where $\tau_{n_w}(\theta)$ and $\rho_{n_w}(\theta)$ are simplified to 
    \begin{align*}
    \tau_{n_w}(\theta) = \frac{1 - \theta - \gamma_w + \sqrt{(1 - \theta - \gamma_w)^2 + 4 \theta}}{2} \quad 
    \mbox{and} \quad 
    \rho_{n_w}(\theta) = 
    \frac{1 + \gamma_w + \theta - \sqrt{(1 - \theta - \gamma_w)^2 + 4 \theta}}{2 \sqrt{(1 - \theta - \gamma_w)^2 + 4 \theta}},
    \end{align*}
    respectively. 
\end{corollary}

\subsection{Resampling-based marginal thresholding}
\label{sec:marginal}
%Another widely used linear estimator in genetic data prediction is marginal thresholding. 
One of the most popular approaches for genetic prediction is ``clumping and thresholding" \citep{choi2020tutorial}.  After removing genetic variants with high correlations with other predictors, this marginal thresholding approach ranks the genetic variants by their $P$-values (or equivalent statistics) and selects the best-performing $P$-value threshold based on prediction accuracy in the validation data. 
In prediction models with non-sparse signals, there is often a trade-off in marginal screening when selecting an optimal subset of variables to maximize the prediction accuracy \citep{zhao2019genetic}. Selecting too many variables can reduce the out-of-sample $R^2$ by introducing noise and overfitting, as variants without predictive power are included. Conversely, selecting too few variables can also lower the out-of-sample $R^2$ by excluding important predictors. In resampling-based
pseudo-training, we aim to select the $P$-value threshold without access to individual-level validation data. 

When individual-level data are available, the marginal thresholding estimator can be formulated as follows
\begin{align*}
    \hat{\bbeta}_{\rm M}(\Theta) = \left[ \Xb^{\train}{}^{\T}\yb^{\train} \right]_{\Theta} :=
    \begin{cases}
    \left[ \Xb^{\train}{}^{\T}\yb^{\train} \right]_{i}& \text{if}\ i \in \Theta,
    \\
    0& \text{otherwise},
    \end{cases}
\end{align*}
where the hyperparameter $\Theta \subset [p]$ denotes the indexes of predictors, a subset of $[p]$, that to be selected by the marginal screening $P$-value threshold. 
When only summary statistics are available, we can similarly define $\hat{\bbeta}_{\rm M}(\Theta)^*$ by replacing $\Xb^{\train}{}^{\T}\yb^{\train}$ by $\sbb^{\train}$
\begin{align}\label{eqn:marginal}
    \hat{\bbeta}_{\rm M}(\Theta)^* = \left[ \sbb^{\train} \right]_{\Theta} :=
    \begin{cases}
    \left[ \sbb^{\train} \right]_{i}& \text{if}\ i \in \Theta,
    \\
    0& \text{otherwise}.
    \end{cases}
\end{align}
%Note that different from Section \ref{sec:ref_ridge},of $\hat{\bbeta}_{\rm G}(\Theta)$
Here we use $\Theta$ to emphasize that it is a vector parameter.
Considering the general estimator defined in Equation~\eqref{eqn:estimator}, we can rewrite 
$\hat{\bbeta}_{\rm M}(\Theta)=\Ab(\Theta) \Xb^{\train}{}^{\T} \yb^{\train}$ with  
\begin{align} \label{eqn:special_linear_MG}
    \Ab(\Theta)_{i,i} = 1\ \text{if}\ i \in \Theta, \quad \Ab(\Theta)_{i,j} = 0\ \text{otherwise}.
\end{align}
We abbreviate $\Ab(\Wb^{\T} \Wb, \Theta)$ as $\Ab(\Theta)$ since $\hat{\bbeta}_{\rm M}(\Theta)$ and $\hat{\bbeta}_{\rm M}(\Theta)^*$ do not depend on $\Wb^{\T} \Wb$. 
In Algorithm \ref{alg:ind} of the
supplementary material, marginal screening selects the subset of variables $\Theta$ that maximize the out-of-sample $R^2$ as follows
\begin{align*}
    \Theta_{\rm ind, M}^* = \arg\max_{\Theta \subset [p]} R^2_{\rm ind, M} (\Theta) = \arg\max_{\Theta \subset [p]} \frac{n^{\valid}}{\|\yb^{\valid}\|_{2}^2} \cdot \frac{\left\langle \Xb^{\valid}{}^{\T}\yb^{\valid}, \left[ \Xb^{\train}{}^{\T}\yb^{\train} \right]_{\Theta} \right\rangle^2}{n^{\valid} \cdot \| \left[ \Xb^{\train}{}^{\T}\yb^{\train} \right]_{\Theta} \|_{\bSigma}^2}. 
\end{align*}
%In another words, marginal screening is equivalent to Algorithm \ref{alg:ind} by defining $\Ab(\Theta)$ as in Equation \eqref{eqn:special_linear_MG}.

When only summary statistics $\Xb^{\T} \yb$ are available, one aims to find $\Theta_{\rm sum, M}^*$ to maximize $R^2_{\rm sum, M} (\Theta)$ in Algorithm \ref{alg:sum}, with $\Ab(\Theta)$ being defined as in Equation \eqref{eqn:special_linear_MG}.
The following theorem provides the asymptotic results of prediction accuracy, $R^2_{\rm sum, M} (\Theta)$, and suggests equivalent performance between Algorithm \ref{alg:sum} and Algorithm \ref{alg:ind} in training the marginal thresholding estimator. 
 
\begin{theorem} \label{thm:marginal_screen}
    Consider any random sequence $\{\bbeta, \bepsilon, \Xb\}_{(p, n) \in \NN^2}$ satisfying Conditions \ref{cond-np-ratio}-\ref{cond-eps} and any $\Theta \subset [p]$. For $\hat{\bbeta}_{\rm M}(\Theta)^*$ defined in Equation \eqref{eqn:marginal}, the out-of-sample $R^2$ is 
    $$R^2_{\rm sum, M}(\Theta) = \frac{n^{\valid}}{\|\yb^{\valid}\|_{2}^2} \cdot \frac{n^{\train}}{p} \cdot \kappa \sigma_{\bbeta}^2 \cdot \frac{\left( \tr \left[ \Ab(\Theta) \bSigma^2 \right]\right)^2}{\tr \left(\bSigma \right) \cdot \tr \left[\Ab(\Theta) \bSigma \right]/h^2 + n^{\train} \cdot \tr \left[\Ab(\Theta) \bSigma^2 \right]} + o_p(1),$$
    where $\Ab(\Theta)$ is defined in Equation \eqref{eqn:special_linear_MG}. 
    %\bxz{Do we need to require the size of $\Theta$ to be proportional to $p$? }
    %\bxz{Use $h^2$ to replace $\sigma_{\bepsilon}^2$, and maybe some $\sigma_{\bbeta}^2$?}\bxs{done}
    Furthermore, for any $\Theta \subset [p]$, we have  
    %\bxz{What is the parameter space of $\Theta$?}\bxs{done}
    \begin{align*}
        R^2_{\rm sum, M}(\Theta)/R^2_{\rm ind, M}(\Theta) \overset{p}{\to} 1. 
    \end{align*}
\end{theorem}

\begin{figure}[t]
    \centering
    \begin{subfigure}[b]{0.45\textwidth}
        \includegraphics[width=\textwidth]{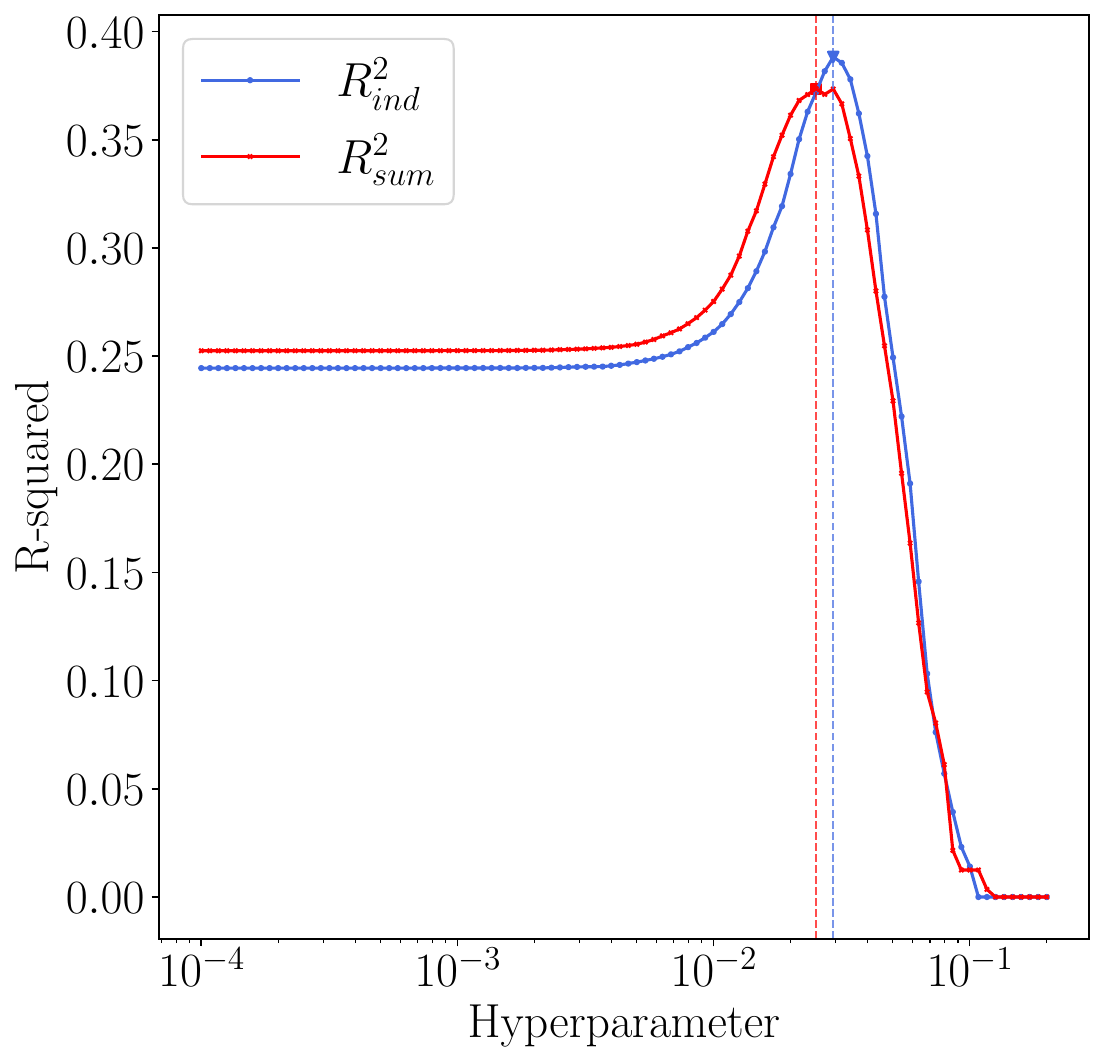}
    \end{subfigure}
    \begin{subfigure}[b]{0.485\textwidth}
        \includegraphics[width=\textwidth]{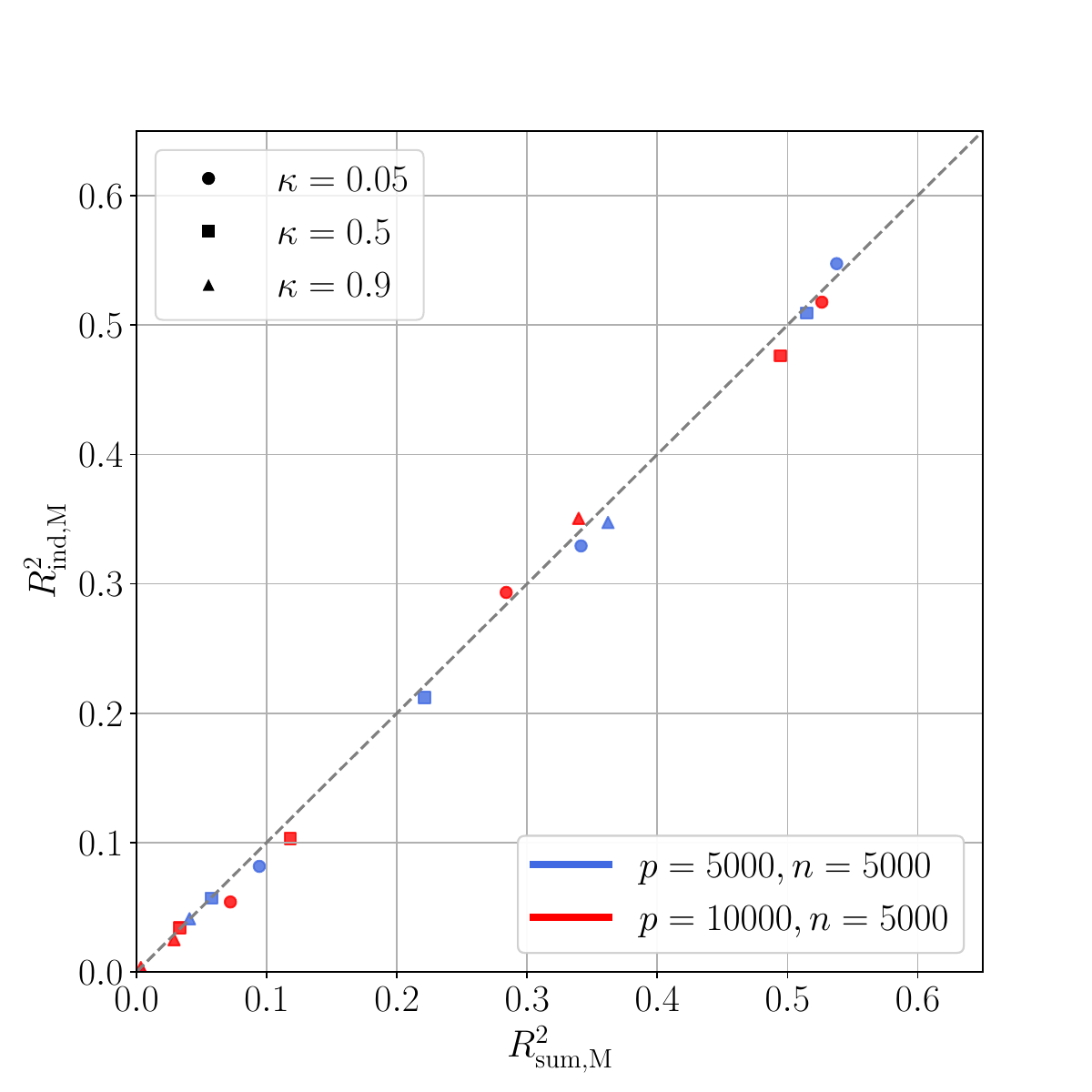}
    \end{subfigure}
    \caption{
    \textbf{Numerical comparison of prediction accuracy between resampling-based and
individual-level training for the marginal thresholding estimator across various hyperparameter values,
heritability levels, dimensionalities, and sparsity levels.}
In this numerical analysis, we assume that $\bSigma$ is a block-wise diagonal matrix, $\bSigma = \diag\{\bSigma_1, \bSigma_2, \dots, \bSigma_{n_{\rm block}}\}$, where each block $\bSigma_i$ follows an AR(1) process with correlation $\rho = 0.9$, as detailed in Equation~\eqref{eqn:AR1} and Section~\ref{sec:numer}.
\textbf{Left:} 
    The out-of-sample $R^2$ values, denoted as $R_{\rm sum, M}^2(\Theta)$ and $R_{\rm ind, M}^2(\Theta)$, are computed using Algorithm~\ref{alg:sum} and Algorithm~\ref{alg:ind}, respectively. 
    We evaluate $R_{\rm sum, M}^2(\Theta)$ and $R_{\rm ind, M}^2(\Theta)$ across a range of hyperparameter values, highlighting the best-performing hyperparameters, $\Theta_{\rm sum, M}^*$ and $\Theta_{\rm ind, M}^*$, with dashed vertical lines. 
    The parameters are set as follows: $h^2 = 0.8$, $n = p = 5000$, and $\kappa = 0.1$. 
    \textbf{Right:} The out-of-sample $R^2$ for $\hat{\bbeta}_{\rm M}(\Theta_{\rm sum, M}^*)^*$ and $\hat{\bbeta}_{\rm M}(\Theta_{\rm ind, M}^*)$, where $\Theta_{\rm sum, M}^*$ and $\Theta_{\rm ind, M}^*$ the best-performing hyperparameters selected by Algorithm~\ref{alg:sum} and Algorithm~\ref{alg:ind}, respectively. 
    We compare the prediction accuracy across varying levels of heritability, dimensionality, and sparsity.
    %and highlight the reference line $y = x$ using a dashed diagonal line. 
    The parameters are set as follows: $h^2 \in \{2/5, 1/2, 2/3, 4/5\}$, $p \in \{5000, 10000\}$, $n = 5000$, $\kappa \in \{0.05, 0.5, 0.9\}$, and $n_{\rm block} = 20$.
    }
    \label{fig:marginal_R2}
\end{figure}

The proof of Theorem \ref{thm:marginal_screen} is presented in Section \ref{sec:proof_thm_marginal_screen} of the
supplementary material, following similar principles to those used in the proof of Theorem \ref{thm:reference_panel}.
Theorem \ref{thm:marginal_screen} shows that the marginal thresholding estimator $\hat{\bbeta}_{\rm M}(\Theta)^*$, based solely on summary statistics, can select the same subset of predictors and achieve the same out-of-sample $R^2$ as when individual-level data is used. 

Similar to Figure \ref{fig:ridge_R2}, Figure \ref{fig:marginal_R2} presents numerical illustrations for the marginal thresholding estimator.
The left panel of Figure \ref{fig:marginal_R2} supports Theorem \ref{thm:marginal_screen}, demonstrating that the out-of-sample $R^2$ pattern obtained from resampling-based self-training closely aligns with that of individual-level training.
Additionally, the right panel of Figure \ref{fig:marginal_R2} illustrates that models selected by resampling-based marginal thresholding achieve prediction accuracy comparable to individual-level marginal thresholding across a wide range of scenarios.

% \bxs{Figure \ref{fig:marginal_R2} presents the counterpart results of Figure \ref{fig:ridge_R2}, focusing on the marginal thresholding estimator.
% Furthermore, the left panel of Figure \ref{fig:marginal_R2} supports the theoretical findings in Theorem \ref{thm:marginal_screen}, demonstrating that the validation $R^2$ obtained from summary data-based self-training closely aligns with that of individual-level training.
% Vertical dashed lines highlight the best-performing hyperparameters for both methods which are nearly identical.
% Additionally, the right panel of Figure \ref{fig:marginal_R2} demonstrates that resampling-based marginal thresholding achieves comparable prediction accuracy when compared to individual-level marginal thresholding, which are consistent with Theorem \ref{thm:marginal_screen}.}

\subsection{Asymptotic results for general linear estimators}
\label{sec:general_lin}

After analyzing the two concrete examples popular in PRS applications,
%, $\hat{\bbeta}_{\rm R}(\theta)$ and $\hat{\bbeta}_{\rm M}(\Theta)$, = \Ab(\Wb^{\T} \Wb, \Theta) \Xb^{\train}{}^{\T} \yb^{\train}
we now present a general theorem for the class of linear estimators $\hat{\bbeta}_{\rm G}(\Theta)$  defined in Equation~\eqref{eqn:estimator}, with a parameter vector $\Theta$.
When only summery statistics are available, we define the $\hat{\bbeta}_{\rm G}(\Theta)^*$ by replacing $\Xb^{\train}{}^{\T} \yb^{\train}$ to $\sbb^{\train}$
\begin{align} \label{eqn:general}
    \hat{\bbeta}_{\rm G}(\Theta)^* = \Ab(\Wb^{\T} \Wb, \Theta) \sbb^{\train}.
\end{align}
%Starting from this section, we will assume $\Theta$ as a parameters vector.
Our analysis of this general estimator indicates that it is not coincidental that Algorithm \ref{alg:sum} incurs no additional prediction accuracy cost compared to Algorithm \ref{alg:ind} for both $\hat{\bbeta}_{\rm R}(\theta)$ and $\hat{\bbeta}_{\rm M}(\Theta)$. 
The key lies in the component $\Ab(\Wb^{\T} \Wb, \Theta)$. 
Briefly, as long as $\Ab(\Wb^{\T} \Wb, \Theta)$ satisfies certain first and second-order deterministic equivalent conditions, Algorithm \ref{alg:sum} can achieve the same asymptotic prediction accuracy and select the same hyperparameter as Algorithm \ref{alg:ind}, without requiring individual-level data.
We summarize these results in the following theorem.

\begin{theorem} \label{thm:general_linear}%[Master Theorem]
Consider any random sequence $\{\bbeta, \bepsilon, \Xb, \Wb\}_{(p, n, n_w) \in \NN^3}$ satisfying Conditions \ref{cond-np-ratio}-\ref{cond-W} and any parameter vector $\Theta$.
Suppose the sequence $\{\Ab(\Wb^{\T}\Wb, \Theta)\}$ satisfies
    \begin{align} \label{eqn:cond-A}
        n_w \cdot \Ab(\Wb^{\T}\Wb, \Theta) \asymp \Db(\Theta) \quad\mbox{and} \quad n_w^2 \cdot \Ab(\Wb^{\T}\Wb, \Theta)^{\T} \bSigma \Ab(\Wb^{\T}\Wb, \Theta) \asymp \Eb(\Theta),
    \end{align}
    for any $\Theta$ and some matrices $\Db$ and $\Eb$ depending on $\bSigma$ and $\Theta$ only. Then for $\hat{\bbeta}_{\rm G}(\Theta)^*$ defined in Equation \eqref{eqn:general}, 
    the out-of-sample $R^2$ is 
    $$R^2_{\rm sum, G}(\Theta) = \frac{n^{\valid}}{\|\yb^{\valid}\|_{2}^2} \cdot \frac{n^{\train}}{p} \cdot \kappa \sigma_{\bbeta}^2 \cdot \frac{\left( \tr \left[ \Db(\Theta) \bSigma^2 \right]\right)^2}{\tr \left(\bSigma \right) \cdot \tr \left[\Eb(\Theta) \bSigma \right]/h^2 + n^{\train} \cdot \tr \left[\Eb(\Theta) \bSigma^2 \right]} + o_p(1).$$
    Furthermore, for any $\Theta$, we have 
    \begin{align*}
        R^2_{\rm sum, G}(\Theta)/R^2_{\rm ind, G}(\Theta) \overset{p}{\to} 1. 
    \end{align*}
\end{theorem}

The proof of Theorem \ref{thm:general_linear} is presented in Section \ref{sec:proof_thm_general_linear} of the
supplementary material.
%is similar to the proof of Theorem \ref{thm:reference_panel}. 
The $\hat{\bbeta}_{\rm G}(\Theta)^*$ in Equation \eqref{eqn:general} replaces $(\Wb^{\T} \Wb + \theta \Ib_{p})^{-1}$ in $\hat{\bbeta}_{\rm R}(\theta)^*$ 
%in Equation \eqref{eqn:ref-panel-ridge-sum} 
by a general matrix $\Ab(\Wb^{\T}\Wb, \Theta)$, and  
Theorem \ref{thm:general_linear} demonstrates that if the first and second-order components of $\Ab(\Wb^{\T}\Wb, \Theta)$ have well-defined limits, then $\hat{\bbeta}_{\rm G}(\Theta)^*$ will exhibit good properties in resampling-based self-training. 
Equation \eqref{eqn:cond-A}  specifies these conditions, requiring the existence of first and second-order deterministic equivalents concerning a sequence of matrices 
depending on $\bSigma$ and $\Theta$.
As demonstrated in our proof, Equation \eqref{eqn:cond-A} is a crucial condition for the concentration inequalities in Lemma B.26 of \cite{bai2010spectral} to hold, thereby making the limiting behavior of the estimators traceable. 
For concrete estimators $\hat{\bbeta}_{\rm R}(\theta)$ and $\hat{\bbeta}_{\rm M}(\Theta)$, we have shown that such condition holds for their corresponding $\Ab(\Wb^{\T}\Wb, \Theta)$. 
For example, to prove Theorem \ref{thm:reference_panel} for $\hat{\bbeta}_{\rm R}(\theta)$, we provide the deterministic equivalents of
$$n_w \cdot \left( \Wb^{\T}\Wb + n_w \theta \Ib_{p} \right)^{-1}\quad \mbox{and} \quad n_w^2 \cdot \left( \Wb^{\T}\Wb + n_w \theta \Ib_{p} \right)^{-1}\bSigma\left( \Wb^{\T}\Wb + n_w \theta \Ib_{p} \right)^{-1},$$
which are given in Lemma \ref{lemma:DE_second_order}.

\section{Summary data-based ensemble learning}
\label{sec:emsemble}
Ensemble learning aims to combine multiple models to improve prediction accuracy and has been widely applied in genetic prediction \citep{pain2021evaluation, yang2022PGS,jin2025pennprs}. Traditional ensemble learning requires access to individual-level data to combine models. In this section, we explore whether the summary data-based model training framework can be extended to perform ensemble learning. We consider ensemble learning as a linear combination of $k$ different general linear estimators defined in Equation~\eqref{eqn:estimator} \citep{opitz1999popular, chen2024fast}. 
%: $\hat{\bbeta}_{\rm E}(\Theta) = \sum_{j=1}^{k} \omega_{j} \Ab_{j}(\Wb^{\T}\Wb, \Theta_{j}) \Xb^{\train}{}^{\T} \yb^{\train}$
Specifically, we can model the ensemble estimator 
%$\hat{\bbeta}_{\rm E}(\Theta)$ 
as 
$\hat{\bbeta}_{\rm E}(\Theta) = \Ab(\Wb^{\T}\Wb, \Theta)\Xb^{\train}{}^{\T} \yb^{\train}$, where 
$\Ab(\Wb^{\T}\Wb, \Theta) = \sum_{j=1}^{k} \omega_{j} \Ab_{j}(\Wb^{\T}\Wb, \Theta_j)$
and 
% the following
% \begin{align*}
%     \hat{\bbeta}_{\rm E}(\Theta) = \Ab(\Wb^{\T}\Wb, \Theta) \Xb^{\train}{}^{\T} \yb^{\train}, \quad \Ab(\Wb^{\T}\Wb, \Theta) = \sum_{j=1}^{k} \omega_{j} \Ab_{j}(\Wb^{\T}\Wb, \Theta_j),
% \end{align*}
% where 
hyperparameter vector $\Theta = \{\omega_{j}, \Theta_{j}\}_{1 \leq j \leq k}$ for some constant $k$. 
% In the following context, we may abbreviate $\Ab(\Wb^{\T} \Wb, \Theta)$ as $\Ab(\Theta)$.
When only summary-level data is available, we similarly define $\hat{\bbeta}_{\rm E}(\Theta)^*$ by replacing $\Xb^{\train}{}^{\T} \yb^{\train}$ with $\sbb^{\train}$
\begin{align}\label{eqn:ensemble}
    \hat{\bbeta}_{\rm E}(\Theta)^* = \sum_{j=1}^{k} \omega_{j} \Ab_{j}(\Wb^{\T}\Wb, \Theta_{j}) \sbb^{\train}.
\end{align}
% \bxz{Will you treat $\omega$ as parameter as well? If so, we may need something like $\Ab(\Wb^{\T}\Wb, \Theta,\Omega)$ }\bxs{done}

When individual-level data are available, Algorithm \ref{alg:ensemble_ind} 
%in Appendix~\ref{sec:algo}
of the supplementary material 
details the conventional ensemble learning steps to maximize $R^2_{\rm ind, E}(\Theta)$. Algorithm \ref{alg:ensemble_sum} outlines the self-training approach for ensemble learning with summary data only, where 
%Unlike Algorithm \ref{alg:sum}, 
we use the same $(\sbb^{\train},\sbb^{\valid})$ to train $k$ different estimators simultaneously, determining both the weights $\omega_{j}$ and hyperparameter $\Theta_{j}$ that maximize $R^2_{\rm sum, E}(\Theta)$.

\begin{algorithm}
\caption{Summary data-based ensemble learning}\label{alg:ensemble_sum}
\begin{algorithmic}
\Require Summary data $\Xb^{\T} \yb$, $\Wb^{\T} \Wb$, and the hyperparameter $\Theta = \{\omega_{j}, \Theta_{j}\}_{j=1}^{k}$. 
\vspace{1mm}

\State $\hb \gets \cN(0,\Ib_p),$  \hfill 		\texttt{//} Sample $p$-dimension standard Gaussian random variable.
\vspace{1mm}

\State $\sbb^{\train} \gets \frac{n^{\train}}{n} \Xb^{\T} \yb + \sqrt{\frac{n^{\train} (n - n^{\train})}{n^2}} \Cov(\Xb^{\T} \yb)^{1/2} \hb,$  \hfill 		\texttt{//} Construct $\sbb^{\train}$ for training.
\vspace{1mm}

\State $\sbb^{\valid} \gets \Xb^{\T} \yb - \sbb^{\train},$  \hfill 		\texttt{//} Construct $\sbb^{\valid}$ for validation.
\vspace{1mm}

\State $\hat{\bbeta}_{\rm E}(\Theta)^* \gets \sum_{j=1}^{k} \omega_{j} \Ab_{j}(\Wb^{\T} \Wb, \Theta_{j}) \sbb^{\train},$  \hfill 		\texttt{//} Obtain the estimator.
\vspace{1mm}

\State $R_{\rm sum, E}^2(\Theta) \gets ({n^{\valid}}\left/{\|\yb^{\valid}\|_{2}^2})\right. \cdot {\left \langle {\sbb^{\valid}}, \hat{\bbeta}_{\rm E}(\Theta)^{*} \right \rangle^{2} } \left/({n^{\valid} \cdot \| \hat{\bbeta}_{\rm E}(\Theta)^{*} \|_{\bSigma}^{2}}), \right.$\vspace{1mm}

\Statex \hfill \texttt{//} Compute the $R^2_{\rm sum, E}(\Theta)$ in \eqref{eqn:R2_sum_def}.
\vspace{1mm}

\State $\Theta^{*}_{\rm sum, E} \gets \max_{\Theta} R^2_{\rm sum, E}(\Theta),$  \hfill 		\texttt{//} Choose the best-performing hyperparameter.
%\bxz{Also maximizing $\omega$, somehow?}\bxs{done}
\vspace{1.5mm}

\Return $R_{\rm sum, E}^2(\Theta)$ and $ \Theta^{*}_{\rm sum, E}$.

\end{algorithmic}
\end{algorithm}

Theorem \ref{thm:ensemble} shows that, asymptotically, Algorithm \ref{alg:ensemble_sum} selects the same hyperparameter as Algorithm \ref{alg:ensemble_ind}.

\begin{theorem}\label{thm:ensemble}
    Consider any random sequence $\{\bbeta, \bepsilon, \Xb, \Wb\}_{(p, n, n_w) \in \NN^3}$ satisfying Conditions \ref{cond-np-ratio}-\ref{cond-W} and any parameter vector $\Theta = \{\omega_{j}, \Theta_{j}\}_{j=1}^{k}$. Suppose the sequence $\{\Ab_{j}(\Wb^{\T}\Wb, \Theta_{j})\}$ satisfies
    \begin{align*}
        n_w \cdot \Ab_{j}\left( \Wb^{\T}\Wb, \Theta_{j} \right) \asymp \Db_{j}(\Theta) \quad \text{and} \quad n_w^2 \cdot \Ab_{j}(\Wb^{\T}\Wb, \Theta_{j})^{\T} \bSigma \Ab_{j}(\Wb^{\T}\Wb, \Theta_{j}) \asymp \Eb_{j}(\Theta_{j}), 
    \end{align*}
    with some matrices $\Db_{j}(\Theta_{j})$ and $\Eb_{j}(\Theta_{j})$ depending on $\bSigma$ and $\Theta_{j}$ only.
    Then for any $\Theta$ and the $\hat{\bbeta}_{\rm E}(\Theta)^*$ defined in Equation \eqref{eqn:ensemble}, the out-of-sample $R^2$ is
    $$R^2_{\rm sum, E}(\Theta) = \frac{n^{\valid}}{\|\yb^{\valid}\|_{2}^2} \cdot \frac{n^{\train}}{p} \cdot \kappa \sigma_{\bbeta}^2 \cdot \frac{\left( \tr \left[ \Db(\Theta) \bSigma^2 \right]\right)^2}{\tr \left(\bSigma \right) \cdot \tr \left[\Eb(\Theta) \bSigma \right]/h^2 + n^{\train} \cdot \tr \left[\Eb(\Theta) \bSigma^2 \right]} + o_p(1),$$
    where
    \begin{align*}
        \Db(\Theta) = \sum_{j=1}^{k} \omega_{j} \Db_{j}(\Theta_{j}) \quad \mbox{and} \quad 
        \Eb(\Theta) = \sum_{i \neq j} \omega_{i} \omega_{j} \Db_{i}(\Theta_{i}) \bSigma \Db_{j}(\Theta_{j}) + \sum_{j=1}^{k} \omega_{j}^2 \Eb_{j}(\Theta_{j}). 
    \end{align*}
    Furthermore, for any $\Theta = \{\omega_{j}, \Theta_{j}\}_{j=1}^{k}$, we have 
     \begin{align*}
        R^2_{\rm sum, E}(\Theta)/R^2_{\rm ind, E}(\Theta) \overset{p}{\to} 1. 
    \end{align*} 
    %\bxz{You need to emphasize to have the same $\omega$ as well? }\bxs{done.}
\end{theorem}

% \bxz{Add discussions of the results. Specifically, similar to previous theorems, we need to tell people the key insights of why the two are equivalent, this is related to the proof steps.}\bxs{The proof is very similar to previous sections.}

Theorem \ref{thm:ensemble} demonstrates that, even when combining multiple estimators, one can rely solely on summary statistics to select the same best-performing parameter and achieve the same out-of-sample $R^2$ as if individual-level data were used. Additional numerical illustrations using real data are provided in Section~\ref{sec:numer}. 

%The proof of Theorem \ref{thm:ensemble} is very similar to previous sections, therefore omitted. 

\section{Multi-ancestry data resources}
\label{sec:multi}
Improving the accuracy of genetic prediction across diverse populations is a key area of interest.
Since most existing genetic prediction models have historically been trained primarily on populations of European ancestry, researchers are now actively developing advanced methods to improve the prediction of complex traits and diseases in non-European populations \citep{zhang2023new,zhang2024ensemble, jin2024mussel}.
Many proposed methods improve predictive power by integrating summary-level data from diverse populations to develop ancestry-specific models tailored to each population \citep{kachuri2024principles}.

In this section, we present and analyze Algorithm \ref{alg:multi_sum}, which uses summary statistics from multi-ancestry data sources to self-train a population-specific model without requiring access to individual-level data. 
We consider $K$ datasets collected from $K$ different populations (such as European, East Asian, and African American), each with an observation vector $\yb_{j}$ and a design matrix $\Xb_{j}$ that satisfy the linear model $\yb_j = \Xb_j \bbeta_{j} + \epsilon_j$ for $j=1, \ldots, K$. Without loss of generality, we assume the target population for model development is the first population, $j = 1$. Thus, we aim to estimate the target coefficient $\bbeta_1$ using data from populations $1$ through $K$. 
Let $\Wb_{j}$ represent the reference panel dataset from the same cohort as the population $j$.  
For simplicity, we assume each population has the same sample size $n$ and each  reference panel has the same sample size $n_w$, although our analysis can be readily extended to cases with different sample sizes.
We make similar assumptions on the datasets as in Section \ref{sec:model}, which are summarized in Condition~\ref{cond-X-cross}.

\begin{condition}\label{cond-X-cross}
    For $j=1, \ldots, K$,
    we assume $\Xb_j = \Xb_0 \bSigma_{j}^{1/2} \in \RR^{n \times p}$ and $\Wb_{j} = \Wb_0 \bSigma_{j}^{1/2} \in \RR^{n_w \times p}$. Entries of $\Xb_0$ are real-value i.i.d. random variables with mean zero, variance one, and a finite $4$-th order moment. The $\bSigma_{j}$ are $p \times p$ population level deterministic positive definite matrices with uniformly bounded eigenvalues. Specifically, we have $0 < c \leq \lambda_{\min}(\bSigma_{j}) \leq \lambda_{\max}(\bSigma_{j}) \leq C$ for all $p$ and all $j$. Moreover, we assume each row of $\Xb_j$ and $\Wb_{j}$ are all jointly independent. 
\end{condition}

% To establish the linear model framework, we adopt the following conditions on $\bbeta$. 
% Let $F(0, \Vb)$ denote a generic distribution with
% mean zero, variance $\Vb$, and finite $4-$th order moments. 

Let $\Fb(0, \Vb)$ denote a multi-dimensional generic distribution with
mean zero, covariance $\Vb$, and each coordinate has a finite $4$-th order moment. We introduce the following conditions on genetic effects and random errors, which extend the Conditions~\ref{cond-beta} and \ref{cond-eps} from one population to $K$ populations.

\begin{condition}\label{cond-beta-cross}
For $1 \leq i,j \leq K$, there exist sparsity level $0 \leq \kappa_{j} \leq 1$ and $\sigma_{i,j}^2 \geq 0$ such that
the joint distribution of $(\bbeta_{1}, \cdots, \bbeta_{K})$ is given by
\begin{align*}
    \begin{pmatrix}
        \bbeta_{1}\\
        \cdots\\
        \bbeta_{K}
    \end{pmatrix} \sim \begin{pmatrix}
        (1 - \kappa_{1}) \delta(\bm 0)\\
        \cdots\\
        (1 - \kappa_{K}) \delta(\bm 0)
    \end{pmatrix} + \begin{pmatrix}
        \kappa_{1}\\
        \cdots\\
        \kappa_{K}
    \end{pmatrix} \odot \Fb \left[ \begin{pmatrix}
        \bm 0\\
        \cdots\\
        \bm 0
    \end{pmatrix}, p^{-1} \begin{pmatrix}
        \sigma_{1, 1}^2 \Ib_{p} & \cdots & \sigma_{1, K}^2 \Ib_{p}
        \\
        \cdots 
        \\
        \sigma_{1, K}^2 \Ib_{p} & \cdots & \sigma_{K, K}^2 \Ib_{p}
    \end{pmatrix} \right],
\end{align*}
where $\odot$ is the Hadamard product between two vectors defined to as $(\ab \odot \bb)_{l} = \ab_{l} \cdot \bb_{l}$.
Furthermore, we assume that $\sum_{l=1}^p (\bbeta_{j})_{l}^2 \to \kappa_{j} \sigma_{j,j}^2$ as $n, p \to \infty$ with $p/n \to \gamma >0$. 
\end{condition}

\begin{condition}\label{cond-eps-cross} 
For $1 \leq j \leq K$, random errors $\bepsilon_j$s are independent random vectors, and each entry has distribution
%\bxz{Need to change to $K$ populations?}\bxs{done}
\begin{align*}
    \bepsilon_{j,i} \overset{i.i.d.}{\sim} F(0, \sigma_{\bepsilon_{j}}^2), \quad \mbox{for} \quad 1 \leq j \leq K \quad \mbox{and} \quad 1 \leq i \leq n.
\end{align*}
\end{condition}

We also extend the definition of heritability to encompass multiple ancestries.
%\ref{eqn:h^2}
\begin{definition}
For $1 \leq j \leq K$, conditional on $\bbeta_{j}$, the heritability $h_{j}^2$ of the data $(\Xb_{j}, \yb_{j})$ is defined as 
$h_{j}^2 =\lim_{n, p \to \infty} \var(\Xb_{j} \bbeta_{j})/\var(\yb_{j})=\lim_{n, p \to \infty} \bbeta_{j}^{\T} \Xb_{j}^{\T} \Xb_{j} \bbeta_{j}/(\bbeta_{j}^{\T} \Xb_{j}^{\T} \Xb_{j} \bbeta_{j} + \bepsilon_{j}^{\T} \bepsilon_{j})$.
% \begin{align*}
%     \begin{split}
%     h_{j}^2 = \lim_{n, p \to \infty} \frac{\var(\Xb_{j} \bbeta_{j})}{\var(\yb_{j})} = \lim_{n, p \to \infty} \frac{\bbeta_{j}^{\T} \Xb_{j}^{\T} \Xb_{j} \bbeta_{j}}{\bbeta_{j}^{\T} \Xb_{j}^{\T} \Xb_{j} \bbeta_{j} + \bepsilon_{j}^{\T} \bepsilon_{j}}.
%     \end{split}
% \end{align*} 
\end{definition}

Algorithm \ref{alg:multi_sum} outlines the detailed procedure for the resampling-based self-training of multi-ancestry data.
When multi-ancestry summary statistics are available, $K$ vectors $\sbb_{j}^{\train}, 1 \leq j \leq K$, are sampled simultaneously for training. The final estimator across the $K$ ancestries is then computed as
\begin{align} \label{eqn:multi}
    \hat{\bbeta}_{\rm MA}(\Theta)^* = \sum_{j=1}^{K} \omega_j \Ab_{j}(\Wb_{j}^{\T} \Wb_{j}, \Theta_{j}) \sbb^{\train}_{j}.
\end{align}
Here $\Ab_{j}(\Wb_{j}^{\T} \Wb_{j}, \Theta_{j}) \sbb^{\train}_{j}$ represents the estimator using summary statistics from the population $j$, as described in Section \ref{sec:general}. 
The weight $\omega_{j}$ quantifies the contribution of data from the population $j$ to the target population, and $\Theta_{j}$ denotes the hyperparameter for this population. When the population $1$ is the target, all hyperparameter is tuned using $\sbb_{1}^{\valid}$.
If individual-level data are available, the corresponding estimator $\hat{\bbeta}_{\rm MA}(\Theta)$ can be trained by replacing $\sbb^{\train}_{j}$ with $\Xb_{j}^{\train}{}^{\T} \yb_{j}^{\train}$, as outlined in Algorithm \ref{alg:multi_ind} of the supplementary material.

\begin{algorithm}
\caption{Summary data-based model training with multi-ancestry data resources}\label{alg:multi_sum}
\begin{algorithmic}
\Require Summary data $\Xb_{j}^{\T} \yb_{j} \in \RR^{p}$, $\Wb_{j}^{\T} \Wb_{j}$, for $1 \leq j \leq K$, and  hyperparameter $\Theta = \{\omega_{j}, \Theta_{j}\}_{j=1}^{K}$. 
\vspace{1mm}

\For{$j \gets 1$ to $K$} 

\State $\hb_{j} \gets \cN(0,\Ib_p),$  \hfill 		\texttt{//} Sample $p$-dimension standard Gaussian random variable.
\vspace{1mm}

\State $\sbb_{j}^{\train} \gets \frac{n^{\train}}{n} \Xb_{j}^{\T} \yb_{j} + \sqrt{\frac{n^{\train} (n - n^{\train})}{n^2}} \Cov(\Xb_{j}^{\T} \yb_{j})^{1/2} \hb_{j},$  \hfill 		\texttt{//} Construct $\sbb_{j}^{\train}$ for training.
\vspace{1mm}
\EndFor
\vspace{1mm}

\State $\sbb_{1}^{\valid} \gets \Xb_{1}^{\T} \yb_{1} - \sbb_{1}^{\train},$  \hfill 		\texttt{//} Construct $\sbb_{1}^{\valid}$ for validation in the population $1$.
\vspace{1mm}

\State $\hat{\bbeta}_{\rm MA}(\Theta)^* \gets \sum_{j=1}^{K} \omega_{j} \Ab_{j}(\Wb_{j}^{\T} \Wb_{j}, \Theta_{j}) \sbb^{\train}_{j},$  \hfill 		\texttt{//} Obtain the estimator.
\vspace{1mm}

\State $R_{\rm sum, MA}^2(\Theta) \gets ({n^{\valid}}\left/{\|\yb^{\valid}\|_{2}^2})\right. \cdot {\left \langle {\sbb^{\valid}_{1}}, \hat{\bbeta}_{\rm MA}(\Theta)^* \right \rangle^{2} } \left/({n^{\valid} \cdot \| \hat{\bbeta}_{\rm MA}(\Theta)^* \|_{\bSigma}^{2}}), \right.$  
\vspace{1mm}

\Statex \hfill \texttt{//} Compute the $R^2_{\rm sum, MA}$ in \eqref{eqn:R2_sum_def}.
\vspace{1mm}

\State $\Theta^{*}_{\rm sum, MA} \gets \max_{\Theta} R^2_{\rm sum, MA}(\Theta),$  \hfill 		\texttt{//} Choose the best-performing hyperparameters.
\vspace{1.5mm}

\Return $R_{\rm sum, MA}^2(\Theta)$ and $\Theta^{*}_{\rm sum, MA}$.

\end{algorithmic}
\end{algorithm}

Theorem \ref{thm:multi} demonstrates that, with multi-ancestry summary-level data, Algorithm \ref{alg:multi_sum} allows for the selection of hyperparameter that is asymptotically equivalent to that obtained using individual-level data.

\begin{theorem} \label{thm:multi}
    Consider any random sequence $\{\bbeta_{j}, \bepsilon_{j}, \Xb_{j}, \Wb_{j}\}_{(p, n, n_w) \in \NN^3}, 1 \leq j \leq K$ satisfying Conditions \ref{cond-np-ratio} and \ref{cond-X-cross}-\ref{cond-eps-cross} and any parameter vector $\Theta = \{\omega_j, \Theta_j\}_{j=1}^{K}$. 
    Suppose the sequence $\{\Ab_{j}(\Wb_{j}^{\T}\Wb_{j}, \Theta_{j})\}$ satisfies
    \begin{align*}
        n_w \cdot \Ab_{j}(\Wb_{j}^{\T}\Wb_{j}, \Theta_{j}) \asymp \Db_{j}(\bSigma_{j}, \Theta_{j}) \quad\mbox{and} \quad n_w^2 \cdot \Ab_{j}(\Wb_{j}^{\T}\Wb_{j}, \Theta_{j})^{\T} \bSigma_{1} \Ab_{j}(\Wb_{j}^{\T}\Wb_{j}, \Theta_{j}) \asymp \Eb_{j}(\bSigma_{j}, \Theta_{j})
    \end{align*}
    with some $\Db_{j}(\bSigma_{j}, \Theta_{j})$ and $\Eb_{j}(\bSigma_{j}, \Theta_{j})$ depending on $\bSigma_{j}$ and $\Theta_{j}$ only. 
    We abbreviate $\Db_{j}(\bSigma_{j}, \Theta_{j})$ and $\Eb_{j}(\bSigma_{j}, \Theta_{j})$ by $\Db_{j}$ and $\Eb_{j}$, respectively.
    Then for any $\Theta$ and $\hat{\bbeta}_{\rm MA}(\Theta)^*$ defined in Equation \eqref{eqn:multi}, the out-of-sample $R^2$ is 
    \begin{align*}
        R^2_{\rm sum, MA}(\theta)
        = \frac{n^{\valid}}{\|\yb^{\valid}\|_{2}^2} \cdot \frac{\left( \Lambda_{\rm sum}^{(1)} \right)^{2}}{\Lambda_{\rm sum}^{(2)}} + o_p(1),
    \end{align*}
    where
    \begin{align*}
        \Lambda_{\rm sum}^{(1)} 
        =\ & \omega_1 \cdot \frac{n^{\train}}{n_w} \cdot \frac{\kappa_{1} \sigma_{\bbeta}^2}{p} \cdot \tr\left( \Db_{1} \bSigma_{1}^2 \right) + \sum_{j=2}^{K} \omega_j \frac{n^{\train}}{n_w}  \cdot \frac{\kappa_{1} \kappa_{j} \sigma_{1,j}^2}{p} \cdot \tr \left( \bSigma_{1} \Db_{j} \bSigma_{j} \right) \quad \mbox{and}
        \\
        \Lambda_{\rm sum}^{(2)} =\ & \sum_{1 \leq i < j \leq K}  2 \omega_i \omega_{j} \left( \frac{n^{\train}}{n_w} \right)^2 \frac{\kappa_{i} \kappa_{j} \sigma_{i,j}^2}{p} \cdot \tr \left( \bSigma_{i} \Db_{i} \bSigma_{1} \Db_{j} \bSigma_{j} \right)
        \\
        &+ \sum_{1 \leq j \leq K} \omega_{j}^{2} \cdot \left\{  \frac{n^{\train}}{n_w^2} \cdot \frac{\kappa_{j} \sigma_{j, j}^2}{p} \cdot \frac{1}{h_{j}^2} \tr \left(\bSigma_{j} \right) \cdot \tr \left(\Eb_{j} \bSigma_{j} \right) + \left( \frac{n^{\train}}{n_w} \right)^2 \cdot \frac{\kappa_{j} \sigma_{j,j}^2}{p}\tr \left(\Eb_{j} \bSigma_{j}^2 \right) \right\}.
    \end{align*}
    Moreover, for any $\Theta = \{\omega_{j}, \Theta_{j}\}_{j=1}^{k}$, we have     
    $$R^2_{\rm sum, MA}(\Theta) / R^2_{\rm ind, MA}(\Theta) \overset{p}{\to} 1.$$
    % \bxz{You need condition like Equation~\ref{eqn:cond-A}? Also, rewrite to have a similar style as previous Theorems. What is $ F_2$?} \bxs{done}
\end{theorem}

%\bxz{Explain the theorem, for example, why the two match, etc.}

The following two corollaries provide additional insights into the optimal weights across populations. 
For simplicity, we focus on the case with $K=2$ populations.
Corollary \ref{cor:optimal_weight} presents the closed-form for the best-performing population weights $\omega_1$ and  $\omega_2$, indicating that the best-performing weights are determined by the covariance structure between the second and target populations, $\sigma_{1,2}^2 \Ib_{p}$. 
The specific forms of $\omega_1$ and  $\omega_2$ are outlined in Equation~\eqref{eqn:optimal_weight}. Briefly, when the genetic effects of the second population are positively correlated with those of the target population (i.e., $\sigma_{1,2}^2 > 0$), incorporating data from the second population improves the prediction accuracy of the target prediction.

%This is further verified by the following corollary. 
\begin{corollary} \label{cor:optimal_weight}
Under the same conditions as in Theorem \ref{thm:multi}, and considering the special case $K = 2$, we have $\omega_2 = 1-  \omega_1$, and the closed-form expression for $R^2_{\rm sum, MA}(\Theta)$ in Theorem \ref{thm:multi} is 
    \begin{align*}
        R^2_{\rm sum, MA}(\Theta)
        = \frac{n^{\valid}}{\|\yb^{\valid}\|_{2}^2} \cdot \frac{ \left[ \omega_1 \cdot N_1(\Theta) + (1 - \omega_1) N_2(\Theta) \right]^2}{\omega_1^2 \cdot D_1(\Theta) + (1 - \omega_1)^2 \cdot D_2(\Theta) + 2 \omega_1 (1 - \omega_1) \cdot D_3(\Theta)} + o_p(1), 
    \end{align*}
    where
    \begin{align*}
        N_1(\Theta) =\ & \frac{\kappa_{1} \sigma_{\bbeta}^2}{p} \cdot \tr\left( \Db_{1} \bSigma_{1}^2 \right), \quad  N_2(\Theta) =\frac{\kappa_{1} \kappa_{2} \sigma_{1,2}^2}{p} \cdot \tr \left( \bSigma_{1} \Db_{2} \bSigma_{2} \right), 
        % \\
        % N_2(\Theta) =\ & \frac{\kappa_{1} \kappa_{2} \sigma_{1,2}^2}{p} \cdot \tr \left( \bSigma_{1} \Db_{2} \bSigma_{2} \right) 
        \\
        D_1(\Theta) =\ & \frac{1}{n^{\train}} \cdot \frac{\kappa_{1} \sigma_{1, 1}^2}{p} \cdot \frac{1}{h_{1}^2} \tr \left(\bSigma_{1} \right) \cdot \tr \left(\Eb_{1} \bSigma_{1} \right) + \frac{\kappa_{1} \sigma_{1,1}^2}{p}\tr \left(\Eb_{1} \bSigma_{1}^2 \right), 
        \\
        D_2(\Theta) =\ & \frac{1}{n^{\train}} \cdot \frac{\kappa_{2} \sigma_{2, 2}^2}{p} \cdot \frac{1}{h_{2}^2} \tr \left(\bSigma_{2} \right) \cdot \tr \left(\Eb_{2} \bSigma_{2} \right) + \frac{\kappa_{2} \sigma_{2,2}^2}{p}\tr \left(\Eb_{2} \bSigma_{2}^2 \right), \quad \mbox{and}
        \\
        D_3(\Theta) =\ & \frac{\kappa_{1} \kappa_{2} \sigma_{1,2}^2}{p} \cdot \tr \left( \bSigma_{1} \Db_{1} \bSigma_{1} \Db_{2} \bSigma_{2} \right).
    \end{align*}
    Moreover, the optimal $R^2_{\rm sum, MA}(\Theta)$ is obtained when 
    \begin{align} \label{eqn:optimal_weight}
    \begin{split}
        \omega_1 =\ & \min \left\{1, \frac{D_2(\Theta) N_1(\Theta)-D_3(\Theta) N_2(\Theta)}{D_2(\Theta) N_1(\Theta) - D_3(\Theta) N_1(\Theta) +D_1(\Theta) N_2(\Theta) - D_3(\Theta) N_2(\Theta)} \right\} 
        \quad \mbox{and} \quad
        \\
        \omega_2 =\ & \max  \left\{0, 1 - \frac{D_2(\Theta) N_1(\Theta) -D_3(\Theta) N_2(\Theta)}{D_2(\Theta) N_1(\Theta) -D_3(\Theta) N_1(\Theta)+D_1(\Theta) N_2(\Theta)-D_3(\Theta) N_2(\Theta)} \right\}.
    \end{split}
    \end{align}
\end{corollary}

Corollary \ref{cor:uncor_weight} further examines the case where genetic effects in the second population are uncorrelated with those in the target population, i.e., $\sigma_{1,2}^2 = 0$. 
In this scenario, the optimal weights are $\omega_1 = 1$ and $\omega_2 = 0$, indicating that prediction should rely exclusively on data from the target population. 
These findings highlight that cross-ancestry genetic correlation \citep{brown2016transethnic,xue2023high} plays a crucial role in determining whether incorporating multi-ancestry data, either through resampling-based self-training or individual-level data training, can improve prediction accuracy for a single ancestry.

%Corollary \ref{cor:uncor_weight} indicates that when signal in second population is uncorrelated with target signal, the optimal strategy in make prediction is to use data from target population only. 
\begin{corollary}%[$\bbeta_1$, $\bbeta_2$ uncorrelated]
\label{cor:uncor_weight}
Under the same conditions as in Theorem \ref{thm:multi}, and considering the special case $K = 2$ and $\sigma_{1,2}^2 = 0$, the closed-form expression for $R^2_{\rm sum, MA}(\Theta)$ in Theorem \ref{thm:multi} is 
    \begin{align*}
        R^2_{\rm sum, MA}(\Theta)
        = \frac{n^{\valid}}{\|\yb^{\valid}\|_{2}^2} \cdot \frac{ \omega_1^2 \cdot N_1^2(\Theta)}{\omega_1^2 \cdot D_1(\Theta) + (1 - \omega_1)^2 \cdot D_2(\Theta)}. 
    \end{align*}
    Moreover, the optimal $R^2_{\rm sum, MA}(\Theta)$ is obtained when 
    $\omega_1 =1$ and $\omega_2 = 0.$
\end{corollary}

\section{Numerical experiments}
\label{sec:numer}
We numerically validate our theoretical findings through extensive synthetic data simulations and real data analyses using the UK Biobank \citep{bycroft2018uk}.
The primary goal of the synthetic data analysis is to demonstrate that our theoretical results hold across general settings, while the real data analysis illustrates their specific applicability in real-world genetic predictions. Overall, these complementary analyses strongly support our theoretical findings.
More importantly, we provide additional insights into genetic data prediction applications. For example, we show that nonlinear estimators may also work well with resampling-based self-training. In addition, we find that resampling-based self-training can even outperform conventional individual-level data training when the tuning dataset has a limited sample size.

\subsection{Simulation study}
In this section, we conduct simulations using synthetic data to numerically illustrate the theoretical results on the performance of resampling-based self-training presented in Section~\ref{sec:general}. Our experiments systematically evaluate various settings of heritability $h^2$, the $p/n$ ratio, sparsity $\kappa$, and sample size $n$. 
The results demonstrate that resampling-based self-training methods (Algorithm \ref{alg:sum}) achieve performance comparable to conventional individual-level data training methods (Algorithm \ref{alg:ind} in the supplementary material).

%in terms of both validation and testing $R^2$.
We generate synthetic data as follows. To replicate the local LD pattern observed in human genetic data, we use a block-wise covariance matrix
$ \bSigma = \diag\{\bSigma_1, \bSigma_2, \dots, \bSigma_{n_{\rm block}}\},$
% \begin{align*} 
%     \bSigma = \diag\{\bSigma_1, \bSigma_2, \dots, \bSigma_{n_{\rm block}}\},
% \end{align*}
where each block $\bSigma_i$ follows an autoregressive AR(1) structure with correlation coefficient $\rho \in (0,1)$:
\begin{align} \label{eqn:AR1}
    \bSigma_{i} = \begin{pmatrix}
        1 & \rho & \cdots & \rho^{p/n_{\rm block}} \\
        \vdots & \ddots & \ddots & \vdots \\
        \rho^{p/n_{\rm block}} & \rho^{p/n_{\rm block} - 1} & \cdots & 1
    \end{pmatrix}.
\end{align}
The values of $n_{\rm block}$ and $\rho$ vary across different simulation settings.
Each row of the data matrix $\Xb$ and the reference panel matrix $\Wb$ is independently sampled from a multivariate normal distribution with covariance matrix $\bSigma$. We use reference panels sample size $n_w = 1000$. The phenotype vector $\yb$ is generated according to the linear model \eqref{eqn:linear_model}, with Gaussian noise variance $\sigma_{\epsilon}^2$ determined by the heritability.
We evaluate the predictive performance of the following methods: (i) resampling-based ridge estimator in Section~\ref{sec:ref_ridge} and (ii) resampling-based marginal thresholding in Section~\ref{sec:marginal}. 

%As briefly introduced in Section~\ref{sec:general}, t
The left panels of Figures \ref{fig:ridge_R2} and \ref{fig:marginal_R2} present the pattern of out-of-sample $R^2$ across various hyperparameter values for both resampling-based self-training and individual-level training methods. 
We repeat Algorithms \ref{alg:sum} and \ref{alg:ind} over $100$ iterations and report the averaged $R^2_{\rm sum, R}(\theta)$, $R^2_{\rm ind, R}(\theta)$, $R^2_{\rm ind, M}(\theta)$, and $R^2_{\rm ind, M}(\theta)$. 
We have the following key observations. First, both the resampling-based ridge-type estimator and marginal thresholding have a unique maximum in $R^2_{\rm sum, R}(\theta)$, $R^2_{\rm ind, R}(\theta)$, $R^2_{\rm ind, M}(\theta)$, and $R^2_{\rm ind, M}(\theta)$. 
This underscores the importance of selecting the optimal hyperparameter to maximize predictive performance. 
Second, we consistently observe that $R^2_{\rm sum, R}(\theta)$ and $R^2_{\rm sum, M}(\theta)$ closely aligns with $R^2_{\rm ind, R}(\theta)$ and $R^2_{\rm ind, M}(\theta)$ across different hyperparameter values, respectively. 
Importantly, the besting-performing tuning parameters, $\theta_{\rm sum, R}^{*}$ and $\theta_{\rm ind, R}^{*}$, which maximize $R^2_{\rm sum, R}(\theta)$ and $R^2_{\rm ind, R}(\theta)$, respectively, are well-aligned.
{A similar phenomenon is observed between $\theta_{\rm sum, M}^{*}$ and $\theta_{\rm ind, M}^{*}$.}
In addition, the right panels of Figures \ref{fig:ridge_R2} and \ref{fig:marginal_R2} compare prediction accuracy under the selected best-performing tuning parameters, specifically:
(i) $\hat{\bbeta}_{\rm R}(\theta^*_{\rm sum, R})^*$ versus $\hat{\bbeta}_{\rm R}(\theta^*_{\rm ind, R})$ and (ii) $\hat{\bbeta}_{\rm M}(\Theta^*_{\rm sum, M})^*$ versus $\hat{\bbeta}_{\rm M}(\Theta^*_{\rm ind, M})$. 
We evaluate a wide range of parameter settings and compute the out-of-sample $R^2$ in an independent dataset drawn from the same distribution as 
$(\Xb, \yb)$, with a sample size of $3000$. These results indicate that the resampling-based self-training procedure achieves predictive accuracy comparable to Algorithm \ref{alg:ind} across a broad range of conditions.

Overall, we find that resampling-based self-training achieves comparable performance in selecting the best-performing hyperparameter and, consequently, similar prediction accuracy to individual-level training, without requiring access to individual-level data. These empirical findings strongly support our theoretical results in Theorems \ref{thm:reference_panel} and \ref{thm:marginal_screen}.

\subsection{Real imaging data analysis}
In this section, we conduct a real data analysis using the whole-body dual-energy X-ray absorptiometry (DXA) imaging data from the UK Biobank study \citep{bycroft2018uk}. 
Specifically, we focus on $71$ imaging-derived body composition traits categorized under data category 124. These DXA traits include bone mass, fat-free mass, tissue mass, and lean mass from different body regions, such as arms, legs, and trunk, with a full list can be found in \url{https://biobank.ndph.ox.ac.uk/ukb/label.cgi?id=124}.
%\cite{DXA}. 
% The corresponding UKB field IDs range from 21110 to 21135 and 23244 to 23289.
We use DXA imaging data from $45,622$ unrelated White British individuals and $2618$ unrelated White non-British individuals in the UK Biobank, all of whom have available genotype data.
For training, we use GWAS summary statistics derived from $45,622$ British individuals, following standard imaging and genetic data quality controls similar to those in \cite{su2024exact}. We adjust for the effects of age (at imaging), sex, their interactions, and the top $40$ genetic principal components.
Each DXA trait generates summary statistics with the number of SNPs ranging from $8,839,000$ and $8,840,000$. 
The mean SNP heritability is estimated to be $18.94\%$ by LDSC \citep{bulik2015ld} (Table~\ref{tab:h2_R2}). 
We further take a subset of approximately one million HapMap 3 genetic variants to use in our analysis  \citep{jin2025pennprs}.
For example, the summary statistics for trait 21110 (android fat-free mass) initially include $8,839,431$ SNPs. After mapping to HapMap 3 genetic variants, $1,080,259$ SNPs remain.
The $2,618$ White non-British individuals are randomly split into independent validation and testing datasets, with dataset sizes varying across different analyses, as specified in the following discussion. Additionally, we use an external reference panel from the 1000 Genomes \citep{10002015global}, consisting of individuals of European ancestry.

\begin{figure}[t]%[!htp]
    \centering
    \begin{subfigure}[b]{0.32\textwidth}
        \includegraphics[width=\textwidth]{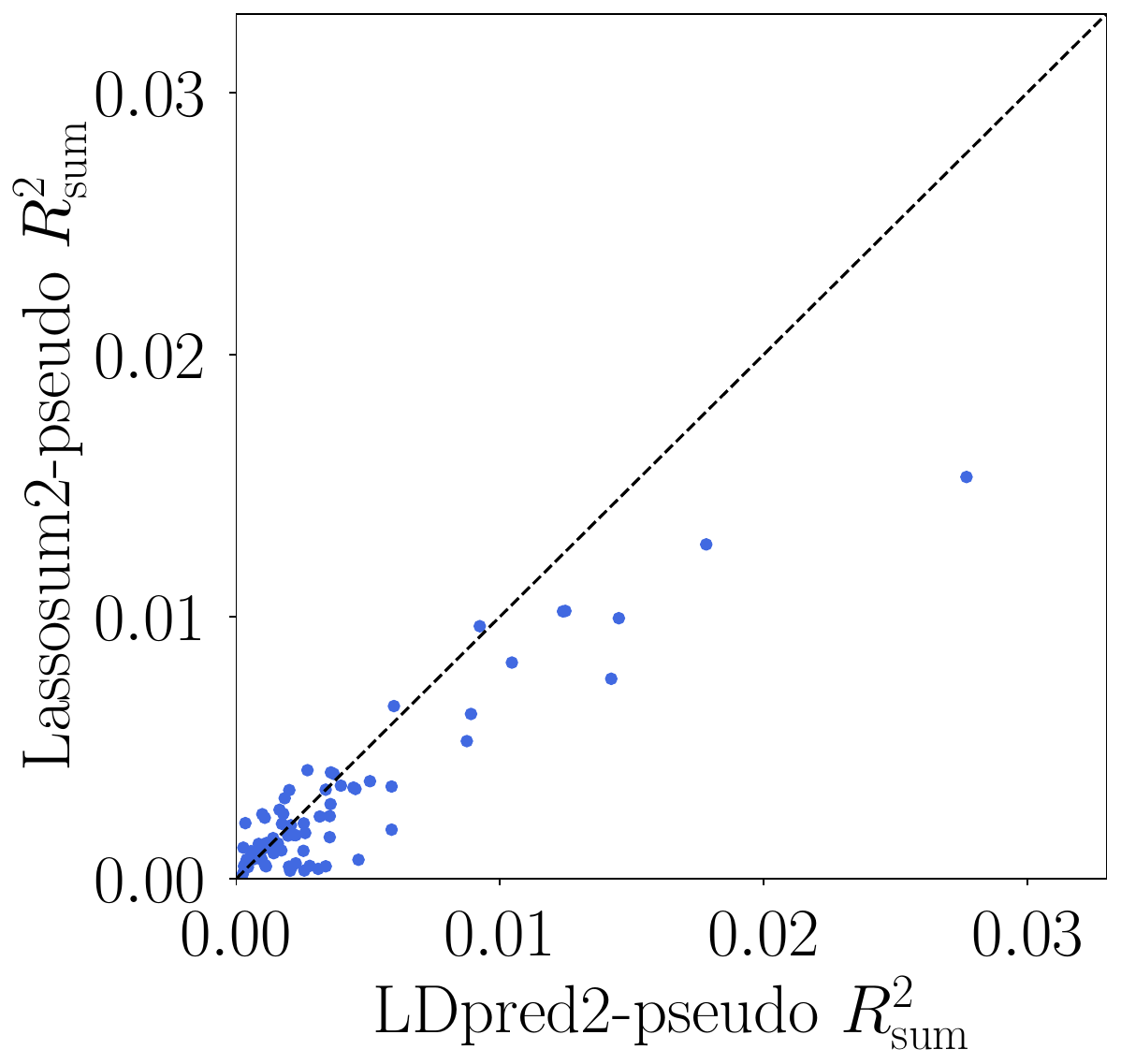}
    \end{subfigure}
    \hfill
    \begin{subfigure}[b]{0.32\textwidth}
        \includegraphics[width=\textwidth]{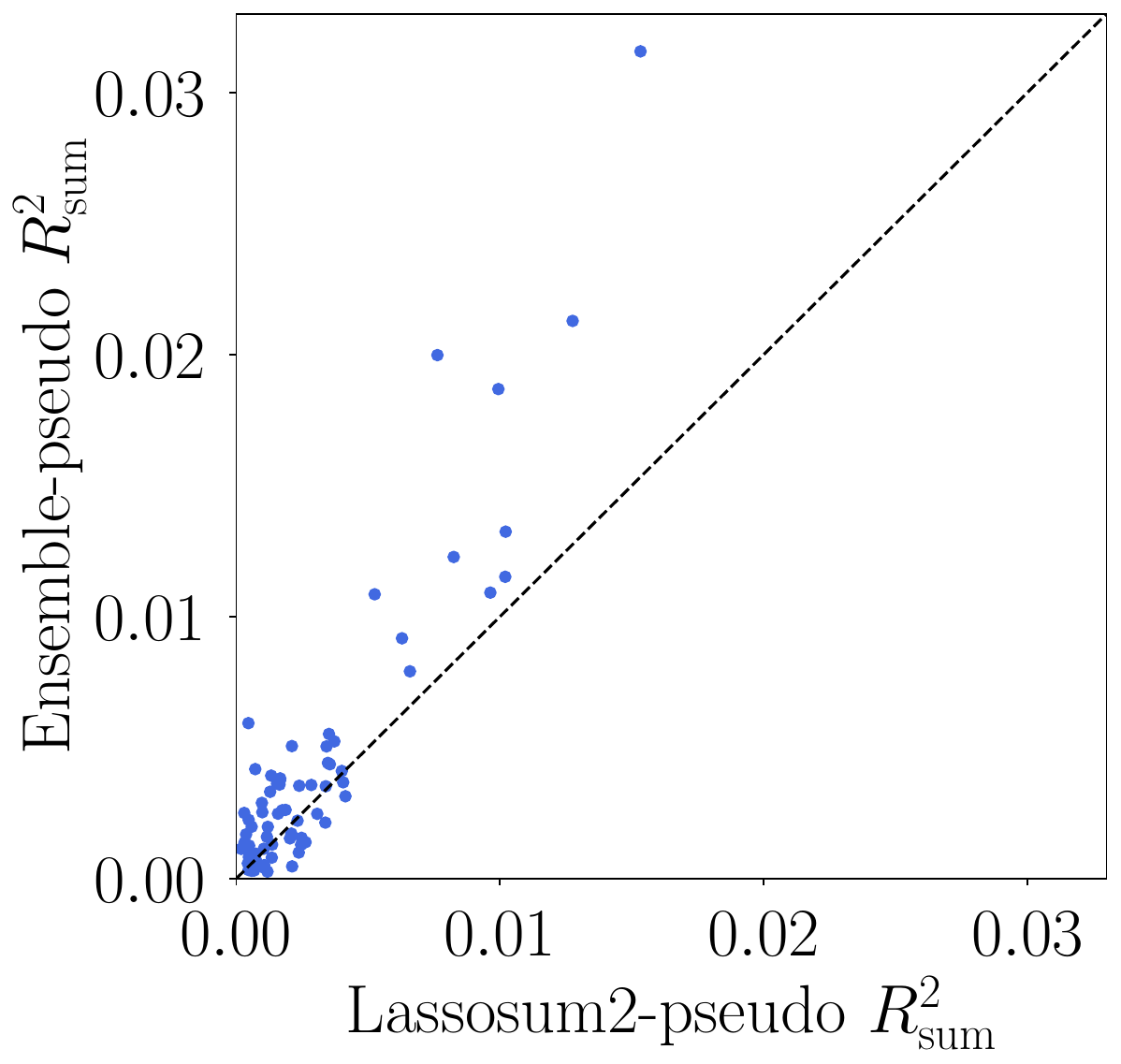}
    \end{subfigure}
    \hfill
    \begin{subfigure}[b]{0.32\textwidth}
        \includegraphics[width=\textwidth]{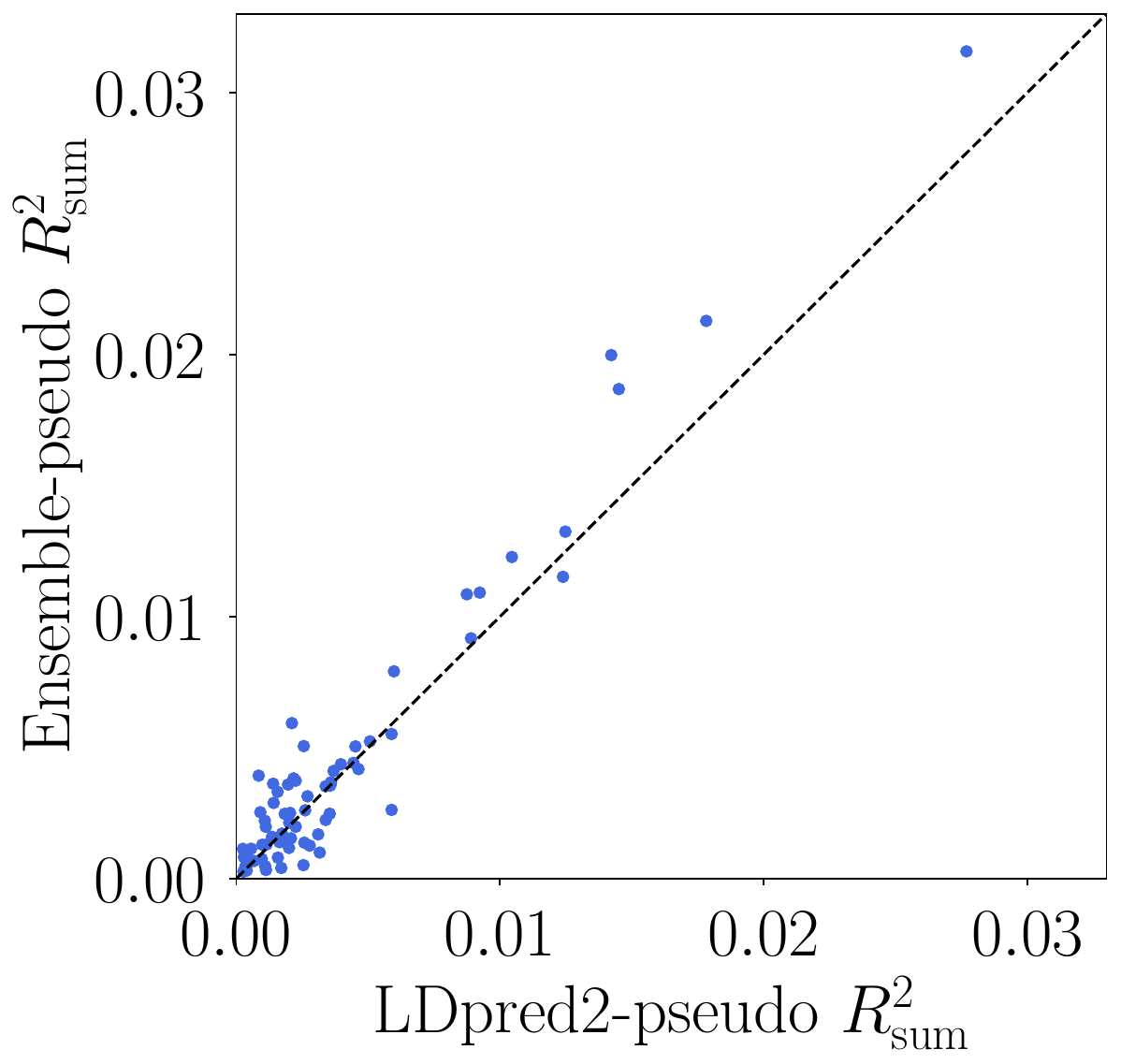}
    \end{subfigure}
    \caption{\textbf{Comparison of out-of-sample $R^2$ across resampling-based self-training  methods using DXA imaging data.} 
    Each scatter plot compares the out-of-sample $R^2$ of $71$ DXA traits obtained using Algorithm \ref{alg:sum} for a pair of resampling-based self-training methods: 
    (\textbf{Left}) LDpred2-pseudo vs. Lassosum2-pseudo, (\textbf{Middle}) Lassosum2-pseudo vs. Ensemble-pseudo, and (\textbf{Right}) LDpred2-pseudo vs. Ensemble-pseudo. 
    %The dashed diagonal line represents the identity line ($y = x$). 
    Data points above the diagonal suggest superior performance of the method on the $y$-axis, while points below the diagonal indicate superior performance of the method on the $x$-axis. 
    Results show that LDpred2-pseudo generally outperforms Lassosum2-pseudo, whereas they have comparable prediction accuracy for lower-heritability traits. 
    Ensemble learning, which combines multiple methods, generally outperforms individual methods, especially for highly heritable traits.}
  \label{fig:test_R_method_comparison}
\end{figure}

We apply resampling-based self-training (Algorithm \ref{alg:sum}) to two widely used PRS methods originally designed for requiring individual-level validation data: Lassosum \citep{mak2017polygenic, prive2022identifying} and LDpred2 \citep{vilhjalmsson2015modeling, prive2020ldpred2}. We refer to their self-training versions as Lassosum2-pseudo and LDpred2-pseudo, respectively.
Additionally, we integrate an ensemble approach by using Algorithm \ref{alg:ensemble_sum}, denoted as Ensemble-pseudo, which combines PRS models trained using different methods via a linear combination strategy \citep{jin2025pennprs}. 
The hyperparameter for these resampling-based methods, implemented through Algorithms \ref{alg:sum} and \ref{alg:ensemble_sum}, is selected using the resampling-based pseudo-validation dataset $\sbb^{(v)}$.
%In Algorithm \ref{alg:sum}, w
We construct 
$\sbb^{(tr)}$ and $\sbb^{(v)}$
 with a training-to-validation ratio of $n^{(tr)}:n^{(v)} = 8:2$ for parameter pseudo-tuning. 
For the conventional individual-level training (Algorithm \ref{alg:ind}), we randomly select $1000$ subjects from the White non-British sample as the validation dataset.
The prediction accuracy of all resampling-based and individual-level training methods is evaluated on the remaining testing subset of White non-British individuals who are not used for individual-level model training, with a sample size of $1618$.

Figure \ref{fig:test_R_method_comparison} presents the out-of-sample $R^2$ of resampling-based self-training methods across $71$ DXA traits. Both LDpred2-pseudo and Lassosum2-pseudo yield reliable prediction accuracy measures and exhibit a consistent pattern across the two methods. Among these traits, LDpred2-pseudo generally achieves better prediction accuracy than Lassosum2-pseudo, particularly for highly heritable traits with bigger out-of-sample $R^2$.  As expected, Ensemble-pseudo outperforms both methods, improving prediction accuracy, especially over Lassosum2-pseudo. For example, in trait 23244 (android bone mass), the most heritable trait, Ensemble-pseudo achieves a prediction accuracy of $3.2\%$, compared to $2.7\%$ for LDpred2-pseudo and $1.5\%$ for Lassosum2-pseudo. 
A complete summary of prediction accuracy measures across all DXA traits can be found in Table \ref{tab:h2_R2}.

\begin{figure}[t] %[!htp]
    \centering
    \begin{subfigure}[b]{0.32\textwidth}
        \includegraphics[width=\textwidth]{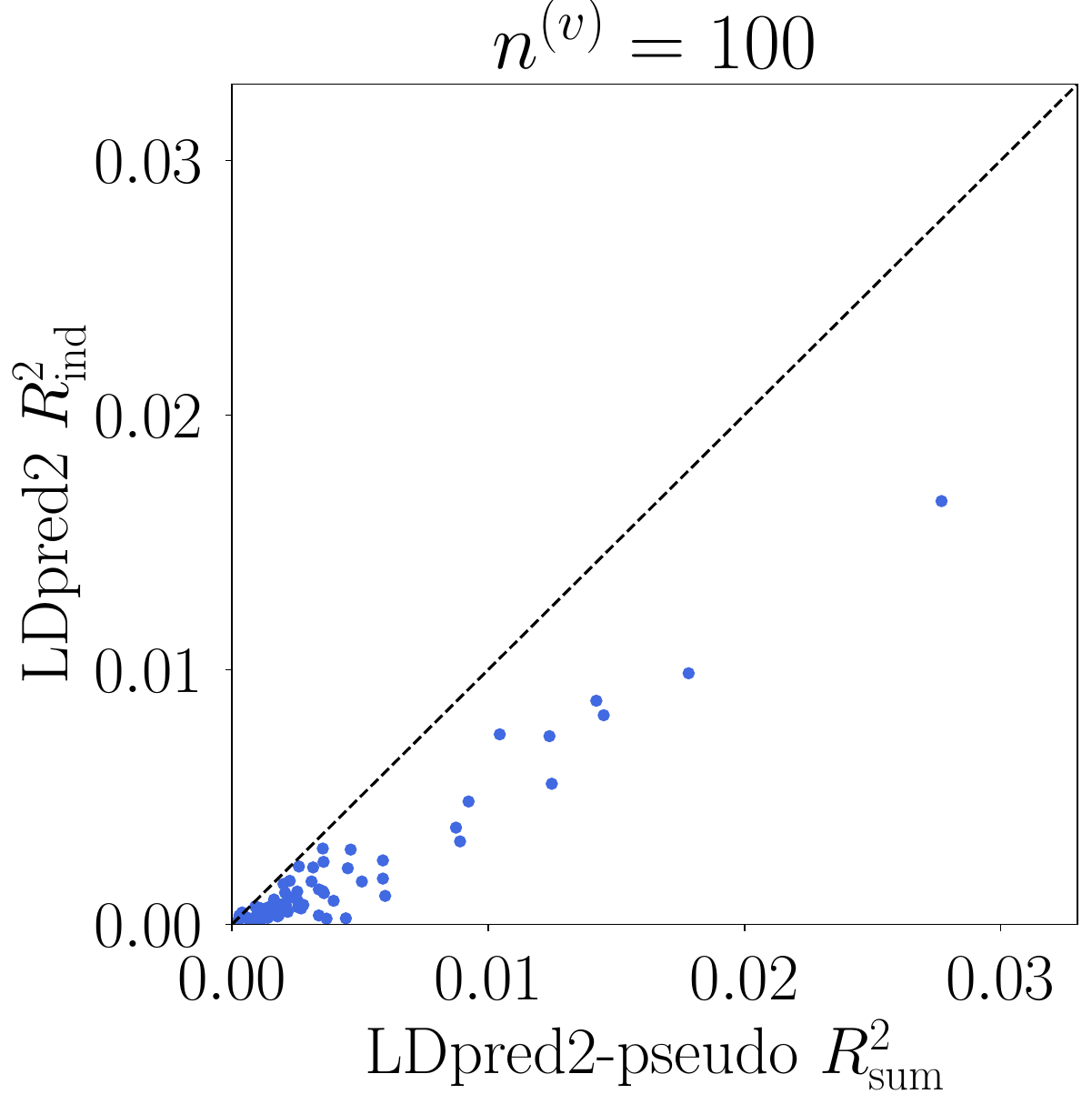}
    \end{subfigure}
    \hfill
    \begin{subfigure}[b]{0.32\textwidth}
        \includegraphics[width=\textwidth]{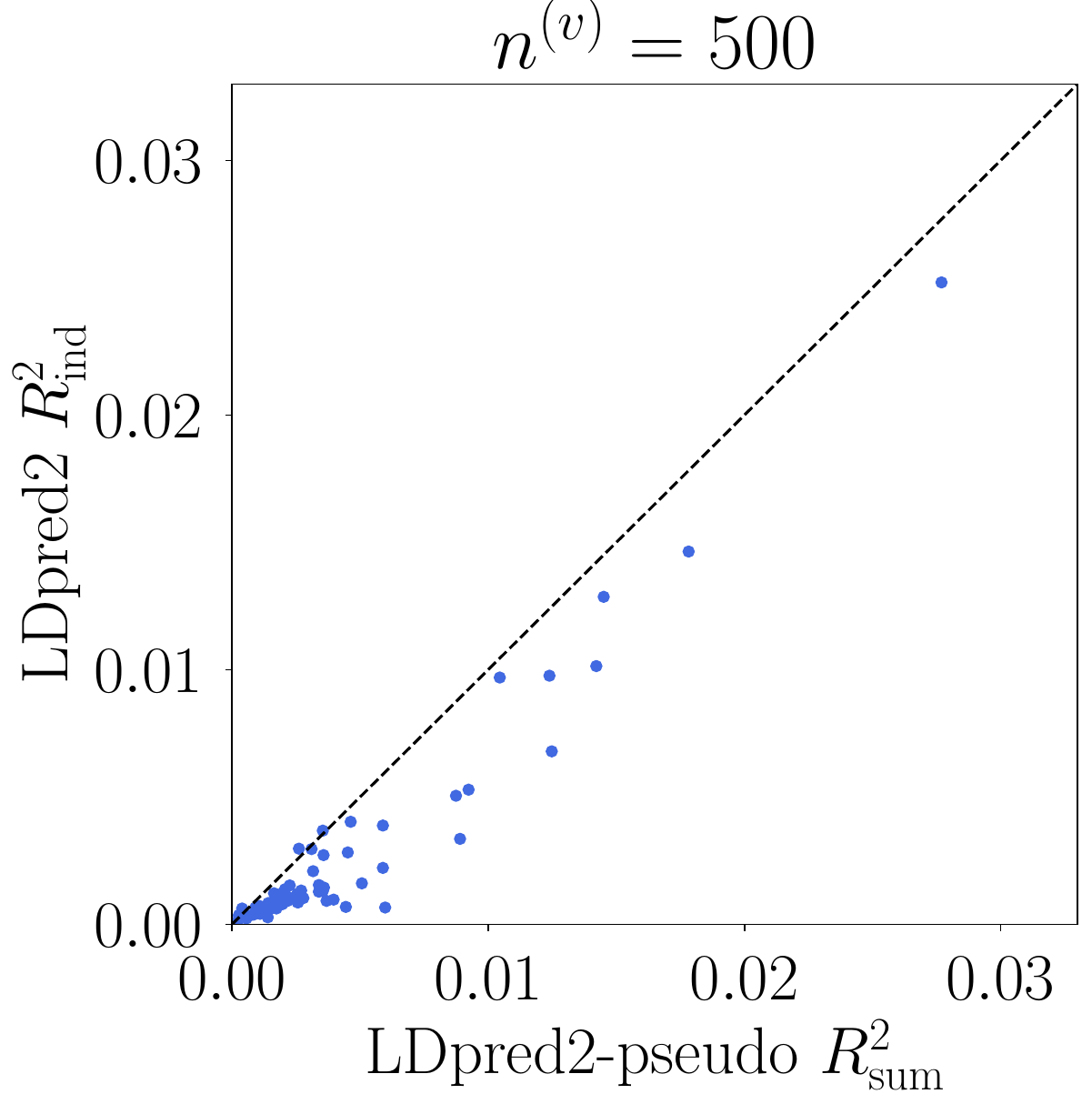}
    \end{subfigure}
    \hfill
    \begin{subfigure}[b]{0.32\textwidth}
        \includegraphics[width=\textwidth]{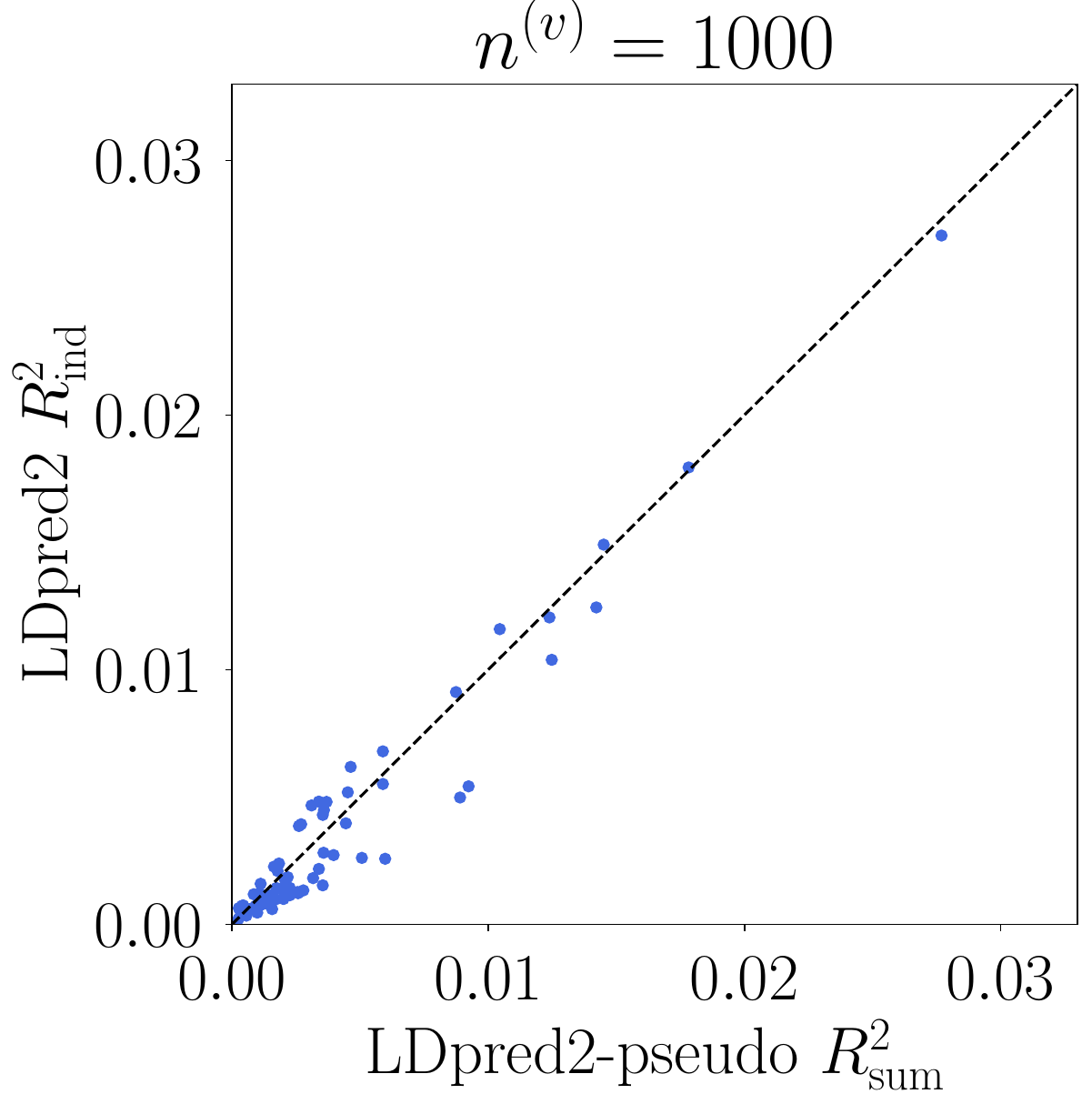}
    \end{subfigure}
    \caption{\textbf{Comparison of out-of-sample $R^2$ between resampling-based self-training and individual-level data training.} 
    Each scatter plot compares the out-of-sample $R^2$ of $71$ DXA traits obtained using Algorithm \ref{alg:sum} and Algorithm \ref{alg:ind}. 
    To assess the impact of validation sample size on prediction accuracy, we evaluate different sample sizes for the individual-level validation dataset in Algorithm \ref{alg:ind}:  
    (\textbf{Left}) $n^{(v)} = 100$, %, with 2,518 individuals used for testing.
    (\textbf{Middle}) $n^{(v)} = 500$, and %, with 2,118 individuals used for testing.
    (\textbf{Right}) $n^{(v)} = 1000$. %, with 1,618 individuals used for testing.
    %The dashed diagonal line represents the identity line ($y = x$), indicating equal performance between the two methods. 
    For Algorithm \ref{alg:sum}, the sample size of the pseudo-validation dataset is fixed to be $20\%$ of the GWAS sample size. 
    Results show that Algorithm \ref{alg:sum}, using only summary data, achieves prediction accuracy comparable to Algorithm \ref{alg:ind} when $n^{(v)} = 1000$. Moreover, Algorithm \ref{alg:sum} may outperform Algorithm \ref{alg:ind} when the individual-level validation dataset has a limited sample size ($n^{(v)} = 100$ or $500$).}
  \label{fig:test_R_alg_comparison}
\end{figure}

Next, we compare resampling-based self-training with individual-level data training.
Since LDpred2-pseudo outperforms Lassosum2-pseudo overall in Figure \ref{fig:test_R_method_comparison}, we focus this analysis on comparing LDpred2-pseudo with LDpred2. 
Figure \ref{fig:test_R_alg_comparison} shows that when using $1000$ subjects as the validation dataset for LDpred2, LDpred2-pseudo and LDpred2 achieve highly similar performance (Table \ref{tab:h2_R2}). 
%\bxs{More details can be found in the last two columns of Table \ref{tab:h2_R2}.} 
These results suggest that LDpred2-pseudo performs comparably to LDpred2, which tunes hyperparameters using a large individual-level dataset.
However, as the individual-level validation sample size decreases, LDpred2-pseudo clearly outperforms LDpred2. This is because LDpred2's performance declines with a smaller validation sample, as it requires a sufficiently large dataset for reliable training. In contrast, LDpred2-pseudo does not face this limitation, as it relies on resampling-based pseudo-training/validation datasets.
For example, in trait 23244 (android bone mass), LDpred2 achieves a prediction accuracy of $2.7\%$, $2.5\%$, and $1.7\%$ when using $1000$, $500$, and $100$ validation samples, respectively, while LDpred2-pseudo has a prediction accuracy of $2.8\%$.

Overall, our real data analysis provides valuable insights into the relative prediction accuracy of different estimators and algorithms for DXA imaging data.
Among $71$ DXA traits, LDpred2-pseudo outperforms Lassosum2-pseudo, while ensemble-pseudo further enhances performance using only summary data.
Additionally, we find that resampling-based self-training methods can outperform individual-level training when the validation sample size is small, highlighting their advantage in data-limited scenarios.

\section{Discussion}
\label{sec:disc}
In this paper, we develop a statistical framework to establish the properties of self-training approaches using summary data. Notably, we demonstrate that pseudo-training and validation datasets based solely on summary statistics can achieve the same asymptotic predictive accuracy as traditional methods using individual-level training and validation data. We show that these results hold in high-dimensional settings but follow a different rationale than in low-dimensional cases.
In low dimensions, the no-cost property of summary data-based training is attributed to the multivariate CLT, which ensures that pseudo-training/validation datasets share the same limiting distribution as individual-level data. 
In high dimensions, however, standard high-dimensional CLT results do not apply \citep{chernozhukov2017central, fang2021high}. 
To address this, we leverage random matrix theory to quantify asymptotic prediction accuracy and show that the no-cost property still holds, even without high-dimensional CLT. Notably, the sampling distribution of $\sbb^{\train}$ does not need to be the limiting distribution of $\Xb^{\train}{}^{\T} \yb^{\train}$. 

Our proofs reveal that the key lies in the asymptotic matching of traces between individual-level and summary-level data, requiring only matched moments of the functionals rather than matching distributions.
We also demonstrate the surprising finding that, despite the lack of independence between resampling-based training and validation datasets, this does not lead to overfitting or reduced out-of-sample performance.
Inspired by recent trends in genetic data analysis, we extend our work to include algorithms and theoretical analyses for ensemble learning \citep{pain2021evaluation, yang2022PGS} and multi-ancestry data analysis \citep{kachuri2024principles,zhang2023new,zhang2024ensemble,jin2024mussel}. In summary, these results offer deeper insights into the theoretical underpinnings of self-training with summary data and may support the broader application of self-training algorithms in fields that utilize shared summary data. 

As the first statistical framework to examine the properties of resampling-based self-training, 
our study has a few limitations. First, we focus on linear estimators and provide two concrete examples commonly used in genetic and dense-signal predictions \citep{choi2020tutorial,ge2019polygenic}. 
While this covers a broad class of estimators, our analysis does not include nonlinear estimators, such as Lasso and Elastic-net, which are also frequently used in similar prediction tasks \citep{mak2017polygenic,qian2020fast,wu2023large,wang2024integrating}. 
This choice is due to our approach of using random matrix theory to derive exact analytical forms for prediction accuracy, which is challenging to apply to nonlinear estimators as they typically lack closed-form solutions \citep{su2024exact}. 
However, our real data analyses 
indicate that the self-training framework performs well with more complicated nonlinear estimators, such as Lassosum \citep{mak2017polygenic, prive2022identifying} and LDpred2 \citep{vilhjalmsson2015modeling, prive2020ldpred2}, which are specifically designed for genetic prediction. 
Thus, our theoretical insights on linear estimators may also extend to many nonlinear and more complicated estimators, which could be better explored in future studies. 
Second, although summary data-based training has the no-cost property regarding asymptotic prediction accuracy, 
it may exhibit greater variance compared to individual-level data training. %(Figure~\ref{fig:boxplot_R}). we observe numerically that
Uncertainty analysis within the framework of random matrix theory is still in its early stages \citep{fu2024uncertainty}.  
It would be interesting to quantify the uncertainty of summary data-based training and to develop improved resampling strategies to reduce this variability. Such advancements could lead to more efficient self-training of prediction models using summary data.

%not independent, but it does not matter 

\section*{Acknowledgement}
We would like to thank Xiaochen Yang and Juan Shu for their helpful discussions and for preparing the data resources. 
Research reported in this publication was supported by National Institute of Mental Health under Award Number R01MH136055 and National Institute on Aging under Award Number RF1AG082938. The content is solely the responsibility of the authors and does not necessarily represent the official views of the National Institutes of Health. The study has also been partially supported by funding from the Department of Statistics and Data Science at the University of Pennsylvania, Wharton Dean’s Research Fund, Analytics at Wharton, Wharton AI \& Analytics Initiative, Perelman School of Medicine CCEB Innovation Center Grant, and the University Research Foundation at the University of Pennsylvania. 
This research has been conducted using the UK Biobank resource (application number 76139), subject to a data transfer agreement. We thank the individuals represented in the UK Biobank for their participation and the research teams for their work in collecting, processing and disseminating these datasets for analysis. We would like to thank Purdue University and the Rosen Center for Advanced Computing
for providing computational resources and support that have contributed to these research results.

\bibliographystyle{apalike}
\bibliography{references}

\begin{thebibliography}{}

\bibitem[1000-Genomes-Consortium, 2015]{10002015global}
1000-Genomes-Consortium (2015).
\newblock A global reference for human genetic variation.
\newblock {\em Nature}, 526(7571):68--74.

\bibitem[Bai and Silverstein, 2010]{bai2010spectral}
Bai, Z. and Silverstein, J.~W. (2010).
\newblock {\em Spectral analysis of large dimensional random matrices}, volume~20.
\newblock Springer.

\bibitem[Bonomi et~al., 2020]{bonomi2020privacy}
Bonomi, L., Huang, Y., and Ohno-Machado, L. (2020).
\newblock Privacy challenges and research opportunities for genomic data sharing.
\newblock {\em Nature Genetics}, 52(7):646--654.

\bibitem[Boyle et~al., 2017]{boyle2017expanded}
Boyle, E.~A., Li, Y.~I., and Pritchard, J.~K. (2017).
\newblock An expanded view of complex traits: from polygenic to omnigenic.
\newblock {\em Cell}, 169(7):1177--1186.

\bibitem[Brown et~al., 2016]{brown2016transethnic}
Brown, B.~C., Ye, C.~J., Price, A.~L., and Zaitlen, N. (2016).
\newblock Transethnic genetic-correlation estimates from summary statistics.
\newblock {\em The American Journal of Human Genetics}, 99(1):76--88.

\bibitem[Bulik-Sullivan et~al., 2015]{bulik2015ld}
Bulik-Sullivan, B.~K., Loh, P.-R., Finucane, H.~K., Ripke, S., Yang, J., Patterson, N., Daly, M.~J., Price, A.~L., Neale, B.~M., of~the Psychiatric Genomics~Consortium, S. W.~G., et~al. (2015).
\newblock Ld score regression distinguishes confounding from polygenicity in genome-wide association studies.
\newblock {\em Nature Genetics}, 47(3):291--295.

\bibitem[Bycroft et~al., 2018]{bycroft2018uk}
Bycroft, C., Freeman, C., Petkova, D., Band, G., Elliott, L., Sharp, K., Motyer, A., Vukcevic, D., Delaneau, O., O'Connell, J., et~al. (2018).
\newblock The uk biobank resource with deep phenotyping and genomic data.
\newblock {\em Nature}, 562(7726):203--209.

\bibitem[Chen et~al., 2024]{chen2024fast}
Chen, T., Zhang, H., Mazumder, R., and Lin, X. (2024).
\newblock Fast and scalable ensemble learning method for versatile polygenic risk prediction.
\newblock {\em Proceedings of the National Academy of Sciences}, 121(33):e2403210121.

\bibitem[Chernozhukov et~al., 2017]{chernozhukov2017central}
Chernozhukov, V., Chetverikov, D., and Kato, K. (2017).
\newblock Central limit theorems and bootstrap in high dimensions.
\newblock {\em The Annals of Probability}, 45(4):2309.

\bibitem[Choi et~al., 2020]{choi2020tutorial}
Choi, S.~W., Mak, T. S.-H., and O’Reilly, P.~F. (2020).
\newblock Tutorial: a guide to performing polygenic risk score analyses.
\newblock {\em Nature Protocols}, 15(9):2759--2772.

\bibitem[Dobriban and Sheng, 2021]{dobriban2021distributed}
Dobriban, E. and Sheng, Y. (2021).
\newblock {Distributed linear regression by averaging}.
\newblock {\em The Annals of Statistics}, 49(2):918 -- 943.

\bibitem[Dobriban and Wager, 2018]{dobriban2018high}
Dobriban, E. and Wager, S. (2018).
\newblock High-dimensional asymptotics of prediction: Ridge regression and classification.
\newblock {\em The Annals of Statistics}, 46(1):247--279.

\bibitem[Fang and Koike, 2021]{fang2021high}
Fang, X. and Koike, Y. (2021).
\newblock High-dimensional central limit theorems by stein’s method.
\newblock {\em The Annals of Applied Probability}, 31(4):1660--1686.

\bibitem[Fu et~al., 2024]{fu2024uncertainty}
Fu, H., Huang, J., Fan, Z., and Zhao, B. (2024).
\newblock Uncertainty of high-dimensional genetic data prediction with polygenic risk scores.
\newblock {\em arXiv preprint arXiv:2412.20611}.

\bibitem[Ge et~al., 2019]{ge2019polygenic}
Ge, T., Chen, C.-Y., Ni, Y., Feng, Y.-C.~A., and Smoller, J.~W. (2019).
\newblock Polygenic prediction via bayesian regression and continuous shrinkage priors.
\newblock {\em Nature Communications}, 10(1):1--10.

\bibitem[Hu et~al., 2017]{hu2017leveraging}
Hu, Y., Lu, Q., Powles, R., Yao, X., Yang, C., Fang, F., Xu, X., and Zhao, H. (2017).
\newblock Leveraging functional annotations in genetic risk prediction for human complex diseases.
\newblock {\em PLoS Computational Biology}, 13(6):e1005589.

\bibitem[Jiang et~al., 2016]{jiang2016high}
Jiang, J., Li, C., Paul, D., Yang, C., and Zhao, H. (2016).
\newblock On high-dimensional misspecified mixed model analysis in genome-wide association study.
\newblock {\em The Annals of Statistics}, 44(5):2127--2160.

\bibitem[Jiang et~al., 2024]{jiang2024tuning}
Jiang, W., Chen, L., Girgenti, M.~J., and Zhao, H. (2024).
\newblock Tuning parameters for polygenic risk score methods using gwas summary statistics from training data.
\newblock {\em Nature Communications}, 15(1):24.

\bibitem[Jin et~al., 2025]{jin2025pennprs}
Jin, J., Li, B., Wang, X., Yang, X., Li, Y., Wang, R., Ye, C., Shu, J., Fan, Z., Xue, F., et~al. (2025).
\newblock Pennprs: a centralized cloud computing platform for efficient polygenic risk score training in precision medicine.
\newblock {\em medRxiv}, pages 2025--02.

\bibitem[Jin et~al., 2024]{jin2024mussel}
Jin, J., Zhan, J., Zhang, J., Zhao, R., O’Connell, J., Jiang, Y., Aslibekyan, S., Auton, A., Babalola, E., Bell, R.~K., et~al. (2024).
\newblock Mussel: Enhanced bayesian polygenic risk prediction leveraging information across multiple ancestry groups.
\newblock {\em Cell Genomics}, 4(4).

\bibitem[Kachuri et~al., 2024]{kachuri2024principles}
Kachuri, L., Chatterjee, N., Hirbo, J., Schaid, D.~J., Martin, I., Kullo, I.~J., Kenny, E.~E., Pasaniuc, B., in~Diverse Populations (PRIMED) Consortium Methods Working Group Auer Paul L. 20 Conomos Matthew P. 21 Conti David V. 22 23 Ding Yi 24 Wang Ying 19 25 26 Zhang Haoyu 27 28 Zhang Yuji~29, P. R.~M., Witte, J.~S., et~al. (2024).
\newblock Principles and methods for transferring polygenic risk scores across global populations.
\newblock {\em Nature Reviews Genetics}, 25(1):8--25.

\bibitem[Ledoit and P{\'e}ch{\'e}, 2011]{ledoit2011eigenvectors}
Ledoit, O. and P{\'e}ch{\'e}, S. (2011).
\newblock Eigenvectors of some large sample covariance matrix ensembles.
\newblock {\em Probability Theory and Related Fields}, 151(1-2):233--264.

\bibitem[Lennon et~al., 2024]{lennon2024selection}
Lennon, N.~J., Kottyan, L.~C., Kachulis, C., Abul-Husn, N.~S., Arias, J., Belbin, G., Below, J.~E., Berndt, S.~I., Chung, W.~K., Cimino, J.~J., et~al. (2024).
\newblock Selection, optimization and validation of ten chronic disease polygenic risk scores for clinical implementation in diverse us populations.
\newblock {\em Nature Medicine}, 30(2):480--487.

\bibitem[Ma and Zhou, 2021]{ma2021genetic}
Ma, Y. and Zhou, X. (2021).
\newblock Genetic prediction of complex traits with polygenic scores: a statistical review.
\newblock {\em Trends in Genetics}, 37(11):995--1011.

\bibitem[Mak et~al., 2017]{mak2017polygenic}
Mak, T. S.~H., Porsch, R.~M., Choi, S.~W., Zhou, X., and Sham, P.~C. (2017).
\newblock Polygenic scores via penalized regression on summary statistics.
\newblock {\em Genetic Epidemiology}, 41(6):469--480.

\bibitem[M{\'a}rquez-Luna et~al., 2021]{marquez2021incorporating}
M{\'a}rquez-Luna, C., Gazal, S., Loh, P.-R., Kim, S.~S., Furlotte, N., Auton, A., and Price, A.~L. (2021).
\newblock Incorporating functional priors improves polygenic prediction accuracy in uk biobank and 23andme data sets.
\newblock {\em Nature Communications}, 12(1):6052.

\bibitem[Opitz and Maclin, 1999]{opitz1999popular}
Opitz, D. and Maclin, R. (1999).
\newblock Popular ensemble methods: An empirical study.
\newblock {\em Journal of Artificial Intelligence Research}, 11:169--198.

\bibitem[Pain et~al., 2021]{pain2021evaluation}
Pain, O., Glanville, K.~P., Hagenaars, S.~P., Selzam, S., Fürtjes, A.~E., Gaspar, H.~A., Coleman, J. R.~I., Rimfeld, K., Breen, G., Plomin, R., Folkersen, L., and Lewis, C.~M. (2021).
\newblock Evaluation of polygenic prediction methodology within a reference-standardized framework.
\newblock {\em PLOS Genetics}, 17(5):1--22.

\bibitem[Pasaniuc and Price, 2017]{pasaniuc2017dissecting}
Pasaniuc, B. and Price, A.~L. (2017).
\newblock Dissecting the genetics of complex traits using summary association statistics.
\newblock {\em Nature Reviews Genetics}, 18(2):117--127.

\bibitem[Pattee and Pan, 2020]{pattee2020penalized}
Pattee, J. and Pan, W. (2020).
\newblock Penalized regression and model selection methods for polygenic scores on summary statistics.
\newblock {\em PLoS Computational Biology}, 16(10):e1008271.

\bibitem[Power et~al., 2015]{power2015polygenic}
Power, R.~A., Steinberg, S., Bjornsdottir, G., Rietveld, C.~A., Abdellaoui, A., Nivard, M.~M., Johannesson, M., Galesloot, T.~E., Hottenga, J.~J., Willemsen, G., et~al. (2015).
\newblock Polygenic risk scores for schizophrenia and bipolar disorder predict creativity.
\newblock {\em Nature Neuroscience}, 18(7):953--955.

\bibitem[Priv{\'e} et~al., 2022]{prive2022identifying}
Priv{\'e}, F., Arbel, J., Aschard, H., and Vilhj{\'a}lmsson, B.~J. (2022).
\newblock Identifying and correcting for misspecifications in gwas summary statistics and polygenic scores.
\newblock {\em Human Genetics and Genomics Advances}, 3(4):100136.

\bibitem[Priv{\'e} et~al., 2020]{prive2020ldpred2}
Priv{\'e}, F., Arbel, J., and Vilhj{\'a}lmsson, B.~J. (2020).
\newblock Ldpred2: better, faster, stronger.
\newblock {\em Bioinformatics}, 36(22-23):5424--5431.

\bibitem[Priv{\'e} et~al., 2019]{prive2019making}
Priv{\'e}, F., Vilhj{\'a}lmsson, B.~J., Aschard, H., and Blum, M.~G. (2019).
\newblock Making the most of clumping and thresholding for polygenic scores.
\newblock {\em The American Journal of Human Genetics}, 105(6):1213--1221.

\bibitem[Purcell et~al., 2009]{purcell2009common}
Purcell, S.~M., Wray, R., Stone, L., Visscher, M., O'Donovan, C., Sullivan, F., Sklar, P., Ruderfer, M., McQuillin, A., Morris, W., et~al. (2009).
\newblock Common polygenic variation contributes to risk of schizophrenia and bipolar disorder.
\newblock {\em Nature}, 460(7256):748--752.

\bibitem[Qian et~al., 2020]{qian2020fast}
Qian, J., Tanigawa, Y., Du, W., Aguirre, M., Chang, C., Tibshirani, R., Rivas, M.~A., and Hastie, T. (2020).
\newblock A fast and scalable framework for large-scale and ultrahigh-dimensional sparse regression with application to the uk biobank.
\newblock {\em PLoS Genetics}, 16(10):e1009141.

\bibitem[Rubio and Mestre, 2011]{rubio2011spectral}
Rubio, F. and Mestre, X. (2011).
\newblock Spectral convergence for a general class of random matrices.
\newblock {\em Statistics \& Probability Letters}, 81(5):592--602.

\bibitem[Song et~al., 2019]{song2019summaryauc}
Song, L., Liu, A., and Shi, J. (2019).
\newblock Summaryauc: a tool for evaluating the performance of polygenic risk prediction models in validation datasets with only summary level statistics.
\newblock {\em Bioinformatics}, 35(20):4038--4044.

\bibitem[Su et~al., 2024]{su2024exact}
Su, B., Sun, Q., Yang, X., and Zhao, B. (2024).
\newblock The exact risks of reference panel-based regularized estimators.
\newblock {\em arXiv preprint arXiv:2401.11359}.

\bibitem[Uffelmann et~al., 2021]{uffelmann2021genome}
Uffelmann, E., Huang, Q.~Q., Munung, N.~S., de~Vries, J., Okada, Y., Martin, A.~R., Martin, H.~C., Lappalainen, T., and Posthuma, D. (2021).
\newblock Genome-wide association studies.
\newblock {\em Nature Reviews Methods Primers}, 59(1):1--21.

\bibitem[Vilhj{\'a}lmsson et~al., 2015]{vilhjalmsson2015modeling}
Vilhj{\'a}lmsson, B.~J., Yang, J., Finucane, H.~K., Gusev, A., Lindstr{\"o}m, S., Ripke, S., Genovese, G., Loh, P.-R., Bhatia, G., Do, R., et~al. (2015).
\newblock Modeling linkage disequilibrium increases accuracy of polygenic risk scores.
\newblock {\em The American Journal of Human Genetics}, 97(4):576--592.

\bibitem[Wan et~al., 2022]{wan2022sociotechnical}
Wan, Z., Hazel, J.~W., Clayton, E.~W., Vorobeychik, Y., Kantarcioglu, M., and Malin, B.~A. (2022).
\newblock Sociotechnical safeguards for genomic data privacy.
\newblock {\em Nature Reviews Genetics}, 23(7):429--445.

\bibitem[Wang et~al., 2024]{wang2024integrating}
Wang, L., Khunsriraksakul, C., Markus, H., Chen, D., Zhang, F., Chen, F., Zhan, X., Carrel, L., Liu, D.~J., and Jiang, B. (2024).
\newblock Integrating single cell expression quantitative trait loci summary statistics to understand complex trait risk genes.
\newblock {\em Nature Communications}, 15(1):4260.

\bibitem[Wu et~al., 2023]{wu2023large}
Wu, C., Zhang, Z., Yang, X., and Zhao, B. (2023).
\newblock Large-scale imputation models for multi-ancestry proteome-wide association analysis.
\newblock {\em bioRxiv}, pages 2023--10.

\bibitem[Xue and Zhao, 2023]{xue2023high}
Xue, F. and Zhao, B. (2023).
\newblock High-dimensional statistical inference for linkage disequilibrium score regression and its cross-ancestry extensions.
\newblock {\em arXiv preprint arXiv:2306.15779}.

\bibitem[Yang and Zhou, 2020]{yang2020accurate}
Yang, S. and Zhou, X. (2020).
\newblock Accurate and scalable construction of polygenic scores in large biobank data sets.
\newblock {\em The American Journal of Human Genetics}, 106(5):679--693.

\bibitem[Yang and Zhou, 2022]{yang2022PGS}
Yang, S. and Zhou, X. (2022).
\newblock {PGS-server: accuracy, robustness and transferability of polygenic score methods for biobank scale studies}.
\newblock {\em Briefings in Bioinformatics}, 23(2):bbac039.

\bibitem[Yao et~al., 2015]{Yao_Zheng_Bai_2015}
Yao, J., Zheng, S., and Bai, Z. (2015).
\newblock {\em Large Sample Covariance Matrices and High-Dimensional Data Analysis}.
\newblock Cambridge Series in Statistical and Probabilistic Mathematics. Cambridge University Press.

\bibitem[Zhang et~al., 2023]{zhang2023new}
Zhang, H., Zhan, J., Jin, J., Zhang, J., Lu, W., Zhao, R., Ahearn, T.~U., Yu, Z., O’Connell, J., Jiang, Y., et~al. (2023).
\newblock A new method for multiancestry polygenic prediction improves performance across diverse populations.
\newblock {\em Nature Genetics}, 55(10):1757--1768.

\bibitem[Zhang et~al., 2024]{zhang2024ensemble}
Zhang, J., Zhan, J., Jin, J., Ma, C., Zhao, R., O’Connell, J., Jiang, Y., 23andMe Research~Team, Koelsch, B.~L., Zhang, H., et~al. (2024).
\newblock An ensemble penalized regression method for multi-ancestry polygenic risk prediction.
\newblock {\em Nature Communications}, 15(1):3238.

\bibitem[Zhao et~al., 2024a]{zhao2024estimating}
Zhao, B., Yang, X., and Zhu, H. (2024a).
\newblock Estimating trans-ancestry genetic correlation with unbalanced data resources.
\newblock {\em Journal of the American Statistical Association}, 119(546):839--850.

\bibitem[Zhao et~al., 2024b]{zhao2024blockwise}
Zhao, B., Zheng, S., and Zhu, H. (2024b).
\newblock On blockwise and reference panel-based estimators for genetic data prediction in high dimensions.
\newblock {\em The Annals of Statistics}, 52(3):948--965.

\bibitem[Zhao and Zhu, 2022]{zhao2019genetic}
Zhao, B. and Zhu, H. (2022).
\newblock On genetic correlation estimation with summary statistics from genome-wide association studies.
\newblock {\em Journal of the American Statistical Association}, 117(537):1--11.

\bibitem[Zhao et~al., 2024c]{zhao2024optimizing}
Zhao, Z., Gruenloh, T., Yan, M., Wu, Y., Sun, Z., Miao, J., Wu, Y., Song, J., and Lu, Q. (2024c).
\newblock Optimizing and benchmarking polygenic risk scores with gwas summary statistics.
\newblock {\em Genome Biology}, 25(1):260.

\bibitem[Zhao et~al., 2021]{zhao2021pumas}
Zhao, Z., Yi, Y., Song, J., Wu, Y., Zhong, X., Lin, Y., Hohman, T.~J., Fletcher, J., and Lu, Q. (2021).
\newblock Pumas: fine-tuning polygenic risk scores with gwas summary statistics.
\newblock {\em Genome Biology}, 22:1--19.

\end{thebibliography}

\newpage
\appendix 
%Notations for Appendices
\renewcommand{\thetable}{S.\arabic{table}}
\renewcommand{\thealgorithm}{S.\arabic{algorithm}}
\setcounter{algorithm}{0}
\renewcommand{\thefigure}{S.\arabic{figure}}
\setcounter{figure}{0}
\renewcommand{\thesection}{S.\arabic{section}}
\renewcommand{\thelemma}{S.\arabic{lemma}}
\renewcommand{\theproposition}{S.\arabic{proposition}}
\renewcommand{\thecondition}{S.\arabic{condition}}

\vspace{30pt}
\noindent{\bf \LARGE Supplementary material}

\section{Individual-level data-based algorithms}\label{sec:algo}
We provide pseudo-code for individual-level data-based model training in Algorithm \ref{alg:ind}. 
Additionally, we outline the ensemble learning approach in Algorithm \ref{alg:ensemble_ind} and the model training with multi-ancestry data resources in Algorithm \ref{alg:multi_ind}. 

\begin{algorithm}
\caption{Individual-level data-based model training}\label{alg:ind}
\begin{algorithmic}
\Require Individual data $\Xb \in \RR^{n \times p}, \yb \in \RR^{n}$, $\Wb^{\T} \Wb$,  and hyperparameter $\theta$.
\vspace{1mm}

\State $\Qb \gets \diag \{q_1, q_2, \cdots q_n\}, \quad q_{i} \overset{i.i.d.}{\sim} \textnormal{Bernoulli}(n^{\train}/n),$  \hfill 		\texttt{//} Sample $n$ Bernoulli random variable.
\vspace{1mm}

\State $\Xb^{\train} \gets \Qb \Xb,$  \hfill 		\texttt{//} Use $n^{\train}$ rows of $\Xb$ for training.
\vspace{1mm}

\State $\Xb^{\valid} \gets (\Ib_n - \Qb) \Xb,$  \hfill 		\texttt{//} Use the remaining $n^{\valid} = n - n^{\train}$ rows of $\Xb$ for validation.
\vspace{1mm}

\State $\hat{\bbeta}_{\rm G}(\theta) \gets \Ab(\Wb^{\T} \Wb, \theta) \Xb^{\train}{}^{\T} \yb^{\train},$  \hfill 		\texttt{//} Obtain the estimator.
\vspace{1mm}

\State $R_{\rm ind, G}^2(\theta) \gets ({n^{\valid}} \left/{\|\yb^{\valid}\|_{2}^2})\right. \cdot {\left\langle {\Xb^{\valid}{}^{\T}\yb^{\valid}}, \hat{\bbeta}_{\rm G}(\theta) \right\rangle^2}\left/({n^{\valid} \cdot \| \hat{\bbeta}_{\rm G}(\theta) \|_{\bSigma}^2}), \right.$\vspace{1mm}

\Statex \hfill \texttt{//} Compute the $R_{\rm ind, G}^2(\theta)$ in \eqref{eqn:R2_ind_def}.
\vspace{1mm}

\State $\theta^{*}_{\rm ind, G} \gets \max_{\theta} R^2_{\rm ind, G}(\theta),$  \hfill 		\texttt{//} Choose the optimal hyperparameter.
\vspace{1.5mm}

\Return $R_{\rm ind, G}^2(\theta)$ and $\theta^{*}_{\rm ind, G}$.

\end{algorithmic}
\end{algorithm}

\begin{algorithm}
\caption{Individual-level data-based ensemble learning}\label{alg:ensemble_ind}
\begin{algorithmic}
\Require Individual data $\Xb \in \RR^{n \times p}, \yb \in \RR^{n}$, $\Wb^{\T} \Wb$, and  the hyperparameter $\Theta = \{\omega_{j}, \Theta_{j}\}_{j=1}^{k}$. 
\vspace{1mm}

\State $\Qb \gets \diag \{q_1, q_2, \cdots q_n\}, \quad q_{i} \overset{i.i.d.}{\sim} \textnormal{Bernoulli}(n^{\train}/n),$  \hfill 		\texttt{//} Sample $n$ Bernoulli random variable.
\vspace{1mm}

\State $\Xb^{\train} \gets \Qb \Xb,$  \hfill 		\texttt{//} Use $n^{\train}$ rows of $\Xb$ for training.
\vspace{1mm}

\State $\Xb^{\valid} \gets (\Ib_n - \Qb) \Xb,$  \hfill 		\texttt{//} Use the remaining $n^{\valid} = n - n^{\train}$ rows of $\Xb$ for validation.
\vspace{1mm}

\State $\hat{\bbeta}_{\rm E}(\Theta) \gets \sum_{j=1}^{k} \omega_{j} \Ab_{j}(\Wb^{\T} \Wb, \Theta_{j}) \Xb^{\train}{}^{\T} \yb^{\train},$  \hfill 		\texttt{//} Obtain the estimator.
\vspace{1mm}

\State $R_{\rm ind, E}^2(\Theta) \gets ({n^{\valid}} \left/{\|\yb^{\valid}\|_{2}^2})\right. {\left\langle {\Xb^{\valid}{}^{\T}\yb^{\valid}}, \hat{\bbeta}_{\rm E}(\Theta) \right\rangle^2}\left/({n^{\valid} \cdot \| \hat{\bbeta}_{\rm E}(\Theta) \|_{\bSigma}^2}), \right.$  \vspace{1mm}

\Statex \hfill \texttt{//} Compute the $R_{\rm ind, E}^2(\theta)$ in \eqref{eqn:R2_ind_def}.
\vspace{1mm}

\State $\Theta^{*}_{\rm ind, E} \gets \max_{\Theta} R_{\rm ind, E}^2(\Theta),$  \hfill 		\texttt{//} Choose the optimal hyperparameter.
\vspace{1.5mm}

\Return $R_{\rm sum, E}^2(\Theta)$ and $\Theta^{*}_{\rm ind, E}$.

\end{algorithmic}
\end{algorithm}

\begin{algorithm}
\caption{Individual-level data-based model training with multi-ancestry data resources}\label{alg:multi_ind}
\begin{algorithmic}
\Require Individual data $\Xb_{j}\in \RR^{n \times p}, \yb_{j} \in \RR^{n}$, $\Wb_{j}^{\T} \Wb_{j}$ for $1\leq j \leq K$, and  the hyperparameter $\Theta = \{\omega_{j}, \Theta_{j}\}_{j=1}^{k}$. 
\vspace{1mm}

\For{$j \gets 1$ to $K$} 

\State $\Qb_{j} \gets \diag \{q_1, q_2, \cdots q_n\}, \quad q_{i} \overset{i.i.d.}{\sim} \textnormal{Bernoulli}(n^{\train}/n),$
 \vspace{1mm}

\Statex \hfill \texttt{//} Sample $n$ Bernoulli random variable.
\vspace{1mm}

\State $\Xb_{j}^{\train} \gets \Qb_{j} \Xb_{j},$  \hfill 		\texttt{//} Use $n^{\train}$ rows of $\Xb$ for training.
\vspace{1mm}
\EndFor

\vspace{1mm}

\State $\Xb_{1}^{\valid} \gets (\Ib_n - \Qb_{1}) \Xb_{1},$
 \vspace{1mm}

\Statex \hfill \texttt{//} Use the remaining $n^{\valid} = n - n^{\train}$ rows of $\Xb$ for validation in the population 1.
\vspace{1mm}

\State $\hat{\bbeta}_{\rm MA}(\Theta) \gets \sum_{j=1}^{k} \omega_{j} \Ab_{j}(\Wb_{j}^{\T} \Wb_{j}, \Theta_{j}) \Xb_{j}^{\train}{}^{\T} \yb_{j}^{\train},$  \hfill 		\texttt{//} Obtain the estimator.
\vspace{1mm}

\State $R_{\rm ind, MA}^2(\Theta) \gets ({n^{\valid}} \left/{\|\yb^{\valid}\|_{2}^2})\right. {\left\langle {\Xb_{1}^{\valid}{}^{\T}\yb_{1}^{\valid}}, \hat{\bbeta}_{\rm MA}(\Theta) \right\rangle^2}\left/({n^{\valid} \cdot \| \hat{\bbeta}_{\rm MA}(\Theta) \|_{\bSigma_1}^2}), \right.$  \vspace{1mm}

\Statex \hfill \texttt{//} Compute the $R_{\rm ind, MA}^2(\theta)$ in \eqref{eqn:R2_ind_def}.
\vspace{1mm}

\State $\Theta^{*}_{\rm ind, MA} \gets \max_{\Theta} R^2_{\rm ind, G}(\Theta),$  \hfill 		\texttt{//} Choose the optimal hyperparameter.
\vspace{1.5mm}

\Return $R_{\rm ind, MA}^2(\Theta)$ and $\Theta^{*}_{\rm ind, MA}$.

\end{algorithmic}
\end{algorithm}

\newpage

%\section{Preliminary on random matrix theory and deterministic equivalence}
\section{Preliminary results}
\label{sec:prelim}
In this section, we provide some details about the calculus of deterministic equivalents from random matrix theory \citep{dobriban2021distributed}. 
Let $\hat{\bSigma}_{n} = \Xb^{\T} \Xb/n$ with each row of $\Xb \sim N(0, \bSigma)$. 
We take $n,p \to\infty$ proportionally, and the simplest example of equivalence is $\hat{\bSigma}_{n} \asymp \bSigma$. 
Here we abuse notation by using $\hat{\bSigma}_{n}$ to denote the covariance matrix $\Xb^{\T} \Xb / n$ for any matrix $\Xb \in \mathbb{R}^{n \times p}$ satisfying Condition \ref{cond-X}, as all such matrices share the same asymptotic limit as $n, p \to \infty$ proportionally. For example, we may abbreviate $\Wb^{\T} \Wb / n_w$ as $\hat{\bSigma}_{n_w}$ in the following sections. 

We present the generalized Marchenko-Pastur Law (e.g., see Theorem 1 in \cite{rubio2011spectral}) below, as it will be frequently used in our proof.

\begin{theorem}[Generalized Marchenko-Pastur Law]
\label{thm:gen_MP_law}
Let $\Xb \in \RR^{n \times p}$ be a random matrix satisfying Conditions \ref{cond-np-ratio} and \ref{cond-X} and $\Ab$ be a $p \times p$ nonnegative definite matrix. 
Then, with probability one, for each $\theta \in \RR_{+}$, as $n, p \to \infty$ proportionally,  we have 
\begin{align} \label{eqn:first_generalized_MP}
    \left(\Ab + \hat{\bSigma}_{n} + \theta \Ib_{p}\right)^{-1} \asymp \left(\Ab + \tau_{n}(\theta)\bSigma + \theta \Ib_{p}\right)^{-1},
\end{align}
where $\tau_{n}(\theta)$ is defined as the solution in $\CC_{+}$ to the fixed point equation
\begin{equation*}
\tau_{n}(\theta)^{-1}=1+\frac{1}{n}\tr\left[\bSigma (\Ab + \tau_{n}(\theta) \bSigma + \theta \Ib_p)^{-1}\right].
\end{equation*}
\end{theorem}

When $\Ab$ is the zero matrix $\rm O_p$, we have the following corollary.
\begin{corollary}
Let $\Xb \in \RR^{n \times p}$ be a random matrix satisfying Conditions \ref{cond-np-ratio} and \ref{cond-X}.
With probability one, for each $\theta > 0$, as $n, p \to \infty$ proportionally, we have 
\begin{align*} 
    \left(\hat{\bSigma}_{n} + \theta \Ib_{p}\right)^{-1} \asymp \left(\tau_{n}(\theta)\bSigma + \theta \Ib_{p}\right)^{-1},
\end{align*}
where $\tau_{n}(\theta)$ is defined as the solution to the fixed point equation
\begin{equation*}
\tau_{n}(\theta)^{-1}=1+\frac{1}{n}\tr\left[\bSigma (\tau_{n}(\theta) \bSigma + \theta \Ib_p)^{-1}\right]. 
\end{equation*}
When $\bSigma = \Ib_{p}$, $\tau_{n}(\theta)$ has closed-form
\begin{align}\label{eqn:first_taupz}
    \tau_{n}(\theta) = \frac{1 - \theta - \gamma + \sqrt{(1 - \theta - \gamma)^2 + 4 \theta}}{2}.
\end{align}
\end{corollary}

To prove Lemma \ref{lemma:DE_second_order}, we also need the following lemma from \cite{dobriban2021distributed}.
\begin{lemma}[Differentiation Rule for Deterministic Equivalence]
\label{lemma:diff_deter_equi}
Suppose $\Tb = (\Tb_n)_{n \geq 0}$ and $\Qb = (\Qb_n)_{n \geq 0}$ are two (deterministic or random) matrix sequences of growing dimensions such that $f(z, \Tb_n) \asymp g(z, \Qb_n)$, where the entries of $f$ and $g$ are analytic functions in $z \in D$ and $D$ is an open connected subset of $\mathbb{C}$. 
Then we have $$f'(z, \Tb_n) \asymp g'(z, \Qb_n)$$ for $z \in D$, where the derivatives are entry-wise with respect to $z$.
\end{lemma}

\subsection{Proof of Lemma \ref{lemma:DE_second_order}}
%\label{sec:proof_of_second_mom}

We now provide the proof of Lemma \ref{lemma:DE_second_order}. 
\begin{proof}[Proof of Lemma \ref{lemma:DE_second_order}]
Let $\Ab = t \cdot \bSigma$ be an Hermitian nonnegative definite matrix. By Theorem \ref{thm:gen_MP_law},
we have \begin{equation*}
  G(t,\theta):=  (\hat{\bSigma}_{n} + t \bSigma + \theta \Ib_{p})^{-1} \asymp ( \tau_{n}(t,\theta) \bSigma + t\bSigma + \theta \Ib_{p})^{-1}=:D(t,\theta),
\end{equation*}
where $$\tau_{n}^{-1}(t,\theta) = 1+ n^{-1}\tr\left[\bSigma (t\bSigma+\tau_{n}(t,\theta) \bSigma + \theta I_p)^{-1}\right].$$
By Lemma \ref{lemma:diff_deter_equi}, we have 
\begin{equation*}
    \begin{aligned}
     & G(0,\theta) \bSigma G(0,\theta) = - \frac{\partial G}{\partial t}\Big|_{t=0}\asymp -\frac{\partial D}{\partial t}\Big|_{t=0} = D(0,\theta) \left(\frac{\partial \tau_{n}}{\partial t}\Big|_{t=0}\bSigma + \bSigma   \right) D(0,\theta). 
    \end{aligned}
\end{equation*}
Take the derivative of both sides of Equation \eqref{eqn:tau} with respect to $t$, we have 
\begin{equation*}
     \frac{\partial \tau_{n}(t,\theta)}{\partial t} 
     = \frac{\tau_{n}^2}{n}  \tr\left[\bSigma (t\bSigma+\tau_{n} \bSigma+ \theta \Ib_{p})^{-1} \left( \bSigma + \frac{\partial \tau_{n}}{\partial t}\bSigma \right)(t\bSigma+\tau_{n} \bSigma + \theta \Ib_{p})^{-1}\right].
\end{equation*}
It follows that 
\begin{equation*}
    {\frac{\partial \tau_{n}}{\partial t}} \left( 1 + \frac{\partial \tau_{n}}{\partial t} \right)^{-1} = \frac{\tau_{n}^2}{n}\tr\left[\bSigma (t\bSigma+\tau_{n} \bSigma+ \theta \Ib_{p})^{-1} \bSigma (t\bSigma+\tau_{n} \bSigma+ \theta \Ib_{p})^{-1} \right].
\end{equation*}
Let $t = 0$, we have
\begin{align*}
    {\frac{\partial \tau_{n}}{\partial t}\Big|_{t=0}} \left( 1 + \frac{\partial \tau_{n}}{\partial t}\Big|_{t=0} \right)^{-1} 
    =\ & \frac{\tau_{n}^2(\theta)}{n}\tr\left[\bSigma (\tau_{n}(\theta) \bSigma + \theta \Ib_{p})^{-1} \bSigma (\tau_{n}(\theta) \bSigma + \theta \Ib_{p})^{-1} \right]
    \\
    =\ & \frac{\tau_{n}^2(\theta)}{n}\tr\left[(\tau_{n}(\theta) \bSigma + \theta \Ib_{p})^{-2} \bSigma^2 \right].
\end{align*}
Therefore,
\begin{equation*}
    \frac{\partial \tau_{n}}{\partial t}\Big|_{t=0} = \left(1-\frac{\tau_{n}^2(\theta)}{n}\tr\left[(\tau_{n}(\theta) \bSigma + \theta \Ib_{p})^{-2} \bSigma^2 \right]\right)^{-1} \frac{\tau_{n}^2(\theta)}{n}\tr\left[(\tau_{n}(\theta) \bSigma + \theta \Ib_{p})^{-2} \bSigma^2 \right] =: \rho_{n}(\theta).
\end{equation*}
We conclude
\begin{equation*}
\begin{split}
    %&
    \left(\hat{\bSigma}_{n} + \theta \Ib_{p}\right)^{-1} \bSigma \left( \hat{\bSigma}_{n} + \theta \Ib_{p}\right)^{-1} 
    %\\
    %&  \qquad
    \asymp (\rho_{n}(\theta) + 1) \cdot \left( \tau_{n}(\theta)\bSigma + \theta\Ib_{p}\right)^{-1} \bSigma \left(\tau_{n}(\theta)\bSigma + \theta\Ib_{p}\right)^{-1}.
\end{split}
\end{equation*}
When $\bSigma = \Ib_{p}$, $\tau_{n}(t, \theta)$ satisfies the equation
\begin{align*}
    \tau_{n}^{-1}(t,\theta) = 1+ \frac{\gamma}{t +\tau_{n}(t,\theta) + \theta}. 
\end{align*}
Therefore, $\tau_{n}(t, \theta)$ admits the closed-form
\begin{align*}
    \tau_{n}(t, \theta) = \frac{1 - t - \gamma - \theta + \sqrt{(1 - t - \gamma - \theta)^2 + 4(t + \theta)}}{2}.
\end{align*}
Therefore, we have 
\begin{align*}
    \rho_{n}(\theta) = \frac{\partial \tau_{n}}{\partial t}\Big|_{t=0} = \frac{1 + \gamma + \theta - \sqrt{(1 - \theta - \gamma)^2 + 4 \theta}}{2 \sqrt{(1 - \theta - \gamma)^2 + 4 \theta}}.
\end{align*}
\end{proof}

\subsection{Proof of Lemma \ref{lemma:second_moment}}
\label{sec:proof_of_second_mom}

The proof is organized into two parts.
In the first part, we establish the convergence of ${1}/{p} \cdot \tr\left(\Cb \hat{\bSigma}_{n}^2 \right)$.
In the second part, we derive the closed-form expression for the limit.
Notably, there is a key distinction between ${1}/{p} \cdot \tr\left(\Cb \hat{\bSigma}_{n}^2 \right)$ and the trace functional considered in Lemma B.26 in \cite{bai2010spectral}. Lemma B.26 in \cite{bai2010spectral} only guarantees the convergence of the first-order trace ${1}/{p} \cdot \tr\left(\Cb \hat{\bSigma}_{n} \right)$. 
Since our lemma requires a second-order concentration inequality, Lemma B.26 from \cite{bai2010spectral} cannot be directly applied in our proof.

Motivated by the proof of Lemma S.16 in Section S.8.8 of \cite{fu2024uncertainty}, we have the following concentration inequality
\begin{align*}
    \PP \left( \left| \frac{1}{p} \tr\left(\Cb \hat{\bSigma}_{n}^2 \right) - \frac{1}{p} \cdot \EE \tr\left(\Cb \hat{\bSigma}_{n}^2 \right) \right| < O_{p}(p^{-1/2 + \delta}) \right) \geq 1 - O_p(p^{-2 \delta}),
\end{align*}
which provides the concentration inequality for the second moment of the quadratic form with high probability.
Here, the expectation is taken with respect to $\Xb$.
This further implies that, as $n, p \to \infty$ proportionally, ${1}/{p} \cdot \tr\left(\Cb \hat{\bSigma}_{n} \right)$ converges to a constant
\begin{align} \label{eqn:convg}
    {1}/{p} \cdot \tr\left(\Cb \hat{\bSigma}_{n}^{2} \right) - {1}/{p} \cdot \EE \tr\left(\Cb \hat{\bSigma}_{n}^{2} \right) \overset{p}{\to} 0.
\end{align}

Next, to derive the limit of ${1}/{p} \cdot \tr\left(\Cb \hat{\bSigma}_{n}^{2} \right)$, we use 
Lemma S.9 in \cite{fu2024uncertainty}, which computes the closed-form of a more general trace functional. 
%of ${1}/{p} \cdot \EE \tr\left(\Cb \hat{\bSigma}_{n}^{2} \right)$ 
%for any finite $n$ and $p$. 

\begin{lemma}[Lemma S.9 in \cite{fu2024uncertainty}]
For any deterministic symmetric matrices 
\(\Pb, \Qb, \Rb \in \mathbb{R}^{p \times p}\), for each entry $i,j$, we have 
\begin{align*}
\bigl[\mathbb{E} \bigl(\Pb \bSigma^{-1/2} \hat{\bSigma}_{n} \bSigma^{-1/2}  \,\Qb\, \bSigma^{-1/2} &\hat{\bSigma}_{n} \bSigma^{-1/2} \,\Rb\bigr)\bigr]_{i,j} 
\\
=\ &
\frac{1}{n} \Bigl\{\,
\mathbb{E}\bigl(x_0^4 - 3\bigr)\,\bigl(\Pb\,\mathrm{diag}(\Qb)\,\Rb\bigr)_{i,j}
\;+\;(n+1)\,\bigl(\Pb\,\Qb\,\Rb\bigr)_{i,j}
\;+\;\mathrm{Tr}(\Qb)\,\bigl(\Pb\,\Rb\bigr)_{i,j}
\Bigr\}.
\end{align*}
\end{lemma}
Let $\Pb = \Rb = \bSigma^{1/2}$ and $\Qb = \bSigma^{1/2} \Cb \bSigma^{1/2}$. As $n, p \to \infty$ proportionally, we obtain
\begin{align*}
    \frac{1}{p} \tr\left(\Cb \hat{\bSigma}_{n}^2 \right)
    =\ & \frac{1}{p} \cdot \frac{1}{n} \Bigl\{\,
    \mathbb{E}\bigl(x_0^4 - 3\bigr)\, \tr \bigl(\Pb\,\mathrm{diag}(\Qb)\,\Rb\bigr)
    \;+\;(n+1)\, \tr \bigl(\Pb\,\Qb\,\Rb\bigr)
    \;+\; \tr(\Qb)\, \tr \bigl(\Pb\,\Rb\bigr)
    \Bigr\}. 
    \\
    =\ & \frac{1}{n} \tr \left(\bSigma \right) \cdot \frac{1}{p} \tr \left(\Cb \bSigma \right) + \frac{1}{p} \tr \left(\Cb \bSigma^2 \right) + o_p(1).
\end{align*}
This completes the proof of Lemma \ref{lemma:second_moment}.

\section{Proof of Theorem \ref{thm:reference_panel}}
\label{sec:proof_thm_reference_panel}

In this section, we present the proof of Theorem \ref{thm:reference_panel}. The proof is organized into two parts: (i) deriving the limit of $R^2_{\rm sum, R}(\theta)$ from Algorithm \ref{alg:sum}, and (ii) calculating the limit of $R^2_{\rm ind, R}(\theta)$ from Algorithm \ref{alg:ind}.

\paragraph{\texorpdfstring{Part I: The limit of $\mathbf{R^2_{\rm sum, R}(\theta)}$.}{Summery}}
Recall that 
\begin{align*}
    R^2_{\rm sum, R}(\theta) = \frac{n^{\valid}}{\|\yb^{\valid}\|_{2}^2} \cdot \frac{\left \langle {\sbb^{\valid}}, (\Wb^{\T} \Wb + \theta n_w \Ib_{p})^{-1} \sbb^{\train} \right \rangle^{2} }{n^{\valid}{}^{2} \cdot \| (\Wb^{\T} \Wb + \theta n_w \Ib_{p})^{-1} \sbb^{\train} \|_{\bSigma}^{2}} = \frac{n^{\valid}}{\|\yb^{\valid}\|_{2}^2} \cdot \frac{\left( \Omega_{\rm sum}^{(1)} \right)^{2}}{\Omega_{\rm sum}^{(2)}}, 
\end{align*}
where $\Omega_{\rm sum}^{(1)}$ and $\Omega_{\rm sum}^{(2)}$ are defined as
\begin{align*}
    \Omega_{\rm sum}^{(1)} = \frac{1}{n^{\valid}} \cdot \left \langle {\sbb^{\valid}}, (\Wb^{\T} \Wb + \theta n_w \Ib_{p})^{-1} \sbb^{\train} \right \rangle
    \quad \mbox{and} \quad
    \Omega_{\rm sum}^{(2)} = \left\| (\Wb^{\T} \Wb + \theta n_w \Ib_{p})^{-1} \sbb^{\train} \right\|_{\bSigma}^2, 
\end{align*}
respectively. For $\Omega_{\rm sum}^{(1)}$, by plugging in the definition of $\sbb^{\train}$ and $\sbb^{\valid}$ and the fact that $n^{\train} + n^{\valid} = n$, we have
\begin{align*}
    \Omega_{\rm sum}^{(1)} =\ & \frac{1}{n^{\valid}} \cdot \left \langle {\sbb^{\valid}}, (\Wb^{\T} \Wb + \theta n_w \Ib_{p})^{-1} \sbb^{\train} \right \rangle
    \\
    =\ & \frac{1}{n^{\valid}} \cdot \left[ \frac{n - n^{\train}}{n} \Xb^{\T}\yb - \sqrt{\frac{n^{\train} (n - n^{\train})}{n^2}} (\Cov(\Xb^{\T} \yb))^{1/2} \hb \right]^{\T}  
    \\
    &\qquad \left( \Wb^{\T} \Wb + \theta n_w \Ib_{p} \right)^{-1} \left[ \frac{n^{\train}}{n} \Xb^{\T}\yb + \sqrt{\frac{n^{\train} (n - n^{\train})}{n^2}} (\Cov(\Xb^{\T} \yb))^{1/2} \hb \right]
    \\
    =\ & \frac{1}{n_w \cdot n^{\valid}} \cdot \left[  \frac{n - n^{\train}}{n} \Xb^{\T}\yb - \sqrt{\frac{n^{\train} (n - n^{\train})}{n^2}} (\Cov(\Xb^{\T} \yb))^{1/2}  \hb \right]^{\T} 
    \\
    &\qquad \left( \Wb^{\T} \Wb/n_w + \theta \Ib_{p} \right)^{-1} \cdot \left[ \frac{n^{\train}}{n} \Xb^{\T}\yb + \sqrt{\frac{n^{\train} (n - n^{\train})}{n^2}} (\Cov(\Xb^{\T} \yb))^{1/2} \hb \right]
    \\
    \overset{(a)}{=}\ & \frac{1}{n_w \cdot n^{\valid}} \cdot \left[  \frac{n - n^{\train}}{n} \Xb^{\T}\yb - \sqrt{\frac{n^{\train} (n - n^{\train})}{n^2}} (\Cov(\Xb^{\T} \yb))^{1/2}  \hb \right]^{\T}
    \\
    &\qquad \left( \tau_{n_w}(\theta) \bSigma + \theta \Ib_p \right)^{-1} \left[ \frac{n^{\train}}{n} \Xb^{\T}\yb + \sqrt{\frac{n^{\train} (n - n^{\train})}{n^2}} (\Cov(\Xb^{\T} \yb))^{1/2} \hb \right] + o_p(1),
\end{align*}
where Equation (a) follows from Lemma \ref{lemma:DE_second_order}. 
We further simplifies $\Omega_{\rm sum}^{(1)}$ as follows
\begin{align*}
    \Omega_{\rm sum}^{(1)} =\ & \frac{n}{n_w} \cdot \left[  \frac{n - n^{\train}}{n} \frac{\Xb^{\T}\Xb \bbeta}{n} + \frac{n - n^{\train}}{n} \frac{\Xb^{\T}\bepsilon}{n} - \sqrt{\frac{n^{\train} (n - n^{\train})}{n^2}} (\Cov(\Xb^{\T} \yb))^{1/2}  \frac{\hb}{n} \right]^{\T} \left( \tau_{n_w}(\theta) \bSigma + \theta \Ib_p \right)^{-1}
    \\
    &\qquad \left[ \frac{n^{\train}}{n^{\valid}} \frac{\Xb^{\T}\Xb \bbeta}{n} + \frac{n^{\train}}{n^{\valid}} \frac{\Xb^{\T} \bepsilon}{n} + \sqrt{\frac{n^{\train}}{n - n^{\train}}} (\Cov(\Xb^{\T} \yb))^{1/2}  \frac{\hb}{n} \right] + o_p(1)
    \\
    \overset{(b)}{=}\ & \frac{n}{n_w} \cdot \left(  \frac{n - n^{\train}}{n} \frac{\Xb^{\T}\Xb \bbeta}{n} \right)^{\T} \left( \tau_{n_w}(\theta) \bSigma + \theta \Ib_p \right)^{-1} \left( \frac{n^{\train}}{n^{\valid}} \frac{\Xb^{\T}\Xb \bbeta}{n} \right)
    \\
    &+ \frac{n}{n_w} \cdot \left( \frac{n - n^{\train}}{n} \frac{\Xb^{\T}\bepsilon}{n} \right)^{\T} \left( \tau_{n_w}(\theta) \bSigma + \theta \Ib_p \right)^{-1} \left( \frac{n^{\train}}{n^{\valid}} \frac{\Xb^{\T} \bepsilon}{n} \right)
    \\
    &- \frac{n}{n_w} \cdot \left[ \sqrt{\frac{n^{\train} (n - n^{\train})}{n^2}} (\Cov(\Xb^{\T} \yb))^{1/2}  \frac{\hb}{n} \right]^{\T} \left( \tau_{n_w}(\theta) \bSigma + \theta \Ib_p \right)^{-1} \left[ \sqrt{\frac{n^{\train}}{n - n^{\train}}} (\Cov(\Xb^{\T} \yb))^{1/2}  \frac{\hb}{n} \right] \\
    & + o_p(1)
    \\
    =\ & \Omega_{\rm sum}^{(2)} + \Omega_{\rm sum}^{(3)} - \Omega_{\rm sum}^{(4)} + o_p(1),
\end{align*}
where Equation (b) follows from Lemma B.26 in \cite{bai2010spectral} and Lemma 5 in \cite{zhao2024estimating}, so do Equations (c), (d), (e), and (f) below. 
We have % $\Omega_{\rm sum}^{(2)}$, $\Omega_{\rm sum}^{(3)}$ and $\Omega_{\rm sum}^{(4)}$ separately. 
\begin{align*}
    \Omega_{\rm sum}^{(2)} =\ & \frac{n}{n_w} \cdot \left(  \frac{n - n^{\train}}{n} \frac{\Xb^{\T}\Xb \bbeta}{n} \right)^{\T} \left( \tau_{n_w}(\theta) \bSigma + \theta \Ib_p \right)^{-1} \left( \frac{n^{\train}}{n^{\valid}} \frac{\Xb^{\T}\Xb \bbeta}{n} \right)
    \\
    =\ & \frac{n^{\train}}{n_w} \cdot \left(  \frac{\Xb^{\T}\Xb \bbeta}{n} \right)^{\T} \left( \tau_{n_w}(\theta) \bSigma + \theta \Ib_p \right)^{-1} \left( \frac{\Xb^{\T}\Xb \bbeta}{n} \right)
    \\
    \overset{(c)}{=}\ & \frac{n^{\train}}{n_w} \cdot \frac{\kappa \sigma_{\bbeta}^2}{p} \cdot \tr \left[ \hat{\bSigma}_{n}^2 \left( \tau_{n_w}(\theta) \bSigma + \theta \Ib_p \right)^{-1} \right] + o_p(1),
\end{align*}
\begin{align*}
    \Omega_{\rm sum}^{(3)} =\ & \frac{n}{n_w} \cdot \left( \frac{n - n^{\train}}{n} \frac{\Xb^{\T}\bepsilon}{n} \right)^{\T} \left( \tau_{n_w}(\theta) \bSigma + \theta \Ib_p \right)^{-1} \left( \frac{n^{\train}}{n^{\valid}} \frac{\Xb^{\T} \bepsilon}{n} \right)
    \\
    =\ & \frac{n^{\train}}{n_w} \cdot \frac{1}{n^2} \bepsilon^{\T} \Xb \left( \tau_{n_w}(\theta) \bSigma + \theta \Ib_p \right)^{-1} \Xb^{\T} \bepsilon
    \\
    \overset{(d)}{=}\ & \frac{n^{\train}}{n_w} \cdot \sigma_{\bepsilon}^2 \cdot \frac{1}{n^2} \tr \left[ \Xb \left( \tau_{n_w}(\theta) \bSigma + \theta \Ib_p \right)^{-1} \Xb^{\T} \right] + o_p(1)
    \\
    =\ & \frac{n^{\train}}{n_w} \cdot  \frac{\sigma_{\bepsilon}^2}{n} \cdot \tr \left[ \hat{\bSigma}_{n} \left( \tau_{n_w}(\theta) \bSigma + \theta \Ib_p \right)^{-1} \right] + o_p(1),
\end{align*}
and 
\begin{align*}
    \Omega_{\rm sum}^{(4)} =\ & \frac{n}{n_w} \cdot \left[ \sqrt{\frac{n^{\train} (n - n^{\train})}{n^2}} (\Cov(\Xb^{\T} \yb))^{1/2}  \frac{\hb}{n} \right]^{\T} \left( \tau_{n_w}(\theta) \bSigma + \theta \Ib_p \right)^{-1} \left[ \sqrt{\frac{n^{\train}}{n - n^{\train}}} (\Cov(\Xb^{\T} \yb))^{1/2}  \frac{\hb}{n} \right]
    \\
    =\ & \frac{n}{n_w} \cdot \frac{n^{\train}}{n} \cdot \left[ (\Cov(\Xb^{\T} \yb))^{1/2}  \frac{\hb}{n} \right]^{\T} \left( \tau_{n_w}(\theta) \bSigma + \theta \Ib_p \right)^{-1} \left[ (\Cov(\Xb^{\T} \yb))^{1/2}  \frac{\hb}{n} \right]
    \\
    \overset{(e)}{=}\ & \frac{n^{\train}}{n_w} \cdot \tr\left[ \left( \tau_{n_w}(\theta) \bSigma + \theta \Ib_p \right)^{-1} \frac{1}{n^2} \Cov(\Xb^{\T} \yb)\right] + o_p(1).
\end{align*}
Plugging in the closed-form of the covariance matrix 
\begin{align} \label{eqn:covariance}
    \Cov(\Xb^{\T} \yb) = \left(\Xb^{\T} \yb - n \bSigma \bbeta \right) \left(\Xb^{\T} \yb - n \bSigma \bbeta \right)^{\T}
\end{align}
into the equation above, we have
\begin{align*}
    \Omega_{\rm sum}^{(4)} =\ & \frac{n^{\train}}{n_w} \tr\left[ \left( \tau_{n_w}(\theta) \bSigma + \theta \Ib_p \right)^{-1}\left( \frac{1}{n} \Xb^{\T} \yb - \bSigma \bbeta \right) \left( \frac{1}{n} \Xb^{\T} \yb - \bSigma \bbeta \right)^{\T}\right] + o_p(1)
    \\
    =\ & \frac{n^{\train}}{n_w} \tr\left[ \left( \frac{1}{n} \Xb^{\T} \Xb \bbeta + \frac{1}{n} \Xb^{\T} \bepsilon- \bSigma \bbeta \right)^{\T} \left( \tau_{n_w}(\theta) \bSigma + \theta \Ib_p \right)^{-1}\left( \frac{1}{n} \Xb^{\T} \Xb \bbeta + \frac{1}{n} \Xb^{\T} \bepsilon- \bSigma \bbeta \right) \right] + o_p(1)
    \\
    \overset{(f)}{=}\ & \frac{n^{\train}}{n_w} \cdot \frac{\kappa \sigma_{\bbeta}^2}{p} \cdot \tr\left[ \left( \frac{1}{n} \Xb^{\T} \Xb - \bSigma \right)^{\T} \left( \tau_{n_w}(\theta) \bSigma + \theta \Ib_p \right)^{-1}\left( \frac{1}{n} \Xb^{\T} \Xb - \bSigma \right) \right]
    \\
    &+ \frac{n^{\train}}{n_w} \cdot \sigma_{\bepsilon}^2 \cdot \frac{1}{n^2} \tr\left[ \Xb \left( \tau_{n_w}(\theta) \bSigma + \theta \Ib_p \right)^{-1} \Xb^{\T}  \right] + o_p(1). 
\end{align*}
We further simplify $\Omega_{\rm sum}^{(4)}$ as follows
\begin{align*}
    \Omega_{\rm sum}^{(4)} =\ & \frac{n^{\train}}{n_w} \cdot \frac{\kappa \sigma_{\bbeta}^2}{p} \cdot \tr\left[ \frac{1}{n} \Xb^{\T} \Xb \left( \tau_{n_w}(\theta) \bSigma + \theta \Ib_p \right)^{-1} \frac{1}{n} \Xb^{\T} \Xb \right]
    \\
    &- 2 \cdot \frac{n^{\train}}{n_w} \cdot \frac{\kappa \sigma_{\bbeta}^2}{p} \cdot \tr\left[ \left( \tau_{n_w}(\theta) \bSigma + \theta \Ib_p \right)^{-1} \frac{1}{n} \Xb^{\T} \Xb \bSigma \right]
    \\
    &+ \frac{n^{\train}}{n_w} \cdot \frac{\kappa \sigma_{\bbeta}^2}{p} \cdot \tr \left[ \left( \tau_{n_w}(\theta) \bSigma + \theta \Ib_p \right)^{-1}\bSigma^2 \right]
    \\
    &+ \frac{n^{\train}}{n_w} \cdot \frac{\sigma_{\bepsilon}^2}{n} \cdot \tr\left[ \left( \tau_{n_w}(\theta) \bSigma + \theta \Ib_p \right)^{-1} \hat{\bSigma}_{n}  \right] + o_p(1)
    \\
    =\ & \frac{n^{\train}}{n_w} \cdot \frac{\kappa \sigma_{\bbeta}^2}{p} \cdot \tr\left[ \left( \tau_{n_w}(\theta) \bSigma + \theta \Ib_p \right)^{-1} \hat{\bSigma}_{n}^2 \right] - 2 \cdot \frac{n^{\train}}{n_w} \cdot \frac{\kappa \sigma_{\bbeta}^2}{p} \cdot \tr\left[ \left( \tau_{n_w}(\theta) \bSigma + \theta \Ib_p \right)^{-1} \hat{\bSigma}_{n} \bSigma \right]
    \\
    &+ \frac{n^{\train}}{n_w} \cdot \frac{\kappa \sigma_{\bbeta}^2}{p} \cdot \tr \left[ \left( \tau_{n_w}(\theta) \bSigma + \theta \Ib_p \right)^{-1} \bSigma^2 \right] + \frac{n^{\train}}{n_w} \cdot \sigma_{\bepsilon}^2 \cdot \frac{1}{n} \tr\left[ \left( \tau_{n_w}(\theta) \bSigma + \theta \Ib_p \right)^{-1} \hat{\bSigma}_{n} \right] + o_p(1).
\end{align*}
It follows that 
\begin{align*}
    \Omega_{\rm sum}^{(1)} 
    =\ & \Omega_{\rm sum}^{(2)} + \Omega_{\rm sum}^{(3)} - \Omega_{\rm sum}^{(4)} + o_p(1)
    \\
    =\ & \frac{n^{\train}}{n_w} \cdot \frac{\kappa \sigma_{\bbeta}^2}{p} \cdot \tr \left[ \left( \tau_{n_w}(\theta) \bSigma + \theta \Ib_p \right)^{-1} \hat{\bSigma}_{n}^2 \right] + \frac{n^{\train}}{n_w} \cdot \sigma_{\bepsilon}^2 \cdot \frac{1}{n} \tr \left[ \left( \tau_{n_w}(\theta) \bSigma + \theta \Ib_p \right)^{-1} \hat{\bSigma}_{n} \right]
    \\
    &- \frac{n^{\train}}{n_w} \cdot \frac{\kappa \sigma_{\bbeta}^2}{p} \cdot \tr\left[ \left( \tau_{n_w}(\theta) \bSigma + \theta \Ib_p \right)^{-1} \hat{\bSigma}_{n}^2 \right] + 2 \cdot \frac{n^{\train}}{n_w} \cdot \frac{\kappa \sigma_{\bbeta}^2}{p} \cdot \tr\left[ \left( \tau_{n_w}(\theta) \bSigma + \theta \Ib_p \right)^{-1} \hat{\bSigma}_{n} \bSigma \right]
    \\
    &- \frac{n^{\train}}{n_w} \cdot \frac{\kappa \sigma_{\bbeta}^2}{p} \cdot \tr \left[ \left( \tau_{n_w}(\theta) \bSigma + \theta \Ib_p \right)^{-1} \bSigma^2 \right] - \frac{n^{\train}}{n_w} \cdot \sigma_{\bepsilon}^2 \cdot \frac{1}{n} \tr\left[ \left( \tau_{n_w}(\theta) \bSigma + \theta \Ib_p \right)^{-1} \hat{\bSigma}_{n} \right] + o_p(1)
    \\
    =\ & 2 \cdot \frac{n^{\train}}{n_w} \cdot \frac{\kappa \sigma_{\bbeta}^2}{p} \cdot \tr\left[ \left( \tau_{n_w}(\theta) \bSigma + \theta \Ib_p \right)^{-1} \hat{\bSigma}_{n} \bSigma \right] - \frac{n^{\train}}{n_w} \cdot \frac{\kappa \sigma_{\bbeta}^2}{p} \cdot \tr \left[ \left( \tau_{n_w}(\theta) \bSigma + \theta \Ib_p \right)^{-1} \bSigma^2 \right] + o_p(1)
    \\
    =\ & \frac{n^{\train}}{n_w} \cdot \frac{\kappa \sigma_{\bbeta}^2}{p} \cdot \tr\left[ \left( \tau_{n_w}(\theta) \bSigma + \theta \Ib_p \right)^{-1} \bSigma^2 \right] + o_p(1),
\end{align*}
where the last equation follows the fact that $\hat{\bSigma}_{n} \asymp \bSigma$. 
This completes the computation of $\Omega_{\rm sum}^{(1)}$, we now proceed to compute $\Omega_{\rm sum}^{(2)}$. Note that  
\begin{align*}
    \Omega_{\rm sum}^{(2)}
    =\ &  \left[ \frac{n^{\train}}{n} \Xb^{\T}\yb + \sqrt{\frac{n^{\train} (n - n^{\train})}{n^2}} (\Cov(\Xb^{\T} \yb))^{1/2} \hb \right]^{\T} 
    \\
    & \qquad \qquad \left( \Wb^{\T} \Wb + \theta n_w \Ib_{p} \right)^{-1} \bSigma \left( \Wb^{\T} \Wb + \theta n_w \Ib_{p} \right)^{-1} \left[ \frac{n^{\train}}{n} \Xb^{\T}\yb + \sqrt{\frac{n^{\train} (n - n^{\train})}{n^2}} (\Cov(\Xb^{\T} \yb))^{1/2} \hb \right]
    \\
    =\ &  \frac{1}{n_w^2} \cdot \left[ \frac{n^{\train}}{n} \Xb^{\T}\yb + \sqrt{\frac{n^{\train} (n - n^{\train})}{n^2}} (\Cov(\Xb^{\T} \yb))^{1/2} \hb \right]^{\T} 
    \\
    & \qquad \qquad \left( \Wb^{\T} \Wb + \theta n_w \Ib_{p} \right)^{-1} \bSigma \left( \Wb^{\T} \Wb + \theta n_w \Ib_{p} \right)^{-1} \left[ \frac{n^{\train}}{n} \Xb^{\T}\yb + \sqrt{\frac{n^{\train} (n - n^{\train})}{n^2}} (\Cov(\Xb^{\T} \yb))^{1/2} \hb \right].
\end{align*}
By Lemma \ref{lemma:DE_second_order}, we have
\begin{align*}
    \Omega_{\rm sum}^{(2)} =\ & \frac{1}{n_w^2} \cdot \left[ \frac{n^{\train}}{n} \Xb^{\T}\yb + \sqrt{\frac{n^{\train} (n - n^{\train})}{n^2}} (\Cov(\Xb^{\T} \yb))^{1/2} \hb \right]^{\T} 
    \\
    & \qquad \qquad \Bb(\theta) \left[ \frac{n^{\train}}{n} \Xb^{\T}\yb + \sqrt{\frac{n^{\train} (n - n^{\train})}{n^2}} (\Cov(\Xb^{\T} \yb))^{1/2} \hb \right] + o_p(1),
\end{align*}
where 
\begin{align} \label{eqn:B}
    \Bb(\theta) = (\rho_{n_w}(\theta) + 1) \cdot \left( \tau_{n_w}(\theta)\bSigma  + \theta \Ib_{p}\right)^{-1} \bSigma \left(\tau_{n_w}(\theta)\bSigma  + \theta \Ib_{p}\right)^{-1}. 
\end{align}
Plugging in the Equation \eqref{eqn:covariance}, we have
\begin{align*}
    \Omega_{\rm sum}^{(2)} =\ & \frac{1}{n_w^2} \cdot \left[ \frac{n^{\train}}{n} \Xb^{\T}\yb + \sqrt{\frac{n^{\train} (n - n^{\train})}{n^2}} (\Cov(\Xb^{\T} \yb))^{1/2} \hb \right]^{\T} 
    \\
    & \qquad \qquad \Bb(\theta) \left[ \frac{n^{\train}}{n} \Xb^{\T}\yb + \sqrt{\frac{n^{\train} (n - n^{\train})}{n^2}} (\Cov(\Xb^{\T} \yb))^{1/2} \hb \right] + o_p(1)
    \\
    =\ & \frac{1}{n_w^2} \cdot \left[ \frac{n^{\train}}{n} \Xb^{\T}\Xb \bbeta + \frac{n^{\train}}{n} \Xb^{\T} \bepsilon + \sqrt{\frac{n^{\train} (n - n^{\train})}{n^2}} (\Cov(\Xb^{\T} \yb))^{1/2} \hb \right]^{\T} 
    \\
    & \qquad \qquad \Bb(\theta) \left[ \frac{n^{\train}}{n} \Xb^{\T}\Xb \bbeta + \frac{n^{\train}}{n} \Xb^{\T} \bepsilon + \sqrt{\frac{n^{\train} (n - n^{\train})}{n^2}} (\Cov(\Xb^{\T} \yb))^{1/2} \hb \right] + o_p(1)
    \\
    =\ & \frac{1}{n_w^2} \cdot \left( \frac{n^{\train}}{n} \right)^2 \bbeta^{\T} \Xb^{\T} \Xb \Bb(\theta) \Xb^{\T} \Xb \bbeta 
    + \frac{1}{n_w^2} \cdot \left( \frac{n^{\train}}{n} \right)^2 \bepsilon^{\T} \Xb \Bb(\theta) \Xb^{\T} \bepsilon
    \\
    \qquad &+ \frac{n^{\train} (n - n^{\train})}{n^2 \cdot n_w^2} \hb^{\T} (\Cov(\Xb^{\T} \yb))^{1/2} \Bb(\theta) (\Cov(\Xb^{\T} \yb))^{1/2} \hb + o_p(1).
\end{align*}
Lemma B.26 in \cite{bai2010spectral} and Lemma 5 in \cite{zhao2024estimating} imply that
\begin{align*}
    \Omega_{\rm sum}^{(2)} =\ & \frac{1}{n_w^2} \cdot \left( \frac{n^{\train}}{n} \right)^2 \frac{\kappa \sigma_{\bbeta}^2}{p} \cdot  \tr \left( \Xb^{\T} \Xb \Bb(\theta) \Xb^{\T} \Xb \right)
    + \frac{1}{n_w^2} \cdot \left( \frac{n^{\train}}{n} \right)^2 \sigma_{\bepsilon}^2 \cdot \tr \left( \Xb \Bb(\theta) \Xb^{\T} \right)
    \\
    &+ \frac{n^{\train} (n - n^{\train})}{n^2 \cdot n_w^2} \tr \left[ (\Cov(\Xb^{\T} \yb))^{1/2} \Bb(\theta) (\Cov(\Xb^{\T} \yb))^{1/2} \right] + o_p(1)
    \\
    =\ & \left( \frac{n^{\train}}{n_w} \right)^2 \frac{\kappa \sigma_{\bbeta}^2}{p} \cdot  \tr \left( \Bb(\theta) \hat{\bSigma}_{n}^2 \right)
    + \frac{1}{n} \cdot \left( \frac{n^{\train}}{n_w} \right)^2 \sigma_{\bepsilon}^2 \cdot \tr \left( \Bb(\theta) \hat{\bSigma}_{n} \right)
    \\
    &+ \frac{n^{\train} (n - n^{\train})}{n_w^2} \tr \left[ \Bb(\theta) \left( \frac{1}{n} \Xb^{\T} \yb - \bSigma \bbeta \right) \left( \frac{1}{n} \Xb^{\T} \yb - \bSigma \bbeta \right)^{\T} \right] + o_p(1).
\end{align*}
Similarly, we have
\begin{align*}
    &\tr \left[ \Bb(\theta) \left( \frac{1}{n} \Xb^{\T} \yb - \bSigma \bbeta \right) \left( \frac{1}{n} \Xb^{\T} \yb - \bSigma \bbeta \right)^{\T} \right]
    \\
    =\ & \frac{\kappa \sigma_{\bbeta}^2}{p} \cdot \tr\left( \Bb(\theta) \hat{\bSigma}_{n}^2 \right) - 2 \cdot \frac{\kappa \sigma_{\bbeta}^2}{p} \cdot \tr\left( \Bb(\theta) \hat{\bSigma}_{n} \bSigma \right)
    + \frac{\kappa \sigma_{\bbeta}^2}{p} \cdot \tr \left( \Bb(\theta) \bSigma^2 \right) + \sigma_{\bepsilon}^2 \cdot \frac{1}{n} \tr\left( \Bb(\theta) \hat{\bSigma}_{n} \right).
\end{align*}
It follows that 
\begin{align*}
    \Omega_{\rm sum}^{(2)} =\ & \left( \frac{n^{\train}}{n_w} \right)^2 \frac{\kappa \sigma_{\bbeta}^2}{p} \cdot  \tr \left( \Bb(\theta) \hat{\bSigma}_{n}^2 \right)
    + \frac{1}{n} \cdot \left( \frac{n^{\train}}{n_w} \right)^2 \sigma_{\bepsilon}^2 \cdot \tr \left( \Bb(\theta) \hat{\bSigma}_{n} \right)
    \\
    &+ \frac{n^{\train} (n - n^{\train})}{n_w^2} \left[ \frac{\kappa \sigma_{\bbeta}^2}{p} \cdot \tr\left( \Bb(\theta) \hat{\bSigma}_{n}^2 \right) - 2 \cdot \frac{\kappa \sigma_{\bbeta}^2}{p} \cdot \tr\left( \Bb(\theta) \hat{\bSigma}_{n} \bSigma \right) \right.
    \\
    &+ \left. \frac{\kappa \sigma_{\bbeta}^2}{p} \cdot \tr \left( \Bb(\theta) \bSigma^2 \right) + \sigma_{\bepsilon}^2 \cdot \frac{1}{n} \tr\left( \Bb(\theta) \hat{\bSigma}_{n} \right) \right] + o_p(1)
    \\
    =\ & \frac{n^{\train}}{n_w} \cdot \frac{n}{n_w} \cdot \frac{\kappa \sigma_{\bbeta}^2}{p} \cdot  \tr \left( \Bb(\theta) \hat{\bSigma}_{n}^2 \right) + \frac{n^{\train}}{n_w} \cdot \frac{n}{n_w} \cdot \frac{\sigma_{\bepsilon}^2}{n} \cdot \tr\left( \Bb(\theta) \hat{\bSigma}_{n} \right)
    \\
    &- \frac{n^{\train} (n - n^{\train})}{n_w^2} \cdot \frac{\kappa \sigma_{\bbeta}^2}{p} \cdot \tr\left( \Bb(\theta) \bSigma^2 \right) + o_p(1)
    \\
    \overset{(g)}{=}\ & \frac{n^{\train}}{n_w} \cdot \frac{n}{n_w} \cdot \frac{\kappa \sigma_{\bbeta}^2}{p} \cdot  \left[\frac{p}{n} \cdot \frac{1}{p} \tr \left(\bSigma \right) \cdot \tr \left(\Bb(\theta) \bSigma \right) + \tr \left(\Bb(\theta) \bSigma^2 \right) \right] + \frac{n^{\train}}{n_w^2} \cdot \sigma_{\bepsilon}^2 \cdot \tr\left(\Bb(\theta) \hat{\bSigma}_{n} \right)
    \\
    &- \frac{n^{\train} (n - n^{\train})}{n_w^2} \cdot \frac{\kappa \sigma_{\bbeta}^2}{p} \cdot \tr\left( \Bb(\theta) \bSigma^2 \right) + o_p(1)
    \\
    =\ & \left( \frac{n^{\train}}{n_w} \right)^2 \cdot \kappa \sigma_{\bbeta}^2 \cdot \frac{p}{n^{\train}} \cdot \frac{1}{p} \tr \left(\bSigma \right) \cdot \frac{1}{p}\tr \left(\Bb(\theta) \bSigma \right) + \left( \frac{n^{\train}}{n_w} \right)^2 \cdot \frac{\kappa \sigma_{\bbeta}^2}{p}\tr \left(\Bb(\theta) \bSigma^2 \right)
    \\
    &+ \frac{n^{\train}}{n_w^2} \cdot \sigma_{\bepsilon}^2 \cdot \tr\left(\Bb(\theta) {\bSigma} \right) + o_p(1),
\end{align*}
where Equation (g) follows from Lemma \ref{lemma:second_moment}.
By the continuous mapping theorem, we conclude that 
\begin{align*}
    R^2_{\rm sum, R}(\theta) 
    = \frac{n^{\valid}}{\|\yb^{\valid}\|_{2}^2} \cdot \frac{n^{\train}}{p} \cdot \kappa \sigma_{\bbeta}^2 \cdot \frac{\left( \tr \left[ \left( \tau_{n_w}(\theta) \bSigma + \theta \Ib_p \right)^{-1} \bSigma^2 \right]\right)^2}{\tr \left(\bSigma \right) \cdot \tr \left(\Bb(\theta) \bSigma \right)/h^2 + n^{\train} \cdot \tr \left(\Bb(\theta) \bSigma^2 \right)} + o_p(1),
\end{align*}
with $\Bb(\theta)$ and $h^2$ being defined in Equations \eqref{eqn:B} and \eqref{eqn:h^2}, respectively. 
% \begin{align*}
%     \tr (\Bb \hat{\bSigma}^{2}) \asymp \frac{d^2}{d w^2} \bigg|_{w=0} \tr \left[ \Bb \left( \tau_{n_w}(-w^{-1})w \bSigma + \Ib_p \right)^{-1} \right].
% \end{align*}

\paragraph{\texorpdfstring{Part II: The limit of $\mathbf{R^2_{\rm ind, R}(\theta)}$.}{Individual}} 
Recall that 
\begin{align*}
    R^2_{\rm ind, R}(\theta) = \frac{n^{\valid}}{\|\yb^{\valid}\|_{2}^2} \cdot \frac{\left\langle {\Xb^{\valid}{}^{\T}\yb^{\valid}}, (\Wb^{\T} \Wb + \theta n_w \Ib_{p})^{-1} \Xb^{\train}{}^{\T} \yb^{\train} \right\rangle^2}{n^{\valid}{}^{2} \cdot \| (\Wb^{\T} \Wb + \theta n_w \Ib_{p})^{-1} \Xb^{\train}{}^{\T} \yb^{\train} \|_{\bSigma}^2} = \frac{n^{\valid}}{\|\yb^{\valid}\|_{2}^2} \cdot \frac{\left( \Omega_{\rm ind}^{(1)} \right)^{2}}{\Omega_{\rm ind}^{(2)}},
\end{align*}
where $\Omega_{\rm ind}^{(1)}$ and $\Omega_{\rm ind}^{(2)}$ are defined as
\begin{align*}
    \Omega_{\rm ind}^{(1)} =\ & \frac{1}{n^{\valid}} \cdot \left\langle {\Xb^{\valid}{}^{\T}\yb^{\valid}}, (\Wb^{\T} \Wb + \theta n_w \Ib_{p})^{-1} \Xb^{\train}{}^{\T} \yb^{\train} \right\rangle \quad \mbox{and}
    \\
    \Omega_{\rm ind}^{(2)} =\ & \left\| (\Wb^{\T} \Wb + \theta n_w \Ib_{p})^{-1} \Xb^{\train}{}^{\T} \yb^{\train} \right\|_{\bSigma}^2, 
\end{align*}
respectively. For $\Omega_{\rm ind}^{(1)}$, we have 
\begin{align*}
    \Omega_{\rm ind}^{(1)}
    =\ & \frac{1}{n^{\valid}} \left(\Xb^{\train}{}^{\T} \yb^{\train} \right)^{\T} \left( \Wb^{\T} \Wb + \theta n_w \Ib_{p} \right)^{-1}  \Xb^{\valid}{}^{\T} \yb^{\valid}
    \\
    =\ & \frac{n^{\train}}{n_w} \cdot \left(\frac{\Xb^{\train}{}^{\T} \yb^{\train}}{n^{\train}} \right)^{\T} \left( \Wb^{\T} \Wb/n_w + \theta \Ib_{p} \right)^{-1}  \frac{\Xb^{\valid}{}^{\T} \yb^{\valid}}{n^{\valid}}.
\end{align*}
Lemma \ref{lemma:DE_second_order} indicates that 
\begin{align*}
    \Omega_{\rm ind}^{(1)} =\ & \frac{n^{\train}}{n_w} \cdot \left(\frac{\Xb^{\train}{}^{\T} \yb^{\train}}{n^{\train}} \right)^{\T} \left( \tau_{n_w}(\theta) \bSigma + \theta \Ib_p \right)^{-1}  \frac{\Xb^{\valid}{}^{\T} \yb^{\valid}}{n^{\valid}} + o_p(1).
\end{align*}
Define independent $p-$dimensional vectors $\bepsilon^{\train}$ and $\bepsilon^{\valid}$ satisfying the following equations and plugging into the equation above
\begin{align*}
    \yb^{\train} = \Xb^{\train} \bbeta + \bepsilon^{\train} \quad \mbox{and} \quad \yb^{\valid} = \Xb^{\valid} \bbeta + \bepsilon^{\valid}. 
\end{align*}
Lemma B.26 in \cite{bai2010spectral} and Lemma 5 in \cite{zhao2024estimating} indicate that 
\begin{align*}
    \Omega_{\rm ind}^{(1)} =\ & \frac{n^{\train}}{n_w} \cdot \left(\frac{\Xb^{\train}{}^{\T} \Xb^{\train} \bbeta}{n^{\train}} \right)^{\T} \left( \tau_{n_w}(\theta) \bSigma + \theta \Ib_p \right)^{-1}  \frac{\Xb^{\valid}{}^{\T} \Xb^{\valid} \bbeta}{n^{\valid}} + o_p(1)
    \\
    =\ & \frac{n^{\train}}{n_w} \cdot \frac{\kappa \sigma_{\bbeta}^2}{p} \cdot \tr \left[ \frac{\Xb^{\train}{}^{\T} \Xb^{\train}}{n^{\train}}  \left( \tau_{n_w}(\theta) \bSigma + \theta \Ib_p \right)^{-1}  \frac{\Xb^{\valid}{}^{\T} \Xb^{\valid}}{n^{\valid}} \right] + o_p(1).
\end{align*}
Based on the fact that $\hat{\bSigma}_{n^{\train}} = {\Xb^{\train}{}^{\T} \Xb^{\train}}/{n^{\train}} \asymp \bSigma$ and $\hat{\bSigma}_{n^{\valid}} = {\Xb^{\valid}{}^{\T} \Xb^{\valid}}/{n^{\valid}} \asymp \bSigma$, we have
\begin{align*}
    \Omega_{\rm ind}^{(1)} = \frac{n^{\train}}{n_w} \cdot \frac{\kappa \sigma_{\bbeta}^2}{p} \cdot \tr \left[ \left( \tau_{n_w}(\theta) \bSigma + \theta \Ib_p \right)^{-1} \bSigma^2 \right] + o_p(1). 
\end{align*}
This demonstrates the limit of $\Omega_{\rm ind}^{(1)}$. For $\Omega_{\rm ind}^{(2)}$, we have 
\begin{align*}
    \Omega_{\rm ind}^{(2)} 
    % =\ & \left(\Xb^{\train}{}^{\T} \yb^{\train} \right)^{\T} \left( \Wb^{\T} \Wb + \theta n_w \Ib_{p} \right)^{-1} \left( \Xb^{\valid} \right)^{\T} \Xb^{\valid} \left( \Wb^{\T} \Wb + \theta n_w \Ib_{p} \right)^{-1}  \Xb^{\train}{}^{\T} \yb^{\train} 
    % \\
    =\ & \left(\Xb^{\train}{}^{\T} \yb^{\train} \right)^{\T} \left( \Wb^{\T} \Wb + \theta n_w \Ib_{p} \right)^{-1} \bSigma \left( \Wb^{\T} \Wb + \theta n_w \Ib_{p} \right)^{-1}  \Xb^{\train}{}^{\T} \yb^{\train}
    \\
    =\ & \left( \frac{\Xb^{\train}{}^{\T} \yb^{\train}}{n_w} \right)^{\T} \left( \Wb^{\T} \Wb/n_w + \theta \Ib_{p} \right)^{-1} \bSigma \left( \Wb^{\T} \Wb/n_w + \theta \Ib_{p} \right)^{-1} \frac{\Xb^{\train}{}^{\T} \yb^{\train}}{n_w}.
\end{align*}
Lemma \ref{lemma:DE_second_order} indicates that
\begin{align*}
    \Omega_{\rm ind}^{(2)} =\ & \left( \frac{\Xb^{\train}{}^{\T} \yb^{\train}}{n_w} \right)^{\T} \Bb(\theta) \frac{\Xb^{\train}{}^{\T} \yb^{\train}}{n_w} + o_p(1)
    \\
    =\ & \left( \frac{n^{\train}}{n_w} \right)^2 \cdot \left( \frac{\Xb^{\train}{}^{\T} \yb^{\train}}{n^{\train}} \right)^{\T} \Bb(\theta) \frac{\Xb^{\train}{}^{\T} \yb^{\train}}{n^{\train}} + o_p(1),
\end{align*}
with $\Bb(\theta)$ being defined in \eqref{eqn:B}. 
Moreover, plugging in $\bepsilon^{\train}$ and $\bepsilon^{\valid}$ defined above, we have
\begin{align*}
    \Omega_{\rm ind}^{(2)} =\ & \left( \frac{n^{\train}}{n_w} \right)^2 \cdot \left( \frac{\Xb^{\train}{}^{\T} \yb^{\train}}{n^{\train}} \right)^{\T} \Bb(\theta) \frac{\Xb^{\train}{}^{\T} \yb^{\train}}{n^{\train}} + o_p(1)
    \\
    =\ & \left( \frac{n^{\train}}{n_w} \right)^2 \cdot \left( \frac{\Xb^{\train}{}^{\T} \Xb^{\train} \bbeta + \Xb^{\train} \bepsilon^{\train}}{n^{\train}} \right)^{\T} \Bb(\theta) \frac{\Xb^{\train}{}^{\T} \Xb^{\train} \bbeta + \Xb^{\train})^{\T} \bepsilon^{\train}}{n^{\train}} + o_p(1)
    \\
    \overset{(h)}{=}\ & \left( \frac{n^{\train}}{n_w} \right)^2 \cdot \left( \bbeta^{\T} \frac{\Xb^{\train}{}^{\T} \Xb^{\train} \Bb(\theta) \Xb^{\train}{}^{\T} \Xb^{\train}}{n^{\train}{}^2} \bbeta + \frac{1}{n^{\train}{}^2} \cdot \bepsilon^{\train}{}^{\T} \Xb^{\train} \Bb(\theta) \Xb^{\train}{}^{\T} \bepsilon^{\train} \right) + o_p(1)
    \\
    \overset{(i)}{=}\ & \left( \frac{n^{\train}}{n_w} \right)^2 \cdot \frac{\kappa \sigma_{\bbeta}^2}{p} \cdot \tr \left(\frac{\Xb^{\train}{}^{\T} \Xb^{\train} \Bb(\theta) \Xb^{\train}{}^{\T} \Xb^{\train}}{n^{\train}{}^2} \right) + \left( \frac{n^{\train}}{n_w} \right)^2 \cdot \frac{\sigma_{\bepsilon}^2}{n^{\train}} \cdot \tr\left( \Bb(\theta) \frac{\Xb^{\train}{}^{\T} \Xb^{\train}}{n^{\train}} \right) + o_p(1),
\end{align*}
where Equations (h) and (i) follow from Lemma B.26 in \cite{bai2010spectral} and Lemma 5 in \cite{zhao2024estimating}. 
Note that $\Omega_{\rm ind}^{(2)}$ can be equivalently written as
\begin{align*}
    \Omega_{\rm ind}^{(2)} =\ & \left( \frac{n^{\train}}{n_w} \right)^2 \cdot \frac{\kappa \sigma_{\bbeta}^2}{p} \cdot \tr \left( \Bb(\theta) \hat{\bSigma}_{n^{\train}}^2 \right) + \left( \frac{n^{\train}}{n_w} \right)^2 \cdot \frac{\sigma_{\bepsilon}^2}{n^{\train}} \cdot \tr\left( \Bb(\theta) \hat{\bSigma}_{n^{\train}} \right) + o_p(1)
    \\
    =\ & \left( \frac{n^{\train}}{n_w} \right)^2 \cdot \frac{\kappa \sigma_{\bbeta}^2}{p} \cdot \tr \left( \Bb(\theta) \hat{\bSigma}_{n^{\train}}^2 \right) +  \frac{n^{\train}}{n_w^2} \cdot \sigma_{\bepsilon}^2 \cdot \tr\left( \Bb(\theta) \hat{\bSigma}_{n^{\train}} \right) + o_p(1).
\end{align*}
Lemma \ref{lemma:second_moment} implies that
\begin{align*}
    \Omega_{\rm ind}^{(2)} =\ & \left( \frac{n^{\train}}{n_w} \right)^2 \cdot \frac{\kappa \sigma_{\bbeta}^2}{p} \cdot \left[\frac{p}{n^{\train}} \cdot \frac{1}{p} \tr \left(\bSigma \right) \cdot \tr \left(\Bb(\theta) \bSigma \right) + \tr \left(\Bb(\theta) \bSigma^2 \right) \right] 
    \\
    & +  \frac{n^{\train}}{n_w^2} \cdot \sigma_{\bepsilon}^2 \cdot \tr\left( \Bb(\theta) \hat{\bSigma}_{n^{\train}} \right) + o_p(1)
    \\
    =\ & \left( \frac{n^{\train}}{n_w} \right)^2 \cdot \kappa \sigma_{\bbeta}^2 \cdot \frac{p}{n^{\train}} \cdot \frac{1}{p} \tr \left(\bSigma \right) \cdot \frac{1}{p}\tr \left(\Bb(\theta) \bSigma \right) + \left( \frac{n^{\train}}{n_w} \right)^2 \cdot \frac{\kappa \sigma_{\bbeta}^2}{p}\tr \left(\Bb(\theta) \bSigma^2 \right)
    \\
    &+ \frac{n^{\train}}{n_w^2} \cdot \sigma_{\bepsilon}^2 \cdot \tr\left(\Bb(\theta) \bSigma \right) + o_p(1). 
\end{align*}
By the continuous mapping theorem, we have $R^2_{\rm ind, R}(\theta)$ as follows.
\begin{align*}
    R^2_{\rm ind, R}(\theta) =\ & \frac{n^{\valid}}{\|\yb^{\valid}\|_{2}^2} \cdot \frac{\left( {\kappa \sigma_{\bbeta}^2}/{p} \cdot \tr \left[ \left( \tau_{n_w}(\theta) \bSigma + \theta \Ib_p \right)^{-1} \bSigma^2 \right]\right)^2}{\kappa \sigma_{\bbeta}^2/p \cdot {p}/{n^{\train}} \cdot \tr \left(\bSigma \right) \cdot {1}/{p} \cdot \tr \left(\Bb(\theta) \bSigma \right) + {\kappa \sigma_{\bbeta}^2}/{p} \tr \left(\Bb(\theta) \bSigma^2 \right) + \sigma_{\bepsilon}^2/n^{\train} \cdot \tr\left( \Bb(\theta) \bSigma \right)} \\
    &+ o_p(1) 
    \\
    =\ & \frac{n^{\valid}}{\|\yb^{\valid}\|_{2}^2} \cdot \frac{n^{\train}}{p} \cdot \kappa \sigma_{\bbeta}^2 \cdot \frac{\left( \tr \left[ \left( \tau_{n_w}(\theta) \bSigma + \theta \Ib_p \right)^{-1} \bSigma^2 \right]\right)^2}{\tr \left(\bSigma \right) \cdot \tr \left(\Bb(\theta) \bSigma \right)/h^2 + n^{\train} \cdot \tr \left(\Bb(\theta) \bSigma^2 \right)} + o_p(1).
\end{align*}
Therefore, the resulting $R^2_{\rm sum, R}(\theta)$ is asymptotically equivalent to $R^2_{\rm ind, R}(\theta)$ for all $\theta \in \RR_{+}$.

\section{Proof of Theorem \ref{thm:marginal_screen}}
\label{sec:proof_thm_marginal_screen}

In this section, we outline the proof of Theorem \ref{thm:marginal_screen}. As the proof closely parallels that of Theorem \ref{thm:reference_panel}, we omit most of the details and focus only on the key differences.

\paragraph{\texorpdfstring{Part I: The limit of $\mathbf{R^2_{\rm sum, M}(\theta)}$.}{Summery}}

Similar to previous section, we first derive the $R^2_{\rm sum, M}(\Theta)$
\begin{align*}
    R^2_{\rm sum, M}(\Theta) = \frac{n^{\valid}}{\|\yb^{\valid}\|_{2}^2} \cdot \frac{\left \langle {\sbb^{\valid}}, \Ab(\Theta) \sbb^{\train} \right \rangle^{2} }{n^{\valid}{}^{2} \cdot \| \Ab(\Theta) \sbb^{\train} \|_{\bSigma}^{2}} 
\end{align*}
with $\Ab(\Theta)$ being defined in Equation \eqref{eqn:special_linear_MG}. 
Following similar steps in the previous section by replacing both $\left( \Wb^{\T}\Wb + \theta \Ib_{p} \right)^{-1}$ and $\left( \tau_{n_w}(\theta) \bSigma + \theta \Ib_p \right)^{-1}$ to the deterministic matrix $\Ab(\Theta)$, we have
\begin{align*}
    R^2_{\rm sum, M}(\Theta)
    =\ & \frac{n^{\valid}}{\|\yb^{\valid}\|_{2}^2} \cdot \frac{n^{\train}}{p} \cdot \kappa \sigma_{\bbeta}^2 \cdot \frac{\left[ \tr \left( \Ab(\Theta) \bSigma^2 \right)\right]^2}{\tr \left(\bSigma \right) \cdot \tr \left(\Ab(\Theta) \bSigma \right)/h^2 + n^{\train} \cdot \tr \left(\Ab(\Theta) \bSigma^2 \right)} + o_p(1).
\end{align*}

\paragraph{\texorpdfstring{Part II: The limit of $\mathbf{R^2_{\rm ind, M}(\theta)}$.}{Individual}}
We compute 
\begin{align*}
    R^2_{\rm ind, M}(\theta) = \frac{n^{\valid}}{\|\yb^{\valid}\|_{2}^2} \cdot \frac{\left\langle {\Xb^{\valid}{}^{\T}\yb^{\valid}}, \Ab(\Theta) \Xb^{\train}{}^{\T} \yb^{\train} \right\rangle^2}{n^{\valid}{}^{2} \cdot \| \Ab(\Theta) \Xb^{\train}{}^{\T} \yb^{\train} \|_{\bSigma}^2}.
\end{align*}
Following similar steps in the previous section as above, we have 
\begin{align*}
    R^2_{\rm ind, M}(\Theta)
    =\ & \frac{n^{\valid}}{\|\yb^{\valid}\|_{2}^2} \cdot \frac{n^{\train}}{p} \cdot \kappa \sigma_{\bbeta}^2 \cdot \frac{\left[ \tr \left( \Ab(\Theta) \bSigma^2 \right)\right]^2}{\tr \left(\bSigma \right) \cdot \tr \left(\Ab(\Theta) \bSigma \right)/h^2 + n^{\train} \cdot \tr \left(\Ab(\Theta) \bSigma^2 \right)} + o_p(1).
\end{align*}

\section{Proof of Theorems \ref{thm:general_linear} and \ref{thm:ensemble}}
\label{sec:proof_thm_general_linear}

In this section, we present the proof of Theorem \ref{thm:general_linear}, emphasizing the differences from the proof of Theorem \ref{thm:reference_panel}. 
We also clarify the reasoning behind Equation \eqref{eqn:cond-A}. 
The proof of Theorem \ref{thm:ensemble} then follows the same structure, with $\Ab\left( \Wb^{\T}\Wb, \Theta \right)$ in Theorem \ref{thm:general_linear}  replaced by $\sum_{j=1}^{k} \omega_{j} \Ab_{j} \left( \Wb^{\T}\Wb, \Theta_{j} \right)$.

\paragraph{\texorpdfstring{Part I: The limit of $\mathbf{R^2_{\rm sum, G}(\theta)}$.}{Summery}}
Recall that 
\begin{align*}
    R^2_{\rm sum, G}(\theta) = \frac{n^{\valid}}{\|\yb^{\valid}\|_{2}^2} \cdot \frac{\left \langle {\sbb^{\valid}}, \Ab\left( \Wb^{\T}\Wb, \Theta \right) \sbb^{\train} \right \rangle^{2} }{n^{\valid}{}^{2} \cdot \| \Ab\left( \Wb^{\T}\Wb, \Theta \right) \sbb^{\train} \|_{\bSigma}^{2}} = \frac{n^{\valid}}{\|\yb^{\valid}\|_{2}^2} \cdot \frac{\left( \Psi_{\rm sum}^{(1)} \right)^{2}}{\Psi_{\rm sum}^{(2)}}, 
\end{align*}
where $\Psi_{\rm sum}^{(1)}$ and $\Psi_{\rm sum}^{(2)}$ are defined as
\begin{align*}
    \Psi_{\rm sum}^{(1)} 
    = \frac{1}{n^{\valid}} \cdot \left \langle {\sbb^{\valid}}, \Ab\left( \Wb^{\T}\Wb, \Theta \right) \sbb^{\train} \right \rangle
    \quad \mbox{and} \quad
    \Psi_{\rm sum}^{(2)} = \left\| \Ab\left( \Wb^{\T}\Wb, \Theta \right) \sbb^{\train} \right\|_{\bSigma}^2, 
\end{align*}
respectively. We can derive the limit of $\Psi_{\rm sum}^{(1)}$ as in previous sections. By Equation \eqref{eqn:cond-A}, we have
\begin{align*}
    \Psi_{\rm sum}^{(1)}
    = \frac{1}{n^{\valid} \cdot n_w} \cdot \left \langle {\sbb^{\valid}}, \Db(\Theta) \sbb^{\train} \right \rangle + o_p(1)
    = \frac{n^{\train}}{n_w} \cdot \frac{\kappa \sigma_{\bbeta}^2}{p} \cdot \tr\left( \Db(\Theta) \bSigma^2 \right)+ o_p(1).
\end{align*}
For $\Psi_{\rm sum}^{(2)}$, we have 
\begin{align*}
    \Psi_{\rm sum}^{(2)}
    =\ &  \left[ \frac{n^{\train}}{n} \Xb^{\T}\yb + \sqrt{\frac{n^{\train} (n - n^{\train})}{n^2}} (\Cov(\Xb^{\T} \yb))^{1/2} \hb \right]^{\T} 
    \\
    & \qquad \qquad \Ab\left( \Wb^{\T}\Wb, \Theta \right) \bSigma \Ab\left( \Wb^{\T}\Wb, \Theta \right) \left[ \frac{n^{\train}}{n} \Xb^{\T}\yb + \sqrt{\frac{n^{\train} (n - n^{\train})}{n^2}} (\Cov(\Xb^{\T} \yb))^{1/2} \hb \right]
    \\
    =\ &  \frac{1}{n_w^2} \cdot \left[ \frac{n^{\train}}{n} \Xb^{\T}\yb + \sqrt{\frac{n^{\train} (n - n^{\train})}{n^2}} (\Cov(\Xb^{\T} \yb))^{1/2} \hb \right]^{\T} 
    \\
    & \qquad \qquad \Eb(\Theta) \left[ \frac{n^{\train}}{n} \Xb^{\T}\yb + \sqrt{\frac{n^{\train} (n - n^{\train})}{n^2}} (\Cov(\Xb^{\T} \yb))^{1/2} \hb \right] + o_p(1),
\end{align*}
where the last equation follows from \eqref{eqn:cond-A}. 
Following similar steps in the previous sections, we have
\begin{align*}
    \Psi_{\rm sum}^{(2)} 
    =\ & \left( \frac{n^{\train}}{n_w} \right)^2 \cdot \kappa \sigma_{\bbeta}^2 \cdot \frac{p}{n^{\train}} \cdot \frac{1}{p} \tr \left(\bSigma \right) \cdot \frac{1}{p}\tr \left(\Eb(\Theta) \bSigma \right) + \left( \frac{n^{\train}}{n_w} \right)^2 \cdot \frac{1}{p}\tr \left(\Eb(\Theta) \bSigma^2 \right) 
    \\
    &+  \frac{n^{\train}}{n_w^2} \cdot \sigma_{\bepsilon}^2 \cdot \tr\left( \Eb(\Theta) {\bSigma} \right) + o_p(1).
\end{align*}
Therefore, the limit of $R^2_{\rm sum, G}(\theta)$ is given by
\begin{align*}
    R^2_{\rm sum, G}(\theta) = \frac{n^{\valid}}{\|\yb^{\valid}\|_{2}^2} \cdot \frac{n^{\train}}{p} \cdot \kappa \sigma_{\bbeta}^2 \cdot \frac{\left[ \tr \left( \Db(\Theta) \bSigma^2 \right)\right]^2}{\tr \left(\bSigma \right) \cdot \tr \left(\Eb(\Theta) \bSigma \right)/h^2 + n^{\train} \cdot \tr \left(\Eb(\Theta) \bSigma^2 \right)} + o_p(1).
\end{align*}

\paragraph{\texorpdfstring{Part II: The limit of $\mathbf{R^2_{\rm ind, G}(\theta)}$.}{Individual}}
Recall that 
\begin{align*}
    R^2_{\rm ind, G}(\theta) = \frac{n^{\valid}}{\|\yb^{\valid}\|_{2}^2} \cdot \frac{\left\langle {\Xb^{\valid}{}^{\T}\yb^{\valid}}, \Ab\left( \Wb^{\T}\Wb, \Theta \right) \Xb^{\train}{}^{\T} \yb^{\train} \right\rangle^2}{n^{\valid}{}^{2} \cdot \| \Ab\left( \Wb^{\T}\Wb, \Theta \right) \Xb^{\train}{}^{\T} \yb^{\train} \|_{\bSigma}^2} = \frac{n^{\valid}}{\|\yb^{\valid}\|_{2}^2} \cdot \frac{\left( \Psi_{\rm ind}^{(1)} \right)^{2}}{\Psi_{\rm ind}^{(2)}}
\end{align*}
where $\Psi_{\rm ind}^{(1)}$ and $\Psi_{\rm ind}^{(2)}$ are defined as
\begin{align*}
    \Psi_{\rm ind}^{(1)} =\ & \frac{1}{n^{\valid}} \cdot \left\langle {\Xb^{\valid}{}^{\T}\yb^{\valid}}, \Ab\left( \Wb^{\T}\Wb, \Theta \right) \Xb^{\train}{}^{\T} \yb^{\train} \right\rangle \quad \mbox{and}
    \\
    \Psi_{\rm ind}^{(2)} =\ & \left\| \Ab\left( \Wb^{\T}\Wb, \Theta \right) \Xb^{\train}{}^{\T} \yb^{\train} \right\|_{\bSigma}^2, 
\end{align*}
respectively. We apply Equation \eqref{eqn:cond-A} and obtain
\begin{align*}
    \Psi_{\rm ind}^{(1)} =\ & \frac{n_w}{n^{\valid}} \cdot \left\langle {\Xb^{\valid}{}^{\T}\yb^{\valid}}, \Db\left( \Theta \right) \Xb^{\train}{}^{\T} \yb^{\train} \right\rangle + o_p(1) \quad \mbox{and}
    \\
    \Psi_{\rm ind}^{(2)} =\ & \left(\Xb^{\train}{}^{\T} \yb^{\train} \right)^{\T} \Ab\left( \Wb^{\T}\Wb, \Theta \right) \bSigma \Ab\left( \Wb^{\T}\Wb, \Theta \right)  \Xb^{\train}{}^{\T} \yb^{\train} 
    \\
    =\ & \frac{1}{n_w^2} \cdot \left(\Xb^{\train}{}^{\T} \yb^{\train} \right)^{\T} \Eb(\Theta) \Xb^{\train}{}^{\T} \yb^{\train}  + o_p(1).
\end{align*}
Similar to Theorem \ref{thm:reference_panel}, we have
\begin{align*}
    R^2_{\rm ind, G}(\theta) = \frac{n^{\valid}}{\|\yb^{\valid}\|_{2}^2} \cdot \frac{n^{\train}}{p} \cdot \kappa \sigma_{\bbeta}^2 \cdot \frac{\left[ \tr \left(\Db(\Theta) \bSigma^2 \right)\right]^2}{\tr \left(\bSigma \right) \cdot \tr \left(\Eb(\Theta) \bSigma \right)/h^2 + n^{\train} \cdot \tr \left(\Eb(\Theta) \bSigma^2 \right)} + o_p(1).
\end{align*}

\section{Proof of Theorem \ref{thm:multi} and Corollary \ref{cor:optimal_weight}}
In this section, we provide a detailed proof of Theorem \ref{thm:multi}. 
The proof is organized into three parts. First, similar to the previous section, we compute $R^2_{\rm sum, MA}(\Theta)$ from Algorithm \ref{alg:multi_sum}, followed by computing $R^2_{\rm ind, MA}(\Theta)$ from Algorithm \ref{alg:multi_ind}. We then demonstrate that the resulting $R^2_{\rm sum, MA}(\Theta)$ is asymptotically equivalent to $R^2_{\rm ind, MA}(\Theta)$.

The final part addresses the case where $K = 2$, focusing on finding the optimal weights $\omega_{j}, j=1, 2$, that maximize $R^2_{\rm sum, MA}(\Theta)$, and consequently $R^2_{\rm ind, MA}(\Theta)$. 
Recall that the proof is based on the condition that, for $1 \leq j \leq K$, we have 
\begin{align*}
    n_w \cdot \Ab_{j}(\Wb_{j}^{\T}\Wb_{j}, \Theta_{j}) \asymp \Db_{j}(\bSigma_{j}, \Theta_{j}) \quad\mbox{and} \quad n_w^2 \cdot \Ab_{j}(\Wb_{j}^{\T}\Wb_{j}, \Theta_{j})^{\T} \bSigma_{1} \Ab_{j}(\Wb_{j}^{\T}\Wb_{j}, \Theta_{j}) \asymp \Eb_{j}(\bSigma_{j}, \Theta_{j}).
\end{align*}
For convenience, we will abbreviate $\Db_{j}(\bSigma_{j}, \Theta_{j})$ as $\Db_{j}$ and $\Eb_{j}(\bSigma_{j}, \Theta_{j})$ as $\Eb_{j}$.

\paragraph{\texorpdfstring{Part I: The limit of $\mathbf{R^2_{\rm sum, MA}(\theta)}$.}{Summery}}
Recall that 
\begin{align*}
    R^2_{\rm sum, MA}(\theta) =\ & \frac{n^{\valid}}{\|\yb^{\valid}\|_{2}^2} \cdot \frac{\left \langle {\sbb_{1}^{\valid}}, \sum_{j=1}^{K} \omega_j \Ab_{j}(\Wb_{j}^{\T} \Wb_{j}, \Theta_{j}) \sbb_{j}^{\train} \right \rangle^{2} }{n^{\valid}{}^{2} \cdot \| \sum_{j=1}^{K} \omega_j \Ab_{j}(\Wb_{j}^{\T} \Wb_{j}, \Theta_{j}) \sbb_{j}^{\train} \|_{\bSigma_{1}}^{2}}
    = \frac{n^{\valid}}{\|\yb^{\valid}\|_{2}^2} \cdot \frac{\left( \Lambda_{\rm sum}^{(1)} \right)^{2}}{\Lambda_{\rm sum}^{(2)}},  
\end{align*}
where $\Lambda_{\rm sum}^{(1)}$ and $\Lambda_{\rm sum}^{(2)}$ are defined as
\begin{align*}
    \Lambda_{\rm sum}^{(1)} = \frac{1}{n^{\valid}} \cdot \left \langle {\sbb_{1}^{\valid}}, \sum_{j=1}^{K} \omega_j \Ab_{j}(\Wb_{j}^{\T} \Wb_{j}, \Theta_{j}) \sbb_{j}^{\train} \right \rangle
    \quad \mbox{and} \quad
    \Lambda_{\rm sum}^{(2)} = \left\| \sum_{j=1}^{K} \omega_j \Ab_{j}(\Wb_{j}^{\T} \Wb_{j}, \Theta_{j}) \sbb_{j}^{\train} \right\|_{\bSigma_{1}}^2, 
\end{align*}
respectively. 
For $\Lambda_{\rm sum}^{(1)}$, plugging in the definition of $\sbb^{\train}$ and $\sbb^{\valid}$ and the fact that $n^{\train} + n^{\valid} = n$, we have
\begin{align*}
    \Lambda_{\rm sum}^{(1)} =\ & \frac{1}{n^{\valid}} \cdot \left \langle {\sbb_{1}^{\valid}}, \sum_{j=1}^{K} \omega_j \Ab_{j}(\Wb_{j}^{\T} \Wb_{j}, \Theta_{j}) \sbb_{j}^{\train} \right \rangle
    \\
    =\ & \frac{1}{n^{\valid}} \cdot \omega_1 \left \langle {\sbb_{1}^{\valid}}, \Ab_{1}(\Wb_{1}^{\T} \Wb_{1}, \Theta_{1}) \sbb_{1}^{\train} \right \rangle + \frac{1}{n^{\valid}} \cdot \left \langle {\sbb_{1}^{\valid}}, \sum_{j=2}^{K} \omega_j \Ab_{j}(\Wb_{j}^{\T} \Wb_{j}, \Theta_{j}) \sbb_{j}^{\train} \right \rangle
    \\
    =\ & \frac{1}{n^{\valid} \cdot n_w} \cdot \omega_1 \left \langle {\sbb_{1}^{\valid}}, \Db_{1} \sbb_{1}^{\train} \right \rangle + \frac{1}{n^{\valid} \cdot n_w} \cdot \left \langle {\sbb_{1}^{\valid}}, \sum_{j=2}^{K} \omega_j \Db_{j} \sbb_{j}^{\train} \right \rangle + o_p(1)
    \\
    =\ & \Lambda_{\rm sum}^{(3)} + \Lambda_{\rm sum}^{(4)} + o_p(1).
\end{align*}
Following similar steps in  Section \ref{sec:proof_thm_general_linear}, for the first term $\Lambda_{\rm sum}^{(3)}$, we have  
\begin{align*}
    \frac{1}{n^{\valid} \cdot n_w} \cdot \omega_1 \left \langle {\sbb_{1}^{\valid}}, \Db_{1} \sbb_{1}^{\train} \right \rangle = \omega_1 \cdot \frac{n^{\train}}{n_w} \cdot \frac{\kappa_{1} \sigma_{1,1}^2}{p} \cdot \tr\left( \Db_{1} \bSigma_{1}^2 \right).
\end{align*}
For the second term $\Lambda_{\rm sum}^{(4)}$, we have  
\begin{align*}
    \Lambda_{\rm sum}^{(4)} =\ &\frac{1}{n^{\valid} \cdot n_w} \cdot \left \langle {\sbb_{1}^{\valid}}, \sum_{j=2}^{K} \omega_j \Db_{j} \sbb_{j}^{\train} \right \rangle
    \\
    =\ & \frac{1}{n^{\valid} \cdot n_w} \cdot \sum_{j=2}^{K} \omega_j \left[ \frac{n - n^{\train}}{n} \Xb_{1}^{\T}\yb_{1} - \sqrt{\frac{n^{\train} (n - n^{\train})}{n^2}} (\Cov(\Xb_{1}^{\T} \yb_{1}))^{1/2} \hb_{1} \right]^{\T}  
    \\
    &\qquad \Db_{j} \left[ \frac{n^{\train}}{n} \Xb_{j}^{\T}\yb_{j} + \sqrt{\frac{n^{\train} (n - n^{\train})}{n^2}} (\Cov(\Xb_{j}^{\T} \yb_{j}))^{1/2} \hb_{j} \right].
\end{align*}
Note that $\hb_1$ is independent with $\hb_{j}$. Therefore, by Lemma B.26 in \cite{bai2010spectral} and Lemma 5 in \cite{zhao2024estimating}, we have 
\begin{align*}
    \Lambda_{\rm sum}^{(4)} =\ &\frac{1}{n^{\valid} \cdot n_w} \cdot \left \langle {\sbb_{1}^{\valid}}, \sum_{j=2}^{K} \omega_j \Db_{j} \sbb_{j}^{\train} \right \rangle
    \\
    =\ & \sum_{j=2}^{K} \omega_j \frac{1}{n^{\valid} \cdot n_w} \cdot \left( \frac{n - n^{\train}}{n} \Xb_{1}^{\T}\yb_{1} \right)^{\T} \Db_{j} \left(  \frac{n^{\train}}{n} \Xb_{j}^{\T}\yb_{j} \right)
    \\
    =\ & \sum_{j=2}^{K} \omega_j \frac{n^{\train}}{n^2 \cdot n_w}  \cdot \left( \Xb_{1}^{\T}\Xb_{1} \bbeta_1 \right)^{\T} \Db_{j} \left( \Xb_{j}^{\T} \Xb_{j} \bbeta_{j} \right)
    \\
    =\ & \sum_{j=2}^{K} \omega_j \frac{n^{\train}}{n^2 \cdot n_w}  \cdot \frac{\kappa_{1} \kappa_{j} \sigma_{1,j}^2}{p} \cdot \tr \left( \Xb_{1}^{\T}\Xb_{1} \Db_{j} \Xb_{j}^{\T} \Xb_{j} \right)
    \\
    =\ &  \sum_{j=2}^{K} \omega_j \frac{n^{\train}}{n_w} \cdot \frac{\kappa_{1} \kappa_{j} \sigma_{1,j}^2}{p} \cdot \tr \left( \frac{\Xb_{1}^{\T}\Xb_{1}}{n} \Db_{j} \frac{\Xb_{j}^{\T} \Xb_{j}}{n} \right).
\end{align*}
We further simplify $\Lambda_{\rm sum}^{(4)}$ as follows 
\begin{align*}
    \Lambda_{\rm sum}^{(4)} =\ & \sum_{j=2}^{K} \omega_j \frac{n^{\train}}{n_w} \cdot \frac{\kappa_{1} \kappa_{j} \sigma_{1,j}^2}{p} \cdot \tr \left( \bSigma_{1} \Db_{j} \bSigma_{j} \right). 
\end{align*}
It follows that  
\begin{align*}
    \Lambda_{\rm sum}^{(1)} 
    =\ & \Lambda_{\rm sum}^{(3)} + \Lambda_{\rm sum}^{(4)}
    \\
    =\ & \omega_1 \cdot \frac{n^{\train}}{n_w} \cdot \frac{\kappa_{1} \sigma_{1,1}^2}{p} \cdot \tr\left( \Db_{1} \bSigma_{1}^2 \right) + \sum_{j=2}^{K} \omega_j \frac{n^{\train}}{n_w}  \cdot \frac{\kappa_{1} \kappa_{j} \sigma_{1,j}^2}{p} \cdot \tr \left( \bSigma_{1} \Db_{j} \bSigma_{j} \right).
    % \\
    % =\ & \frac{n^{\train}}{n_w} \cdot \sum_{j=1}^{K} \omega_j \tr \left( \bSigma_{1} \Db_{j} \bSigma_{j} \bPhi_{\bbeta_1 \bbeta_{j}} \right)
\end{align*}
For $\Lambda_{\rm sum}^{(2)}$, we have  
\begin{align*}
    \Lambda_{\rm sum}^{(2)} =\ & \left\| \sum_{j=1}^{K} \omega_j \Ab_{j}(\Wb_{j}^{\T} \Wb_{j}, \Theta_{j}) \sbb_{j}^{\train} \right\|_{\bSigma_{1}}^2
    \\
    =\ & \sum_{1 \leq i, j \leq K} \omega_i \omega_{j} \left[ \frac{n^{\train}}{n} \Xb_{i}^{\T}\yb_{i} + \sqrt{\frac{n^{\train} (n - n^{\train})}{n^2}} (\Cov(\Xb_{i}^{\T} \yb_{i}))^{1/2} \hb_{i} \right]^{\T} 
    \\
    & \qquad \qquad \Ab_{i}(\Wb_{i}^{\T} \Wb_{i}, \Theta_{i}) \bSigma_{1} \Ab_{j}(\Wb_{j}^{\T} \Wb_{j}, \Theta_{j}) \left[ \frac{n^{\train}}{n} \Xb_{j}^{\T}\yb_{j} + \sqrt{\frac{n^{\train} (n - n^{\train})}{n^2}} (\Cov(\Xb_{j}^{\T} \yb_{j}))^{1/2} \hb_{j} \right]
    \\
    =\ & \sum_{i \neq j} \omega_i \omega_{j} \left[ \frac{n^{\train}}{n} \Xb_{i}^{\T}\yb_{i} + \sqrt{\frac{n^{\train} (n - n^{\train})}{n^2}} (\Cov(\Xb_{i}^{\T} \yb_{i}))^{1/2} \hb_{i} \right]^{\T} 
    \\
    & \qquad \qquad \Ab_{i}(\Wb_{i}^{\T} \Wb_{i}, \Theta_{i}) \bSigma_{1} \Ab_{j}(\Wb_{j}^{\T} \Wb_{j}, \Theta_{j}) \left[ \frac{n^{\train}}{n} \Xb_{j}^{\T}\yb_{j} + \sqrt{\frac{n^{\train} (n - n^{\train})}{n^2}} (\Cov(\Xb_{j}^{\T} \yb_{j}))^{1/2} \hb_{j} \right]
    \\
    &+ \sum_{1 \leq j \leq K} \omega_{j}^{2} \left[ \frac{n^{\train}}{n} \Xb_{j}^{\T}\yb_{j} + \sqrt{\frac{n^{\train} (n - n^{\train})}{n^2}} (\Cov(\Xb_{j}^{\T} \yb_{j}))^{1/2} \hb_{j} \right]^{\T} 
    \\
    & \qquad \qquad \Ab_{j}(\Wb_{j}^{\T} \Wb_{j}, \Theta_{j}) \bSigma_{1} \Ab_{j}(\Wb_{j}^{\T} \Wb_{j}, \Theta_{j}) \left[ \frac{n^{\train}}{n} \Xb_{j}^{\T}\yb_{j} + \sqrt{\frac{n^{\train} (n - n^{\train})}{n^2}} (\Cov(\Xb_{j}^{\T} \yb_{j}))^{1/2} \hb_{j} \right].
\end{align*}
By Equation \eqref{eqn:cond-A}, we have
\begin{align*}
    \Lambda_{\rm sum}^{(2)} =\ & \sum_{i \neq j} \frac{\omega_i \omega_{j}}{n_w^2} \left[ \frac{n^{\train}}{n} \Xb_{i}^{\T}\yb_{i} + \sqrt{\frac{n^{\train} (n - n^{\train})}{n^2}} (\Cov(\Xb_{i}^{\T} \yb_{i}))^{1/2} \hb_{i} \right]^{\T}      \\
    & \qquad \qquad \Db_{i} \bSigma_{1} \Db_{j} \left[ \frac{n^{\train}}{n} \Xb_{j}^{\T}\yb_{j} + \sqrt{\frac{n^{\train} (n - n^{\train})}{n^2}} (\Cov(\Xb_{j}^{\T} \yb_{j}))^{1/2} \hb_{j} \right]
    \\
    &+ \sum_{1 \leq j \leq K} \frac{\omega_{j}^{2}}{n_w^2} \left[ \frac{n^{\train}}{n} \Xb_{j}^{\T}\yb_{j} + \sqrt{\frac{n^{\train} (n - n^{\train})}{n^2}} (\Cov(\Xb_{j}^{\T} \yb_{j}))^{1/2} \hb_{j} \right]^{\T} 
    \\
    & \qquad \qquad \Eb_{j} \left[ \frac{n^{\train}}{n} \Xb_{j}^{\T}\yb_{j} + \sqrt{\frac{n^{\train} (n - n^{\train})}{n^2}} (\Cov(\Xb_{j}^{\T} \yb_{j}))^{1/2} \hb_{j} \right]
    \\
    =\ & \Lambda_{\rm sum}^{(5)} + \Lambda_{\rm sum}^{(6)}. 
\end{align*}
By Lemma B.26 in \cite{bai2010spectral} and Lemma 5 in \cite{zhao2024estimating}, for $\Lambda_{\rm sum}^{(5)}$ we have 
\begin{align*}
    \Lambda_{\rm sum}^{(5)} =\ & \sum_{i \neq j}  \omega_i \omega_{j} \left( \frac{n^{\train}}{n_w} \right)^2 \frac{\kappa_{i} \kappa_{j} \sigma_{i,j}^2}{p} \cdot \tr \left( \bSigma_{i} \Db_{i} \bSigma_{1} \Db_{j} \bSigma_{j} \right).
\end{align*}
Following similar steps in Section \ref{sec:proof_thm_reference_panel}, we have
\begin{align*}
    \Lambda_{\rm sum}^{(6)} =\ & \sum_{1 \leq j \leq K} \omega_{j}^{2} \cdot \left[  \left( \frac{n^{\train}}{n_w} \right)^2 \cdot \kappa_{j} \sigma_{j, j}^2 \cdot \frac{p}{n^{\train}} \cdot \frac{1}{p} \tr \left(\bSigma_{j} \right) \cdot \frac{1}{p}\tr \left(\Eb_{j} \bSigma_{j} \right) + \left( \frac{n^{\train}}{n_w} \right)^2 \cdot \frac{\kappa_{j} \sigma_{j,j}^2}{p}\tr \left(\Eb_{j} \bSigma_{j}^2 \right) \right.
    \\
    & \qquad \qquad + \left. \frac{n^{\train}}{n_w^2} \cdot \sigma_{\bepsilon}^2 \cdot \tr\left( \Eb_{j} {\bSigma_{j}} \right) \right] + o_p(1).
\end{align*}
Therefore, we have 
\begin{align*}
    \Lambda_{\rm sum}^{(2)} =\ & \sum_{1 \leq i < j \leq K}  2 \omega_i \omega_{j} \left( \frac{n^{\train}}{n_w} \right)^2 \frac{\kappa_{i} \kappa_{j} \sigma_{i,j}^2}{p} \cdot \tr \left( \bSigma_{i} \Db_{i} \bSigma_{1} \Db_{j} \bSigma_{j} \right)
    \\
    &+ \sum_{1 \leq j \leq K} \omega_{j}^{2} \cdot \left[  \frac{n^{\train}}{n_w^2} \cdot \frac{\kappa_{j} \sigma_{j, j}^2}{p} \cdot \frac{1}{h_{j}^2} \tr \left(\bSigma_{j} \right) \cdot \tr \left(\Eb_{j} \bSigma_{j} \right) + \left( \frac{n^{\train}}{n_w} \right)^2 \cdot \frac{\kappa_{j} \sigma_{j,j}^2}{p}\tr \left(\Eb_{j} \bSigma_{j}^2 \right) \right] + o_p(1).
\end{align*}
It follows that 
\begin{align*}
    R^2_{\rm sum, MA}(\theta)
    = \frac{n^{\valid}}{\|\yb^{\valid}\|_{2}^2} \cdot \frac{\left( \Lambda_{\rm sum}^{(1)} \right)^{2}}{\Lambda_{\rm sum}^{(2)}},
\end{align*}
where
\begin{align*}
    \Lambda_{\rm sum}^{(1)} 
    =\ & \omega_1 \cdot \frac{n^{\train}}{n_w} \cdot \frac{\kappa_{1} \sigma_{\bbeta}^2}{p} \cdot \tr\left( \Db_{1} \bSigma_{1}^2 \right) + \sum_{j=2}^{K} \omega_j \frac{n^{\train}}{n_w}  \cdot \frac{\kappa_{1} \kappa_{j} \sigma_{1,j}^2}{p} \cdot \tr \left( \bSigma_{1} \Db_{j} \bSigma_{j} \right) + o_p(1) \quad \mbox{and}
    \\
    \Lambda_{\rm sum}^{(2)} =\ & \sum_{1 \leq i < j \leq K}  2 \omega_i \omega_{j} \left( \frac{n^{\train}}{n_w} \right)^2 \frac{\kappa_{i} \kappa_{j} \sigma_{i,j}^2}{p} \cdot \tr \left( \bSigma_{i} \Db_{i} \bSigma_{1} \Db_{j} \bSigma_{j} \right)
    \\
    &+ \sum_{1 \leq j \leq K} \omega_{j}^{2} \cdot \left[  \frac{n^{\train}}{n_w^2} \cdot \frac{\kappa_{j} \sigma_{j, j}^2}{p} \cdot \frac{1}{h_{j}^2} \tr \left(\bSigma_{j} \right) \cdot \tr \left(\Eb_{j} \bSigma_{j} \right) + \left( \frac{n^{\train}}{n_w} \right)^2 \cdot \frac{\kappa_{j} \sigma_{j,j}^2}{p} \tr \left(\Eb_{j} \bSigma_{j}^2 \right) \right] + o_p(1).
\end{align*}

\paragraph{\texorpdfstring{Part II: The limit of $\mathbf{R^2_{\rm ind, MA}(\theta)}$.}{Individual}}
Recall that 
\begin{align*}
    R^2_{\rm ind, MA}(\theta) = \frac{n^{\valid}}{\|\yb^{\valid}\|_{2}^2} \cdot \frac{\left\langle {\Xb_{1}^{\valid}{}^{\T}\yb_{1}^{\valid}}, \sum_{j=1}^{K} \omega_j \Ab_{j}(\Wb_{j}^{\T} \Wb_{j}, \Theta_{j}) \Xb_{j}^{\train}{}^{\T} \yb_{j}^{\train} \right\rangle^2}{n^{\valid}{}^{2} \cdot \| \sum_{j=1}^{K} \omega_j \Ab_{j}(\Wb_{j}^{\T} \Wb_{j}, \Theta_{j}) \Xb_{j}^{\train}{}^{\T} \yb_{j}^{\train} \|_{\bSigma_{1}}^2} = \frac{n^{\valid}}{\|\yb^{\valid}\|_{2}^2} \cdot \frac{\left( \Lambda_{\rm ind}^{(1)} \right)^{2}}{\Lambda_{\rm ind}^{(2)}},
\end{align*}
where $\Lambda_{\rm ind}^{(1)}$ and $\Lambda_{\rm ind}^{(2)}$ are defined as
\begin{align*}
    \Lambda_{\rm ind}^{(1)} =\ & \frac{1}{n^{\valid}} \cdot \left\langle {\Xb_{1}^{\valid}{}^{\T}\yb_{1}^{\valid}}, \sum_{j=1}^{K} \omega_j \Ab_{j}(\Wb_{j}^{\T} \Wb_{j}, \Theta_{j}) \Xb_{j}^{\train}{}^{\T} \yb_{j}^{\train} \right\rangle \quad \mbox{and}
    \\
    \Lambda_{\rm ind}^{(2)} =\ & \left\| \sum_{j=1}^{K} \omega_j \Ab_{j}(\Wb_{j}^{\T} \Wb_{j}, \Theta_{j}) \Xb_{j}^{\train}{}^{\T} \yb_{j}^{\train} \right\|_{\bSigma_{1}}^2. 
\end{align*}
Consider $\Lambda_{\rm ind}^{(1)}$ first, following similar steps  as above, we have
\begin{align*}
    \Lambda_{\rm ind}^{(1)}
    =\ & \frac{1}{n^{\valid} \cdot n_w} \sum_{j=1}^{K} \omega_j  \left(\Xb_{1}^{\valid}{}^{\T} \yb_{1}^{\valid} \right)^{\T} \Db_{j} \Xb_{j}^{\train}{}^{\T} \yb_{j}^{\train}
    \\
    =\ & \omega_1 \cdot \frac{n^{\train}}{n_w} \cdot \frac{\kappa_{1} \sigma_{\bbeta}^2}{p} \cdot \tr\left( \Db_{1} \bSigma_{1}^2 \right) + \sum_{j=2}^{K} \omega_j \frac{n^{\train}}{n_w}  \cdot \frac{\kappa_{1} \kappa_{j} \sigma_{1,j}^2}{p} \cdot \tr \left( \bSigma_{1} \Db_{j} \bSigma_{j} \right) + o_p(1). 
\end{align*}
For $\Lambda_{\rm ind}^{(2)}$, we have  
\begin{align*}
    \Lambda_{\rm ind}^{(2)} =\ &  \left\| \sum_{j=1}^{K} \omega_j \Ab_{j}(\Wb_{j}^{\T} \Wb_{j}, \Theta_{j}) \Xb_{j}^{\train}{}^{\T} \yb_{j}^{\train} \right\|_{\bSigma_{1}}^2
    \\
    =\ & \sum_{1 \leq i, j \leq K} \omega_{i} \omega_j \left( \Xb_{i}^{\train}{}^{\T} \yb_{i}^{\train} \right)^{\T} \Ab_{i}(\Wb_{i}^{\T} \Wb_{i}, \Theta_{i}) \bSigma_{1} \Ab_{j}(\Wb_{j}^{\T} \Wb_{j}, \Theta_{j}) \Xb_{j}^{\train}{}^{\T} \yb_{j}^{\train}
    \\
    =\ & \sum_{i \neq j} \omega_{i} \omega_j \left( \Xb_{i}^{\train}{}^{\T} \yb_{i}^{\train} \right)^{\T} \Ab_{i}(\Wb_{i}^{\T} \Wb_{i}, \Theta_{i}) \bSigma_{1} \Ab_{j}(\Wb_{j}^{\T} \Wb_{j}, \Theta_{j}) \Xb_{j}^{\train}{}^{\T} \yb_{j}^{\train}
    \\
    &+ \sum_{1 \leq j \leq K} \omega_j^2 \left( \Xb_{j}^{\train}{}^{\T} \yb_{j}^{\train} \right)^{\T} \Ab_{j}(\Wb_{j}^{\T} \Wb_{j}, \Theta_{j}) \bSigma_{1} \Ab_{j}(\Wb_{j}^{\T} \Wb_{j}, \Theta_{j}) \Xb_{j}^{\train}{}^{\T} \yb_{j}^{\train}.
\end{align*}
By Equation \eqref{eqn:cond-A}, for $\Lambda_{\rm ind}^{(2)}$ we have 
\begin{align*}
    \Lambda_{\rm ind}^{(2)} =\ & \sum_{i \neq j} \frac{\omega_{i} \omega_j}{n_w^2} \left( \Xb_{i}^{\train}{}^{\T} \yb_{i}^{\train} \right)^{\T} \Db_{i} \bSigma_{1} \Db_{j} \Xb_{j}^{\train}{}^{\T} \yb_{j}^{\train}
    + \sum_{1 \leq j \leq K} \frac{\omega_j^2}{n_w^2} \left( \Xb_{j}^{\train}{}^{\T} \yb_{j}^{\train} \right)^{\T} \Eb_{j} \Xb_{j}^{\train}{}^{\T} \yb_{j}^{\train} + o_p(1)
    \\
    =\ & \Lambda_{\rm ind}^{(3)} + \Lambda_{\rm ind}^{(4)} + o_p(1).
\end{align*}
By Lemma B.26 in \cite{bai2010spectral} and Lemma 5 in \cite{zhao2024estimating}, for $\Lambda_{\rm ind}^{(3)}$ we have  
\begin{align*}
    \Lambda_{\rm ind}^{(3)} =\ & \sum_{1 \leq i < j \leq K}  2 \omega_i \omega_{j} \left( \frac{n^{\train}}{n_w} \right)^2 \frac{\kappa_{i} \kappa_{j} \sigma_{i,j}^2}{p} \cdot \tr \left( \bSigma_{i} \Db_{i} \bSigma_{1} \Db_{j} \bSigma_{j} \right).
\end{align*}
Applying similar steps as in Section \ref{sec:proof_thm_reference_panel}, we have
\begin{align*}
    & \frac{1}{n_w^2} \left( \Xb_{j}^{\train}{}^{\T} \yb_{j}^{\train} \right)^{\T} \Eb_{j} \Xb_{j}^{\train}{}^{\T} \yb_{j}^{\train}
    \\
    =\ & \left( \frac{n^{\train}}{n_w} \right)^2 \cdot \kappa_{j} \sigma_{j,j}^2 \cdot \frac{p}{n^{\train}} \cdot \frac{1}{p} \tr \left(\bSigma_{j} \right) \cdot \frac{1}{p}\tr \left(\Eb_{j} \bSigma_{j} \right) + \left( \frac{n^{\train}}{n_w} \right)^2 \cdot \frac{\kappa_{j} \sigma_{j,j}^2}{p}\tr \left(\Eb_{j} \bSigma_{j}^2 \right)
    \\
    & \qquad \qquad + \frac{n^{\train}}{n_w^2} \cdot \sigma_{\bepsilon}^2 \cdot \tr\left( \Eb_{j} {\bSigma_{j}} \right).
\end{align*}
Therefore, we have  
\begin{align*}
    \Lambda_{\rm ind}^{(2)} =\ & \sum_{1 \leq i < j \leq K}  2 \omega_i \omega_{j} \left( \frac{n^{\train}}{n_w} \right)^2 \frac{\kappa_{i} \kappa_{j} \sigma_{i,j}^2}{p} \cdot \tr \left( \bSigma_{i} \Db_{i} \bSigma_{1} \Db_{j} \bSigma_{j} \right)
    \\
    &+ \sum_{1 \leq j \leq K} \omega_{j}^{2} \cdot \left[  \frac{n^{\train}}{n_w^2} \cdot \frac{\kappa_{j} \sigma_{j, j}^2}{p} \cdot \frac{1}{h_{j}^2} \tr \left(\bSigma_{j} \right) \cdot \tr \left(\Eb_{j} \bSigma_{j} \right) + \left( \frac{n^{\train}}{n_w} \right)^2 \cdot \frac{\kappa_{j} \sigma_{j,j}^2}{p}\tr \left(\Eb_{j} \bSigma_{j}^2 \right) \right] + o_p(1).
\end{align*}
It follows that 
\begin{align*}
    R^2_{\rm ind, MA}(\theta)
    = \frac{n^{\valid}}{\|\yb^{\valid}\|_{2}^2} \cdot \frac{\left( \Lambda_{\rm ind}^{(1)} \right)^{2}}{\Lambda_{\rm ind}^{(2)}},
\end{align*}
where
\begin{align*}
    \Lambda_{\rm ind}^{(1)} 
    =\ & \omega_1 \cdot \frac{n^{\train}}{n_w} \cdot \frac{\kappa_{1} \sigma_{\bbeta}^2}{p} \cdot \tr\left( \Db_{1} \bSigma_{1}^2 \right) + \sum_{j=2}^{K} \omega_j \frac{n^{\train}}{n_w}  \cdot \frac{\kappa_{1} \kappa_{j} \sigma_{1,j}^2}{p} \cdot \tr \left( \bSigma_{1} \Db_{j} \bSigma_{j} \right) + o_p(1) \quad \mbox{and}
    \\
    \Lambda_{\rm ind}^{(2)} =\ & \sum_{1 \leq i < j \leq K}  2 \omega_i \omega_{j} \left( \frac{n^{\train}}{n_w} \right)^2 \frac{\kappa_{i} \kappa_{j} \sigma_{i,j}^2}{p} \cdot \tr \left( \bSigma_{i} \Db_{i} \bSigma_{1} \Db_{j} \bSigma_{j} \right)
    \\
    &+ \sum_{1 \leq j \leq K} \omega_{j}^{2} \cdot \left[  \frac{n^{\train}}{n_w^2} \cdot \frac{\kappa_{j} \sigma_{j, j}^2}{p} \cdot \frac{1}{h_{j}^2} \tr \left(\bSigma_{j} \right) \cdot \tr \left(\Eb_{j} \bSigma_{j} \right) + \left( \frac{n^{\train}}{n_w} \right)^2 \cdot \frac{\kappa_{j} \sigma_{j,j}^2}{p}\tr \left(\Eb_{j} \bSigma_{j}^2 \right) \right] + o_p(1).
\end{align*}

\paragraph{Proof of Corollary \ref{cor:optimal_weight}.}

When $K = 2$, we have
\begin{align*}
    R^2_{\rm sum, MA}(\theta)
    = \frac{n^{\valid}}{\|\yb^{\valid}\|_{2}^2} \cdot \frac{ \left[ \omega_1 \cdot N_1(\Theta) + (1 - \omega_1) N_2(\Theta) \right]^2}{\omega_1^2 \cdot D_1(\Theta) + (1 - \omega_1)^2 \cdot D_2(\Theta) + 2 \omega_1 (1 - \omega_1) \cdot D_3(\Theta)}, 
\end{align*}
where
\begin{align*}
    N_1(\Theta) =\ & \frac{\kappa_{1} \sigma_{\bbeta}^2}{p} \cdot \tr\left( \Db_{1} \bSigma_{1}^2 \right),
    \\
    N_2(\Theta) =\ & \frac{\kappa_{1} \kappa_{2} \sigma_{1,2}^2}{p} \cdot \tr \left( \bSigma_{1} \Db_{2} \bSigma_{2} \right), 
    \\
    D_1(\Theta) =\ & \frac{1}{n^{\train}} \cdot \frac{\kappa_{1} \sigma_{1, 1}^2}{p} \cdot \frac{1}{h_{1}^2} \tr \left(\bSigma_{1} \right) \cdot \tr \left(\Eb_{1} \bSigma_{1} \right) + \frac{\kappa_{1} \sigma_{1,1}^2}{p}\tr \left(\Eb_{1} \bSigma_{1}^2 \right), 
    \\
    D_2(\Theta) =\ & \frac{1}{n^{\train}} \cdot \frac{\kappa_{2} \sigma_{2, 2}^2}{p} \cdot \frac{1}{h_{2}^2} \tr \left(\bSigma_{2} \right) \cdot \tr \left(\Eb_{2} \bSigma_{2} \right) + \frac{\kappa_{2} \sigma_{2,2}^2}{p}\tr \left(\Eb_{2} \bSigma_{2}^2 \right), \quad \mbox{and}
    \\
    D_3(\Theta) =\ & \frac{\kappa_{1} \kappa_{2} \sigma_{1,2}^2}{p} \cdot \tr \left( \bSigma_{1} \Db_{1} \bSigma_{1} \Db_{2} \bSigma_{2} \right).
\end{align*}
We abbreviate $N_1(\Theta), N_2(\Theta), D_1(\Theta), D_2(\Theta), D_3(\Theta)$ by $N_1, N_2, D_1, D_2, D_3$, respectively. 
Taking derivative with respect to $\omega_1$, we have
\begin{align*}
    \frac{\partial R^2_{\rm sum, MA}(\theta)}{\partial \omega_1} = -\frac{2 \left[N_1 \omega_1 + N_2 (1 - \omega_1)\right] \cdot 
    \left[- D_2 N_1 (1 - \omega_1)+D_1 N_2 \omega_1 +D_3
    \left(- N_1 \omega_1 - N_2 \omega_1
    +N_2\right)\right]}{\left[D_1 \omega_1 ^2 + (1 - \omega_1)
    \left(D_2 (1 - \omega_1) + 2 D_3 \omega_1
    \right)\right]^2}.
\end{align*}
The denominator above is always positive and $N_1 \omega_1 + N_2 (1 - \omega_1) > 0$ as $\omega_1 \in [0,1]$. 
Therefore, the only solution to ${\partial R^2_{\rm sum}(\theta)}/{\partial \omega_1} = 0$ is given by
\begin{align*}
    \omega_1 = \frac{D_2 N_1-D_3 N_2}{D_2 N_1-D_3 N_1+D_1 N_2-D_3 N_2}. 
\end{align*}
Consider the second derivative at $\omega_1$, we have
\begin{align*}
    \frac{\partial R^2_{\rm sum, MA}(\theta)}{\partial \omega_1}& \Bigg|_{\omega_1 = \frac{D_2 N_1-D_3 N_2}{D_2 N_1-D_3 N_1+D_1 N_2-D_3 N_2}} 
    \\
    =\ & -\frac{2 \left[D_2 N_1+D_1 N_2-D_3
   \left(N_1+N_2\right)\right]^4}{\left(D_3^2-D_1
   D_2\right)^2 \left(D_2 N_1^2+ D_1 N_2^2 -2
   D_3 N_1 N_2 \right)} < 0,
\end{align*}
where the last inequality follows from the fact that $D_2 N_1^2+ D_1 N_2^2 - 2 D_3 N_1 N_2 > 0$.
Therefore, we conclude that the optimal $R^2_{\rm sum, MA}(\theta)$ is obtained when
\begin{align*}
\begin{split}
    \omega_1 =\ & \min \left\{1, \frac{D_2(\Theta) N_1(\Theta)-D_3(\Theta) N_2(\Theta)}{D_2(\Theta) N_1(\Theta) - D_3(\Theta) N_1(\Theta) +D_1(\Theta) N_2(\Theta) - D_3(\Theta) N_2(\Theta)} \right\} 
    \quad \mbox{and} \quad
    \\
    \omega_2 =\ & \max  \left\{0, 1 - \frac{D_2(\Theta) N_1(\Theta) -D_3(\Theta) N_2(\Theta)}{D_2(\Theta) N_1(\Theta) -D_3(\Theta) N_1(\Theta)+D_1(\Theta) N_2(\Theta)-D_3(\Theta) N_2(\Theta)} \right\}.
\end{split}
\end{align*}

\newpage

{
\scriptsize
\setlength{\tabcolsep}{8pt}
\renewcommand{\arraystretch}{1.1}
\begin{longtable}{lccccccc}
    \caption{ {\bf Out-of-sample $R^2$ estimates from various methods across 71 DXA imaging traits.} The heritability is estimated by using LDSC (\url{https://github.com/bulik/ldsc}), and more information on these imaging traits can be found at \url{https://biobank.ndph.ox.ac.uk/ukb/label.cgi?id=124}. The third to fifth columns report out-of-sample $R^2$ for resampling-based self-training methods (Lassosum2-pseudo, Ensemble-pseudo, and LDpred2-pseudo), while the last three columns present out-of-sample $R^2$ for individual-level data training of LDpred2 with varying validation sample sizes ($n^{(v)} = 1000$, $500$, and $100$). The final three rows summarize the mean, median, and standard deviation (Std.) of heritability and out-of-sample $R^2$ values across all DXA traits.} \label{tab:h2_R2} \\
    \toprule
    Trait ID & Heritability & \begin{tabular}{@{}c@{}} Lassosum2- \\ pseudo \end{tabular} & \begin{tabular}{@{}c@{}} Ensemble- \\ pseudo \end{tabular} & \begin{tabular}{@{}c@{}} LDpred2- \\ pseudo \end{tabular} & \begin{tabular}{@{}c@{}} LDpred2 \\ $(n^{(v)} = 1000)$ \end{tabular} & \begin{tabular}{@{}c@{}} LDpred2 \\ $(n^{(v)} = 500)$ \end{tabular} & \begin{tabular}{@{}c@{}} LDpred2 \\ $(n^{(v)} = 100)$ \end{tabular}   \\
    \midrule
    \endfirsthead

    % Header repeated on new pages
    \toprule
    Trait ID & Heritability & \begin{tabular}{@{}c@{}} Lassosum2- \\ pseudo \end{tabular} & \begin{tabular}{@{}c@{}} Ensemble- \\ pseudo \end{tabular} & \begin{tabular}{@{}c@{}} LDpred2- \\ pseudo \end{tabular} & \begin{tabular}{@{}c@{}} LDpred2 \\ $(n^{(v)} = 1000)$ \end{tabular} & \begin{tabular}{@{}c@{}} LDpred2 \\ $(n^{(v)} = 500)$ \end{tabular} & \begin{tabular}{@{}c@{}} LDpred2 \\ $(n^{(v)} = 100)$ \end{tabular}   \\
    \midrule
    \endhead

    % Footer repeated on new pages
    \midrule
    \multicolumn{8}{r}{\textit{Continued on next page...}} \\
    \midrule
    \endfoot

    % Footer at the end of the table
    \bottomrule
    \endlastfoot

        21110 & 0.2067 & 0.0016 & 0.0025 & 0.0035 & 0.0015 & 0.0013 & 0.0013 \\
        21111 & 0.1028 & 0.0007 & 0.0008 & 0.0010 & 0.0007 & 0.0005 & 0.0004\\
        21112 & 0.3388 & 0.0063 & 0.0092 & 0.0089 & 0.0050 & 0.0034 & 0.0033\\
        21116 & 0.3441 & 0.0053 & 0.0109 & 0.0087 & 0.0091 & 0.0051 & 0.0038\\
        21113 & 0.1412 & 0.0066 & 0.0079 & 0.0060 & 0.0026 & 0.0007 & 0.0011\\
        21117 & 0.1526 & 0.0035 & 0.0055 & 0.0059 & 0.0068 & 0.0022 & 0.0018\\
        21114 & 0.0074 & 0.0035 & 0.0044 & 0.0044 & 0.0040 & 0.0007 & 0.0002\\
        21118 & 0.0216 & 0.0040 & 0.0041 & 0.0037 & 0.0048 & 0.0009 & 0.0002\\
        21119 & 0.3467 & 0.0096 & 0.0109 & 0.0092 & 0.0054 & 0.0053 & 0.0048\\
        21120 & 0.1857 & 0.0034 & 0.0021 & 0.0020 & 0.0010 & 0.0013 & 0.0016\\
        21121 & 0.0959 & 0.0014 & 0.0013 & 0.0011 & 0.0008 & 0.0006 & 0.0005\\
        21123 & 0.2649 & 0.0024 & 0.0010 & 0.0032 & 0.0018 & 0.0021 & 0.0022\\
        21124 & 0.1249 & 0.0004 & 0.0006 & 0.0004 & 0.0007 & 0.0004 & 0.0003\\
        21125 & 0.2688 & 0.0082 & 0.0123 & 0.0105 & 0.0116 & 0.0097 & 0.0075\\
        21128 & 0.2599 & 0.0102 & 0.0132 & 0.0125 & 0.0104 & 0.0068 & 0.0055\\
        21126 & 0.2564 & 0.0036 & 0.0044 & 0.0040 & 0.0027 & 0.0010 & 0.0009\\
        21129 & 0.2320 & 0.0041 & 0.0037 & 0.0036 & 0.0045 & 0.0014 & 0.0012\\
        21127 & 0.1820 & 0.0034 & 0.0035 & 0.0034 & 0.0048 & 0.0013 & 0.0004\\
        21130 & 0.1896 & 0.0041 & 0.0031 & 0.0027 & 0.0039 & 0.0013 & 0.0006\\
        21131 & 0.2934 & 0.0102 & 0.0115 & 0.0124 & 0.0121 & 0.0098 & 0.0074\\
        21132 & 0.2463 & 0.0013 & 0.0008 & 0.0016 & 0.0012 & 0.0007 & 0.0005\\
        21133 & 0.1647 & 0.0025 & 0.0013 & 0.0010 & 0.0005 & 0.0005 & 0.0004\\
        21122 & 0.3780 & 0.0128 & 0.0213 & 0.0178 & 0.0179 & 0.0146 & 0.0099\\
        21134 & 0.3471 & 0.0099 & 0.0187 & 0.0145 & 0.0149 & 0.0129 & 0.0082\\
        21135 & 0.0778 & 0.0011 & 0.0005 & 0.0026 & 0.0013 & 0.0012 & 0.0013\\
        23244 & 0.3277 & 0.0153 & 0.0316 & 0.0277 & 0.0271 & 0.0252 & 0.0166\\
        23245 & 0.1669 & 0.0006 & 0.0003 & 0.0004 & 0.0005 & 0.0004 & 0.0003\\
        23246 & 0.1923 & 0.0005 & 0.0005 & 0.0011 & 0.0007 & 0.0004 & 0.0003\\
        23247 & 0.1946 & 0.0019 & 0.0026 & 0.0059 & 0.0055 & 0.0039 & 0.0025\\
        23248 & 0.1167 & 0.0021 & 0.0005 & 0.0003 & 0.0004 & 0.0003 & 0.0003\\
        23249 & 0.1160 & 0.0021 & 0.0017 & 0.0017 & 0.0014 & 0.0006 & 0.0004\\
        23253 & 0.1386 & 0.0025 & 0.0016 & 0.0018 & 0.0021 & 0.0008 & 0.0003\\
        23250 & 0.1374 & 0.0037 & 0.0052 & 0.0051 & 0.0026 & 0.0016 & 0.0017\\
        23254 & 0.1422 & 0.0034 & 0.0050 & 0.0045 & 0.0052 & 0.0028 & 0.0022\\
        23251 & 0.2300 & 0.0012 & 0.0020 & 0.0011 & 0.0016 & 0.0005 & 0.0001\\
        23255 & 0.2164 & 0.0016 & 0.0036 & 0.0020 & 0.0014 & 0.0008 & 0.0006\\
        23252 & 0.0389 & 0.0015 & 0.0036 & 0.0014 & 0.0009 & 0.0003 & 0.0003\\
        23256 & 0.0544 & 0.0023 & 0.0022 & 0.0011 & 0.0012 & 0.0007 & 0.0006\\
        23257 & 0.1606 & 0.0007 & 0.0010 & 0.0004 & 0.0007 & 0.0004 & 0.0003\\
        23258 & 0.1855 & 0.0028 & 0.0036 & 0.0036 & 0.0028 & 0.0027 & 0.0025\\
        23259 & 0.2261 & 0.0003 & 0.0025 & 0.0020 & 0.0012 & 0.0009 & 0.0008\\
        23260 & 0.0967 & 0.0013 & 0.0033 & 0.0016 & 0.0006 & 0.0006 & 0.0006\\
        23261 & 0.3114 & 0.0076 & 0.0200 & 0.0142 & 0.0124 & 0.0101 & 0.0088\\
        23262 & 0.1745 & 0.0010 & 0.0025 & 0.0009 & 0.0007 & 0.0007 & 0.0007\\
        23263 & 0.2528 & 0.0006 & 0.0020 & 0.0023 & 0.0014 & 0.0015 & 0.0017\\
        23264 & 0.2321 & 0.0017 & 0.0026 & 0.0026 & 0.0039 & 0.0030 & 0.0023\\
        23265 & 0.1194 & 0.0005 & 0.0008 & 0.0003 & 0.0006 & 0.0004 & 0.0004\\
        23266 & 0.2121 & 0.0017 & 0.0038 & 0.0022 & 0.0019 & 0.0009 & 0.0005\\
        23270 & 0.2173 & 0.0031 & 0.0025 & 0.0018 & 0.0024 & 0.0010 & 0.0004\\
        23267 & 0.2302 & 0.0020 & 0.0015 & 0.0021 & 0.0017 & 0.0014 & 0.0012\\
        23271 & 0.2140 & 0.0026 & 0.0014 & 0.0016 & 0.0023 & 0.0012 & 0.0010\\
        23268 & 0.2492 & 0.0017 & 0.0037 & 0.0023 & 0.0011 & 0.0010 & 0.0010\\
        23272 & 0.2531 & 0.0021 & 0.0051 & 0.0026 & 0.0013 & 0.0009 & 0.0009\\
        23269 & 0.1161 & 0.0013 & 0.0039 & 0.0009 & 0.0012 & 0.0004 & 0.0002\\
        23273 & 0.1160 & 0.0012 & 0.0016 & 0.0014 & 0.0008 & 0.0004 & 0.0004\\
        23274 & 0.1912 & 0.0010 & 0.0029 & 0.0014 & 0.0011 & 0.0008 & 0.0007\\
        23275 & 0.2399 & 0.0011 & 0.0004 & 0.0017 & 0.0010 & 0.0010 & 0.0008\\
        23276 & 0.2550 & 0.0024 & 0.0035 & 0.0035 & 0.0043 & 0.0037 & 0.0030\\
        23277 & 0.1448 & 0.0002 & 0.0011 & 0.0003 & 0.0002 & 0.0002 & 0.0001\\
        23278 & 0.1570 & 0.0012 & 0.0003 & 0.0003 & 0.0006 & 0.0004 & 0.0002\\
        23279 & 0.2498 & 0.0005 & 0.0013 & 0.0028 & 0.0013 & 0.0010 & 0.0008\\
        23280 & 0.2461 & 0.0003 & 0.0014 & 0.0026 & 0.0012 & 0.0009 & 0.0007\\
        23281 & 0.2178 & 0.0004 & 0.0017 & 0.0031 & 0.0047 & 0.0030 & 0.0017\\
        23282 & 0.0052 & 0.0008 & 0.0007 & 0.0007 & 0.0006 & 0.0003 & 0.0002\\
        23283 & 0.0163 & 0.0011 & 0.0011 & 0.0006 & 0.0003 & 0.0002 & 0.0003\\
        23284 & 0.1647 & 0.0005 & 0.0003 & 0.0011 & 0.0008 & 0.0006 & 0.0006\\
        23285 & 0.2425 & 0.0005 & 0.0022 & 0.0034 & 0.0022 & 0.0015 & 0.0014\\
        23286 & 0.2012 & 0.0007 & 0.0042 & 0.0046 & 0.0062 & 0.0040 & 0.0029\\
        23287 & 0.0818 & 0.0007 & 0.0003 & 0.0004 & 0.0007 & 0.0006 & 0.0005\\
        23288 & 0.1822 & 0.0005 & 0.0059 & 0.0021 & 0.0013 & 0.0010 & 0.0008\\
        23289 & 0.1822 & 0.0005 & 0.0012 & 0.0020 & 0.0014 & 0.0010 & 0.0008\\
        \hline \hline
        \text{Mean} & 0.1894 & 0.0029 & 0.0043 & 0.0038 & 0.0035 & 0.0024 & 0.0018 \\
        \text{Median} & 0.1912 & 0.0017 & 0.0025 & 0.0022 & 0.0014 & 0.0010 & 0.0008\\
        \text{Std.} & 0.0843 & 0.0031 & 0.0055 & 0.0047 & 0.0046 & 0.0040 & 0.0028\\
\end{longtable}
}

\end{document}